\documentclass[10pt]{article}

\usepackage[margin = 1in]{geometry}  
\usepackage[english]{babel}  
\usepackage[final]{showkeys}  

\usepackage{graphicx}  
\graphicspath{{./images/}}

\usepackage{tikz, pgfplots}  
\usepgfplotslibrary{fillbetween}

\usepackage[small, labelsep = endash, labelfont = bf, up, up]{caption}  
\usepackage{subcaption}
\usepackage{tabularray}  
\usepackage{titlesec}  
\usepackage{amsfonts, amsmath, amssymb, amsthm, mathtools}  
\usepackage{amsfonts}  
\usepackage{bm}  
\usepackage[mathscr]{euscript}  
\usepackage{cancel}  
\usepackage{xcolor}  
\definecolor{ao(english)}{rgb}{0.0, 0.5, 0.0}
\usepackage[bookmarks = false, colorlinks = true, allcolors = ao(english)]{hyperref}  
\usepackage{soul}  
\usepackage{comment}  

\newtheorem{theorem}{Theorem}[section]  
\newtheorem{proposition}[theorem]{Proposition}  
\newtheorem{definition}[theorem]{Definition}  
\newtheorem{remark}[theorem]{Remark}  
\newtheorem{example}[theorem]{Example}  

\newcommand{\Z}{\mathbb{Z}}  
\newcommand{\R}{\mathbb{R}}  
\newcommand{\C}{\mathbb{C}}  

\newcommand{\smallerrelaux}[2]{\vcenter{\hbox{{\scalebox{.75}{$#1#2$}}}}}  
\newcommand{\smallerrel}[1]{\mathrel{\mathpalette\smallerrelaux{#1}}}
\newcommand{\smallin}{\smallerrel{\in}}

\titleformat{\section}[block]{\bf\filcenter}{\thesection.}{5pt}{}  
\titleformat{\subsection}[block]{\bf}{\thesubsection.}{5pt}{}  
\titleformat{\subsubsection}[block]{\bf}{\thesubsubsection.}{5pt}{}  

\makeatletter  
\g@addto@macro\bfseries{\boldmath}
\makeatother

\addto\captionsenglish{}  

\newcommand{\sgn}{sgn}

\pgfplotsset{compat=1.18}

\title{EDGE STATES IN SQUARE LATTICE MEDIA \\ AND THEIR DEFORMATIONS}
\author{J. CHABAN\thanks{Department of Applied Physics and Applied Mathematics, Columbia University} \hspace{20pt} J. L. MARZUOLA\thanks{Department of Mathematics, University of North Carolina at Chapel Hill} \hspace{20pt} M. I. WEINSTEIN\thanks{Department of Applied Physics and Applied Mathematics and Department of Mathematics, Columbia University}}
\date{\today}

\makeatletter  
\def\@maketitle{
    \begin{center}
    {\large \bf \@title} \\
    \vspace{\baselineskip}
    \let\footnote\thanks
    {\@author} \\
    \vspace{1.5\baselineskip}
    {\@date}
    \end{center}}
\makeatother

\begin{document}


\newpage
\maketitle

\begin{abstract}
\noindent Edge states are time-harmonic solutions of conservative wave systems which are plane wave-like parallel to and localized transverse to an interface between two bulk media. We study a class of 2D edge Hamiltonians modeling a medium which slowly interpolates between periodic bulk media via a domain wall across a ``rational'' line defect. We consider the cases of (1) periodic bulk media having the symmetries of a square lattice, and (2) linear deformations of such media. Our bulk Hamiltonians break time-reversal symmetry due to perturbation by a magnetic term, which opens a band gap about the band structure degeneracies of the unperturbed bulk Hamiltonian. In case (1), these are quadratic band degeneracies; in case (2), they are pairs of conical degeneracies. We demonstrate that this band gap is traversed by two distinct edge state curves, consistent with the bulk-edge correspondence principle of topological physics. Blow-ups of these curves near the bulk band degeneracies are described by effective (homogenized) edge Hamiltonians derived via multiple-scale analysis which control the bifurcation of edge states. In case \nolinebreak (1), the bifurcation is governed by a matrix Schr\"{o}dinger operator; in case (2), it is governed by a pair of Dirac operators. We present analytical results and numerical simulations for both the full 2D edge Hamiltonian spectral problem and the spectra of effective edge Hamiltonians.
\end{abstract}

\section{Introduction}
\label{sec:intro}

An {\it edge state} is a time-harmonic solution of a conservative wave system which propagates plane wave-like parallel to and is localized transverse to a line defect, or edge, separating distinct bulk media. Such states are central to topologically protected energy transport phenomena in condensed matter physics and other fields. Examples include the integer quantum Hall effect \cite{TKNN:82}, topologically protected states in graphene and related materials in the presence of magnetic fields \cite{geim2007rise}, topological insulators \cite{Kane-Mele:05, kane2005quantum}, as well as engineered topological metamaterials in photonics \cite{top-photonics19}, mechanics \cite{PCV:15}, and other areas. This article draws motivation from the seminal and influential work of Haldane and Raghu in topological photonics \cite{HR08, RH:08} and its first experimental observation in Wang et al. \cite{Soljacic-etal:08}.

We study edge states arising from two classes of 2D bulk media with line defects: (1) periodic media with the symmetries of a square lattice, or {\it square lattice media}, and (2) linearly deformed square lattice media. Each is modeled by a {\it bulk Hamiltonian}: a self-adjoint second-order elliptic operator $H_{\rm bulk}$ on $L^2(\R^2)$ with coefficients periodic with respect to a lattice $\Lambda$; specifically, we take $H_{\rm bulk}$ to be a Schr\"{o}dinger operator with a real-valued, $\Lambda$-periodic potential. The wave transport properties of a periodic operator are encoded in its {\it band structure}: the collection of dispersion relations (the graphs of which are dispersion surfaces, or bands) \linebreak and corresponding spatially extended, plane wave-like Floquet-Bloch eigenstates; see Section \ref{sec:fb-thy-bulk}. The spectrum of $H_{\rm bulk}$ is continuous; under its time-evolution, a localized initial condition spreads out and decays in amplitude. Further, our $H_{\rm bulk}$ is invariant under spatial inversion, or parity $\mathcal{P}$ and complex conjugation, or time-reversal $\mathcal{C}$, which may lead to band structure degeneracies: the energy-quasimomentum pairs ${(E_\star, {\bm k}_\star)}$ at which two consecutive dispersion surfaces touch. For the (undeformed) square lattice, the band structure of $H_{\rm bulk}$ contains {\it quadratic band degeneracies} \cite{keller2018spectral, keller2020erratum, chong2008effective} which, under linear deformation, split into pairs of conical degeneracies, or {\it Dirac points} \cite{chaban2024instability}; see also \cite{fefferman2012honeycomb, drouot2021ubiq, DL:24}. These degeneracies are ``symmetry-protected'' in the sense that they persist under small perturbations of $H_{\rm bulk}$ which commute with $\mathcal{PC}$.

We model a line defect, or edge, via a spatially non-compact perturbation of $H_{\rm bulk}$ which retains translation invariance in the direction of a prescribed vector ${\boldsymbol{\mathfrak{v}}_{\rm edge} \smallin \Lambda}$; such edge directions are called ``rational''. The construction proceeds as follows: First, we consider two ``insulating'' periodic media which are perturbations of the bulk medium by a magnetic term. These are modeled by the Hamiltonians
\begin{equation}
\label{eq:asym-op_r}
H^{\pm, \, \delta}_{\rm bulk} = H_{\rm bulk} \pm \delta^r W_\infty .
\end{equation}
The operator $W_\infty$ is self-adjoint, $\Lambda$-periodic, and commutes with $\mathcal{P}$, but anticommutes with $\mathcal{C}$. The scaling parameter ${r = 1}$, $2$ depends on the local character of the dispersion surfaces of $H_{\rm bulk}$ about the relevant degeneracy: ${r = 2}$ for a quadratic band degeneracy, ${r = 1}$ for conical degeneracies.

\begin{figure}[!t]
\centering
\includegraphics[scale = 1.0]{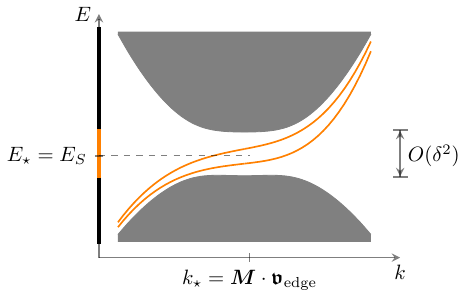}
\caption{(Edge state diagram: square lattice case.) Schematic of the spectrum of the edge Hamiltonian $H^\delta_{\rm edge}$ \eqref{eq:edge-op_r}, as the union of its $L^2_k(\R^2/\Z\boldsymbol{\mathfrak{v}}_{\rm edge})$ \eqref{eq:def-L2kpar} spectra, near ${(E_\star, k_\star) = (E_S, \, {\bm M} \cdot \boldsymbol{\mathfrak{v}}_{\rm edge})}$, where ${(E_\star, {\bm k}_\star) = (E_S, {\bm M})}$ is a quadratic band degeneracy point of $H_{\rm bulk}$ modeling a square lattice medium. For each fixed $k$, $H^\delta_{\rm edge}$ has essential spectrum (gray) determined by the spectra of perturbed bulk Hamiltonians $\smash{H^{\pm, \, \delta}_{\rm bulk}}$ \eqref{eq:asym-op_r}, which both have a band gap of width $O(\delta^2)$ about ${(E_\star, {\bm k}_\star)}$. The band gap of $H^\delta_{\rm edge}$ about ${(E_\star, k_\star)}$ is traversed by two edge state eigenvalue curves (orange). The eigenvalue curves are approximately those of an effective matrix Schr\"{o}dinger operator.}
\label{fig:schem-sql-intro}
\end{figure}

\begin{remark}
\label{rmk:scaling}
{\rm (Scaling.)}
Suppose the dispersion relations have the local behavior ${|E({\bm k}) - E_\star| \sim |{\bm k} - {\bm k}_\star|^r}$ in a quasimomentum neighborhood of the $H_{\rm bulk}$ band structure degeneracy at ${(E_\star, {\bm k}_\star)}$. Then, the dispersive (diffractive) effect on a wavepacket spectrally concentrated in a quasimomentum disk of radius $\delta$ about ${\bm k}_\star$ will be order one on the timescale ${t_{\rm disp} = O(\delta^{-r})}$. In \eqref{eq:asym-op_r}, and later in \eqref{eq:edge-op_r}, the amplitude of the perturbation to $H_{\rm bulk}$ is chosen to balance dispersion on the timescale $t_{\rm disp}$.
\end{remark}

\noindent Since $W_\infty$ breaks the $\mathcal{PC}$ symmetry of $H_{\rm bulk}$, the degeneracy at ${(E_\star, {\bm k}_\star)}$ is lifted; the perturbed Hamiltonians $\smash{H^{\pm, \, \delta}_{\rm bulk}}$ have a common $L^2(\R^2)$ energy gap of width $O(\delta^r)$ about $E_\star$.\footnote{In general, the gap that opens is only local in both energy and quasimomentum; see Section \ref{sec:fb-thy-gaps}.}

We next introduce a class of edge Hamiltonians, periodic with respect to ${\Z \boldsymbol{\mathfrak{v}}_{\rm edge}}$, which slowly interpolate between $\smash{H^{-, \, \delta}_{\rm bulk}}$ and $\smash{H^{+, \, \delta}_{\rm bulk}}$ across ${\R \boldsymbol{\mathfrak{v}}_{\rm edge}}$. Specifically,
\begin{equation}
\label{eq:edge-op_r}
H^\delta_{\rm edge} = H_{\rm bulk} + \delta^r W^\delta_{\rm edge} ,
\end{equation}
where $W^\delta_{\rm edge}$ is a spatially non-compact and slowly varying (with $O(\delta^{-1})$ spatial scale) interpolation between ${\pm W_\infty}$ transverse to $\R \boldsymbol{\mathfrak{v}}_{\rm edge}$. The edge Hamiltonian $\smash{H^\delta_{\rm edge}}$ has a two-scale structure, arising from the $O(1)$ scale of the lattice $\Lambda$ and the $O(\delta^{-1})$ scale of the domain wall perturbation. Analogous edge Hamiltonians have been introduced in the context of 2D honeycomb media \cite{HR08, RH:08, FLW-2d_materials:15, FLW-2d_edge:16, LWZ19, drouotweinstein2020} and in 1D dislocated media \cite{FLW-PNAS:14, FLW-MAMS:17, drouot2020defect, drouot21-discloc}. For studies involving square lattice media with sharp interfaces, see, e.g.,  \cite{ablowitz2020discrete, qiu2024mathematical}.

Since $H^\delta_{\rm edge}$ is invariant under translations in ${\Z \boldsymbol{\mathfrak{v}}_{\rm edge}}$, edge states arise as solutions of a family of Floquet-Bloch eigenvalue problems parameterized by a scalar parallel quasimomentum; see Section \ref{sec:fb-thy-edge}. These states are represented in an {\it edge state diagram}; see Figures \ref{fig:schem-sql-intro} - \ref{fig:schem-dfm-intro}. This edge state diagram displays a gap between two bands of continuous spectrum, arising from the non-degenerate dispersion surfaces of the asymptotic Hamiltonians $H^{\pm, \, \delta}_{\rm bulk}$, as well as branches of eigenvalues corresponding to edge state eigenfunctions.

\subsection{Summary of results}
\label{sec:summary}

\begin{itemize}
\item {\it Edge state eigenvalue curves which traverse the bulk band gap:} We explain, using asymptotic analysis and numerical computations, how the band gap of $H^{\pm, \, \delta}_{\rm bulk}$ containing  the energy-parallel quasimomentum pair ${(E, k) = (E_\star, k_\star = {\bm k}_\star \cdot \boldsymbol{\mathfrak{v}}_{\rm edge})}$ is traversed by two edge state eigenvalue curves $k\smallin [-\pi, \, \pi] \mapsto E(k)$ of $H^\delta_{\rm edge}$ acting in the space $L^2_k(\R^2/\Z\boldsymbol{\mathfrak{v}}_{\rm edge})$; see \eqref{eq:def-L2kpar}.

\item {\it Multiple-scale structure of edge states:} Each of these eigenvalue curves is seeded by the discrete spectrum of an effective (homogenized) edge Hamiltonian derived via a two-scale asymptotic expansion of solutions to the edge state eigenvalue problem. The corresponding edge states have a wavepacket-like structure: a slowly varying envelope modulation of the Floquet-Bloch eigenstates associated with the band degeneracy, which is localized transverse to the edge.

\item {\it Effective (homogenized) edge Hamiltonians:} The effective edge Hamiltonians are of two types: (1) in the square lattice case, a matrix Schr\"{o}dinger operator $\mathfrak{S}(\kappa)$, with ${\kappa \smallin \R}$, arising from a quadratic band degeneracy, and (2) in the deformed square lattice case, a pair of Dirac operators $\cancel{\mathfrak{D}}^\pm(\kappa)$ arising from a pair of tilted conical degeneracies. The matrix Schr\"{o}dinger operator is of the type that arises in the theory of {\it $d$-wave superconductivity}\footnote{We thank G. Bal for observing this.} \cite{volovik-2009}, while the Dirac operators are variants of the well-known Dirac operator arising in the analysis of edge states in 2D honeycomb structures and 1D dislocations \cite{drouot2019a, drouot2020defect, drouotweinstein2020, FLW-2d_materials:15}), quantum field theory \cite{JR76}, as well as the theory of {\it $p$-wave superconductivity} \cite{volovik-2009, bernevig-hughes-2013}.

\item {\it Gap-traversing edge state eigenvalue curves and band structure topology:} The two gap-traversing eigenvalue curves we construct have a topological explanation.

The phenomenon of gap-traversing eigenvalue curves (see Figures \ref{fig:schem-sql-intro} - \ref{fig:schem-dfm-intro}) is a manifestation of the bulk-edge correspondence principle of topological physics; see, e.g., \cite{Hatsugai:93, EG:02, EGS:05, HK:10, Graf-Porta:13, prodan-sb-2016, drouotweinstein2020, bal-jmp-2023}. In the present context, an  integer-valued topological {\it edge index}, which is a physical measure of  quantum  {\it edge conductivity}, is equal to the difference of {\it first Chern numbers}. The Chern number is an integer-valued topological index, which is proportional to the integral of the {\it Berry curvature} over the Brillouin zone of $H^{\pm, \, \delta}_{\rm bulk}$. This edge index is also equal to an algebraic count of the number of edge state eigenvalue curves which traverse the band gap about ${(E, k) = (E_\star, k_\star)}$, or ``spectral flow''; see, e.g., \cite{doll2023spectral}.

Multiple-scale analysis reduces the computation of bulk and edge topological indices to corresponding calculations for effective operators $\mathfrak{S}(\kappa)$ and $\cancel{\mathfrak{D}}^\pm(\kappa)$; see Section \ref{sec:multi-bec}. A general discussion of bulk and edge indices for the types of effective Hamiltonians arising in the present work appears in \cite{bal-cpde-2022, bal-jmp-2023, quinn-ball-sima-2024}; see also \cite{volovik-2009, bernevig-hughes-2013}. In the square lattice case, the effective edge Hamiltonian is known to have two gap traversing eigenvalue curves (Figure \ref{fig:schem-sql-intro}) and, in the deformed square lattice case, each of the two Dirac operators contributes one eigenvalue curve (Figure \ref{fig:schem-dfm-intro}), for a total of two.

\item {\it PDE/Variational analysis of effective edge Hamiltonians:} We present analytical results, supported by numerical simulations, on the effective edge Hamiltonians $\mathfrak{S}(\kappa)$ and $\cancel{\mathfrak{D}}^\pm(\kappa)$ by direct PDE/variational analysis, which yield information about the eigenvalue curves and eigenstates not accessible to approaches rooted in the topology of the space of Fredholm operators.

\item {\it Numerical simulations of the full spectral problem for $H^\delta_{\rm edge}$:} We present detailed numerical simulations of the full spectral problem for the multiple-scale Hamiltonian $\smash{H^\delta_{\rm edge}}$ \eqref{eq:edge-op_r} illustrating these phenomena, and interpret the results using our effective Hamiltonians.
\end{itemize}

\begin{figure}[!t]
\centering
\includegraphics[scale = 1.0]{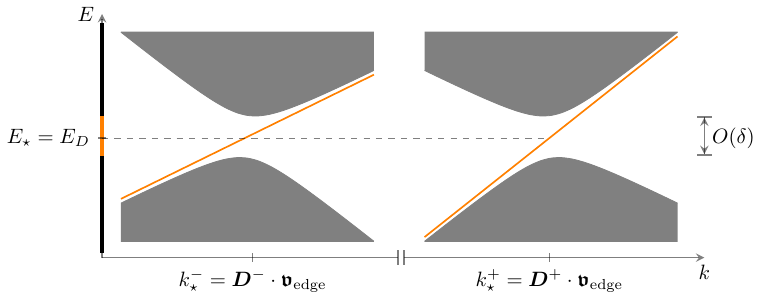}
\caption{(Edge state diagram: deformed square lattice case.) Schematic of the spectrum of the edge Hamiltonian $H^\delta_{\rm edge}$ \eqref{eq:edge-op_r}, as the union of its $L^2_k(\R^2/\Z\boldsymbol{\mathfrak{v}}_{\rm edge})$ \eqref{eq:def-L2kpar} spectra, near ${(E_\star, k^\pm_\star) = (E_D, {\bm D}^\pm \cdot \boldsymbol{\mathfrak{v}}_{\rm edge})}$, where ${(E_\star, {\bm k}^\pm_\star) =}$ ${(E_D, {\bm D}^\pm)}$ are Dirac points of $H_{\rm bulk}$ modeling a deformed square lattice medium; compare with Figure \ref{fig:schem-sql-intro} for the square lattice case. Here, $\smash{H^{\pm, \, \delta}_{\rm bulk}}$ \eqref{eq:asym-op_r} both have band gaps of width $O(\delta)$ about ${(E_\star, {\bm k}^-_\star)}$ and ${(E_\star, {\bm k}^+_\star)}$. The eigenvalue curves of $H^\delta_{\rm edge}$ are approximately those of a pair of effective Dirac operators.}
\label{fig:schem-dfm-intro}
\end{figure}

\subsection{Outline of article}
\label{sec:outline}

The remainder of the article is organized as follows:
\begin{itemize}
\item Section \ref{sec:the-model} provides a detailed, minimally technical introduction to the edge Hamiltonian $H^\delta_{\rm edge}$.

\item Section \ref{sec:fb-thy} provides a brief overview of Floquet-Bloch theory.
 
\item In Section \ref{sec:bulk}, we begin with a brief review of square lattice media modeled by a class of Schr\"odinger operators with {\it square lattice potentials}, which have the symmetries of a tiling of $\R^2$ by squares. We then turn to deformed square lattice media, modeled by Schr\"{o}dinger operators with potentials ${V \circ \, T^{-1}}$, where $V$ is a square lattice potential and $T$ is a real, invertible ${2 \times 2}$ matrix.

In Section \ref{sec:sql-bulk}, we review prior results on the quadratic band degeneracies of Schr\"{o}dinger operators with square lattice potentials \cite{keller2018spectral, keller2020erratum} and, in Section \ref{sec:dfm-bulk}, recent results on the emergence of conical degeneracies when the square lattice potential is linearly deformed \cite{chaban2024instability}.

In Section \ref{sec:bulk-breakC}, we review the lifting of both families of degeneracies under $\mathcal{C}$-breaking perturbations, leading to locally non-degenerate, or isolated bands. 

Finally, $\mathcal{C}$-breaking gives rise to isolated bands with nontrivial topology, reflected in a nonzero Chern number, as reviewed in Section \ref{sec:bulk-edge}. 

\item In Section \ref{sec:edge}, we set up our class of edge Hamiltonians $H^\delta_{\rm edge}$ for an arbitrary prescribed rational edge. 

\item In Section \ref{sec:multi}, we present a construction of edge states with energies near $E_\star$ via a two-scale expansion in powers of the small parameter $\delta$. At leading order, these edge states have the form of slowly varying modulations of Floquet-Bloch eigenstates spanning the $L^2_k(\R^2/\Z\mathfrak{v}_{\rm edge})$ kernel of ${H_{\rm bulk} - E_\star}$. The slowly varying amplitudes provide the localization of edge states in the direction transverse to the edge, and are given by the bound states of an {\it effective edge Hamiltonian}.

\item In Section \ref{sec:spec}, we discuss the spectral properties of these effective edge Hamiltonians. 

\item In Section \ref{sec:numerics}, we present numerical simulations of the full spectral problem for $H^\delta_{\rm edge}$ (which include a computation of edge state curves).  These simulations show that there are always two edge state curves which traverse the band gap; see Figures \ref{fig:schem-sql-intro} - \ref{fig:schem-dfm-intro}, as well as the schematic Figures \ref{fig:schem-sql} - \ref{fig:schem-dfm}. 

\item An overall perspective on the current work and an open problem are discussed in Section \ref{sec:an_open_prob}.
\end{itemize}

\noindent An overview of the bulk band structure degeneracies, as described in \cite{keller2018spectral,keller2020erratum} for quadratic band degeneracies and in \cite{chaban2024instability} for Dirac points, is given in Appendix \ref{apx:par-bulk}. Details of certain calculations in Sections \ref{sec:multi} - \ref{sec:spec} are placed in Appendices \ref{apx:pf_multi} - \ref{apx:dfm-edge-eff-gap}. Further supporting analytic and numerical calculations are located in Supplementary Materials \ref{supp:sql-eff-exact} - \ref{supp:dfm-edge-eff-disc}.

\subsection{Notation and conventions}
\label{sec:notation}

\begin{itemize}
\item Pauli matrices, clockwise $\pi/2$ rotation matrix:
\begin{equation}
\label{eq:def-pauli}
\sigma_0 \equiv I =
\! \begin{bmatrix}
1 & 0 \\
0 & 1
\end{bmatrix} \!
, \quad \sigma_1 \equiv 
\! \begin{bmatrix}
0 & 1 \\
1 & 0
\end{bmatrix} \!
, \quad \sigma_2 \equiv
\! \begin{bmatrix}
0 & -i \\
i & 0
\end{bmatrix} \!
, \quad \sigma_3 \equiv
\! \begin{bmatrix}
1 & 0 \\
0 & -1
\end{bmatrix} \!
, \quad R \equiv i \sigma_2 =
\! \begin{bmatrix}
0 & 1 \\
-1 & 0
\end{bmatrix} \!
.
\end{equation}

\item Symmetry operations:
\begin{equation}
\label{eq:sym-op}
\begin{aligned}
\text{($\mathcal{P}$)} & \quad \text{Parity (spatial inversion)} & \mathcal{P}[f]({\bm x}) & \equiv f(-{\bm x}) , \\
\text{($\mathcal{C}$)} & \quad \text{Time-reversal (complex conjugation)} & \mathcal{C}[f]({\bm x}) & \equiv \smash{\overline{f({\bm x})}} , \\
\text{($\mathcal{R}$)} & \quad \text{Clockwise $\pi/2$ rotation} & \mathcal{R}[f]({\bm x}) & \equiv f(R^\mathsf{T} {\bm x}) , \\
\text{($\Sigma_1$)} & \quad \text{Reflection about ${x_1 = x_2}$} & \Sigma_1[f]({\bm x}) & \equiv f(\sigma_1 {\bm x}) .
\end{aligned}
\end{equation}

\item Primitive vectors of the lattice ${\Z^2 = \Z {\bm v}_1 \oplus \Z {\bm v}_2}$ and its dual lattice ${(\Z^2)^* = 2 \pi \Z {\bm k}_1 \oplus 2 \pi \Z {\bm k}_2}$:
\begin{equation}
{\bm v}_1 \equiv [1, 0]^\mathsf{T} , \quad {\bm v}_2 \equiv [0, 1]^\mathsf{T} \quad \text{and} \quad {\bm k}_1 \equiv [1, 0]^\mathsf{T}, \quad {\bm k}_2 \equiv [0, 1]^\mathsf{T} .
\end{equation}
\end{itemize}

\smallskip

\section*{Acknowledgements}

The authors thank Mikael Rechtsman for helpful conversations relating to numerical simulations. We also thank Guillaume Bal and for stimulating discussions pointing the authors to relevant references. This work was supported in part by NSF grants: DMS-1909035 (JLM), DMS-2307384 (JLM), DMS-1908657 (MIW), DMS-1937254 (JC), and Simons Foundation Math+X Investigator Award \#376319 (MIW). Part of this research was carried out during the 2023-24 academic year, when MIW was a Visiting Member in the School of Mathematics, Institute of Advanced Study, Princeton, supported by the Charles Simonyi Endowment, and a Visiting Fellow in the Department of Mathematics at Princeton University.

\smallskip

\section{Overview of the model and results}
\label{sec:the-model}

\setcounter{equation}{0}
\setcounter{figure}{0}

Our models are built up in three steps: (1) Periodic bulk media with $\mathcal{PC}$ symmetry and band structure degeneracies; (2) Periodic bulk media with $\mathcal{PC}$ symmetry-breaking perturbations, which ``gap'' the degeneracies; (3) Edge media which slowly interpolate, across a line defect, between two distinct periodic bulk media with $\mathcal{PC}$ symmetry-breaking perturbations. A schematic of this construction is presented in Figure \ref{fig:schem-sql} for square lattice media, and in Figure \ref{fig:schem-dfm} for their deformed versions.

\subsection{Periodic bulk media with band structure degeneracies}
\label{sec:the-model-1}

\noindent Our fundamental bulk Hamiltonian is the Schr\"{o}dinger operator
\begin{equation}
\label{eq:sql-bulk-op}
H_{\rm bulk} = H_V \equiv -\Delta_{\bm x} + V({\bm x}) .
\end{equation}
 The potential ${{\bm x} \mapsto V({\bm x})}$ is real-valued, periodic with respect to ${\Lambda = \Z^2}$, and invariant under the symmetry group of the square. We call such potentials {\it square lattice potentials}; see Definition \ref{def:sql-pot}. The spectral theory of ${H_{\rm bulk} = H_V}$ acting in $L^2(\R^2)$ is given by its Floquet-Bloch {\it band structure}; see Section \ref{sec:fb-thy-bulk}. In \cite{keller2018spectral, keller2020erratum}, it was shown that the band structure of a Schr\"{o}dinger operator with a square lattice potential exhibits {\it quadratic band degeneracies} ${(E_\star, {\bm k}_\star) = (E_S, {\bm M})}$; see Figure \ref{fig:schem-sql-1}. The high-symmetry quasimomentum ${{\bm M} = [\pi, \pi]^\mathsf{T}}$ and its rotational equivalents make up the four vertices of the Brillouin zone ${\mathcal{B} = [-\pi, \, \pi] \times [-\pi, \, \pi]}$.

In \cite{chaban2024instability}, we studied small deformations of $H_V$ given by
\begin{equation}
\label{eq:dfm-bulk-op}
H_{\rm bulk} = H_{V \circ \, T^{-1}} \equiv -\Delta_{\bm x} + V(T^{-1}{\bm x}) ,
\end{equation}
where $T$ is a real, invertible ${2 \times 2}$ matrix close to the identity. The bulk Hamiltonian ${H_{\rm bulk} = H_{V \circ \, T^{-1}}}$ is periodic with respect to ${\Lambda = T \Z^2}$, and may not have any of the point group symmetries of $H_V$. We proved that, for generic $T$, the quadratic band degeneracy point at ${(E_S, {\bm M})}$ splits into two conical degeneracies, or {\it Dirac points} ${(E_\star, {\bm k}^\pm_\star) = (E_D, {\bm D}^\pm)}$ at quasimomenta ${\bm D}^\pm$ inversion symmetric about ${\bm M}$; see Figure \ref{fig:schem-dfm-1}. That twofold degeneracies persist under deformation is a consequence of the invariance of $H_{V \circ \, T^{-1}}$ under $\mathcal{PC}$ as defined in \eqref{eq:sym-op}.

\begin{figure}[!t]
\begin{subfigure}{\columnwidth}
    \centering
    \includegraphics[scale = 1.0]{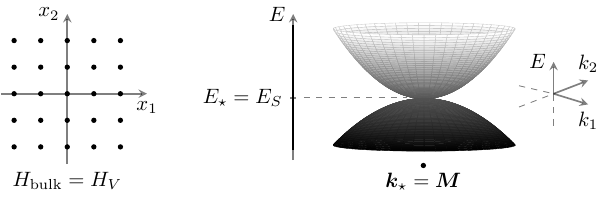}
    \subcaption{{\it (Left)} Square lattice bulk medium, modeled by ${H_{\rm bulk} = H_V}$ \eqref{eq:sql-bulk-op}. {\it (Right)} Blow-up of its band structure near a quadratic band degeneracy point ${(E_\star, {\bm k}_\star) = (E_S, {\bm M})}$.}
    \label{fig:schem-sql-1}
\end{subfigure}
\begin{subfigure}{\columnwidth}
    \centering
    \vspace{0.3cm}
    \includegraphics[scale = 1.0]{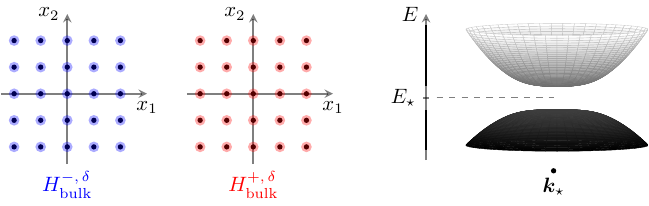}
    \subcaption{{\it (Left)} Asymptotic square lattice media, modeled by $H^{\pm, \, \delta}_{\rm bulk}$ \eqref{eq:bulk-op-breakC}. {\it (Right)} $\mathcal{C}$-breaking perturbation induces a band gap near ${(E_\star, {\bm k}_\star)}$. The two asymptotic media acquire distinct nonzero topological indices (Chern numbers).}
    \label{fig:schem-sql-2}
\end{subfigure}
\begin{subfigure}{\columnwidth}
    \centering
    \vspace{0.3cm}
    \includegraphics[scale = 1.0]{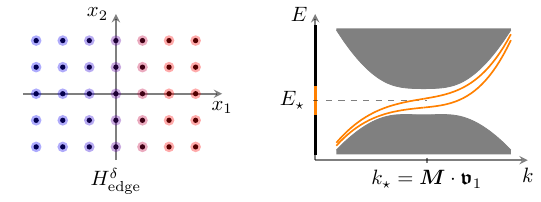}
    \subcaption{{\it (Left)} Square lattice edge medium, modeled by $\smash{H^\delta_{\rm edge}}$ \eqref{eq:edge-op}, which interpolates between $\smash{H^{\pm, \, \delta}_{\rm bulk}}$ transverse to $\smash{\R \boldsymbol{\mathfrak{v}}_{\rm edge}}$. Here, ${\boldsymbol{\mathfrak{v}}_{\rm edge} = [0, 1]^\mathsf{T}}$. {\it (Right)} The band gap is traversed by two edge state eigenvalue curves, arising from the eigenvalue curves of a matrix Schr\"{o}dinger operator.}
    \label{fig:schem-sql-3}
\end{subfigure}
\caption{Square lattice media and spectra near energy ${E_\star = E_S}$.}
\label{fig:schem-sql}
\end{figure}

\subsection{Band structure degeneracies gapped via symmetry breaking}
\label{sec:the-model-2}

\noindent We next construct perturbed bulk Hamiltonians, to be later joined across an interface:
\begin{equation}
\label{eq:bulk-op-breakC}
H^{\pm, \, \delta}_{\rm bulk} = H_{\rm bulk} \pm \delta^r W_\infty({\bm x}, -i\nabla_{\bm x}) .
\end{equation}
Here, ${\delta > 0}$ is small. We set ${r = 2}$ in the undeformed square lattice case, where quadratic degeneracies arise, and ${r = 1}$ in the deformed square lattice case, where conical degeneracies arise; see Remark \ref{rmk:scaling}. The perturbing operator $W_\infty$ is a self-adjoint, $\Lambda$-periodic magnetic term which breaks $\mathcal{PC}$ symmetry; specifically, it commutes with $\mathcal{P}$ and anticommutes with $\mathcal{C}$. Such perturbations cause the band structure degeneracies of Section \ref{sec:the-model-1} to gap, or separate, yielding smooth, locally non-degenerate dispersion surfaces and a local band gap near ${(E_\star, {\bm k}_\star)}$; see Figure \ref{fig:schem-sql-2} and Figure \ref{fig:schem-dfm-2}.
   
If the dispersion surfaces of $\smash{H^{\pm, \, \delta}_{\rm bulk}}$ are non-degenerate across the full Brillouin zone, one can associate integer-valued topological indices, the first Chern numbers, to the collection of bands below the band gap. For $\mathcal{C}$-invariant perturbations of $H_{\rm bulk}$, the Chern number is zero, and for $\mathcal{C}$-breaking perturbations of $H_{\rm bulk}$, it is nonzero. Our models use $\mathcal{C}$-breaking perturbations resulting in Chern numbers of $+1$ and $-1$; see the discussion in Section \ref{sec:bti}.

\begin{figure}[!t]
\begin{subfigure}{\columnwidth}
    \centering
    \includegraphics[scale = 1.0]{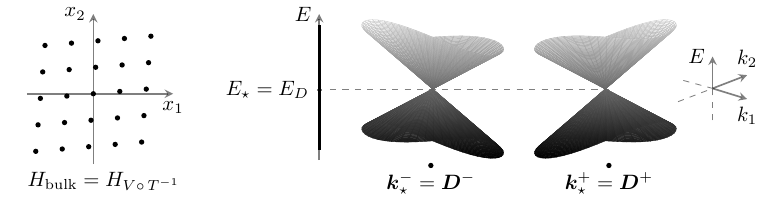}
    \subcaption{{\it (Left)} Deformed square lattice bulk medium, modeled by $\smash{H_{\rm bulk} = H_{V \circ \, T^{-1}}}$ \eqref{eq:dfm-bulk-op}. {\it (Right)} Blow-up of its band structure near two Dirac points $\smash{(E_\star, {\bm k}^\pm_\star) = (E_D, {\bm D}^\pm)}$.}
    \label{fig:schem-dfm-1}
\end{subfigure}
\begin{subfigure}{\columnwidth}
    \centering
    \vspace{0.3cm}
    \includegraphics[scale = 1.0]{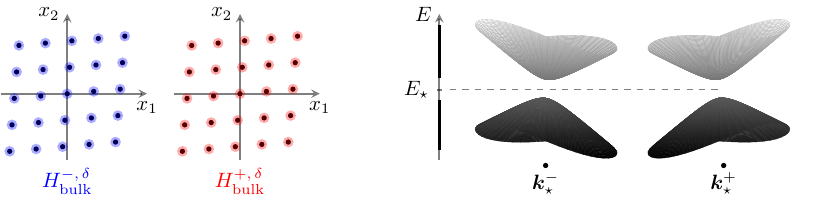}
    \subcaption{{\it (Left)} Asymptotic deformed square lattice media, modeled by $\smash{H^{\pm, \, \delta}_{\rm bulk}}$ \eqref{eq:bulk-op-breakC}. {\it (Right)} $\mathcal{C}$-breaking perturbation induces a band gap near ${(E_\star, {\bm k}^\pm_\star)}$. The two asymptotic media acquire distinct nonzero topological indices (Chern numbers), equal to those arising for the undeformed square lattice; see Figure \ref{fig:schem-sql}.}
    \label{fig:schem-dfm-2}
\end{subfigure}
\begin{subfigure}{\columnwidth}
    \centering
    \vspace{0.3cm}
    \includegraphics[scale = 1.0]{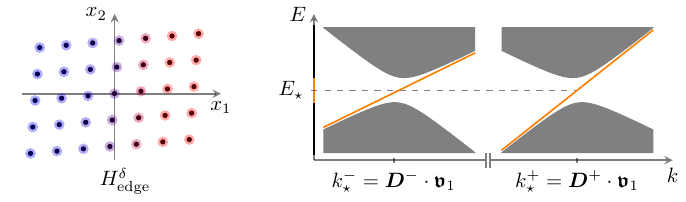}
    \subcaption{{\it (Left)} Deformed square lattice edge medium, modeled by $\smash{H^\delta_{\rm edge}}$ \eqref{eq:edge-op}, which interpolates between $\smash{H^{\pm, \, \delta}_{\rm bulk}}$ across $\smash{\R T \mathfrak{v}_{\rm edge}}$. Again, $\smash{\boldsymbol{\mathfrak{v}}_{\rm edge} = [0, 1]^\mathsf{T}}$. {\it (Right)} The band gap is traversed by two edge state eigenvalue curves, seeded by the eigenvalue curves of two Dirac operators.}
    \label{fig:schem-dfm-3}
\end{subfigure}
\caption{Deformed square lattice media and spectra near energy ${E_\star = E_D}$.}
\label{fig:schem-dfm}
\end{figure}

\subsection{Interpolating between gapped bulk media across a line defect}
\label{sec:the-model-3}

\noindent Finally, choose a vector ${\boldsymbol{\mathfrak{v}}_{\rm edge} \smallin \Lambda}$. The line $\R \boldsymbol{\mathfrak{v}}_{\rm edge}$ divides $\R^2$ into two half-spaces $\Omega_-$ and $\Omega_+$. We define the edge Hamiltonian
\begin{equation}
\label{eq:edge-op}
H^\delta_{\rm edge} = H_{\rm bulk} + \delta^r W^\delta_{\rm edge}({\bm x}, \delta {\bm x}, -i \nabla_{\bm x}) ,
\end{equation}
where $\smash{W^\delta_{\rm edge} \to \pm W_\infty}$, and hence $\smash{H^\delta_{\rm edge} \to H^{\pm, \, \delta}_{\rm bulk}}$ as ${|{\bm x}| \to +\infty}$ for ${{\bm x} \smallin \Omega_\pm}$; see Figure \ref{fig:schem-sql-3} and Figure \ref{fig:schem-dfm-3}. The interpolation is via a domain wall with an $O(\delta^{-1})$ length scale. By the choice of $\boldsymbol{\mathfrak{v}}_{\rm edge}$, $H^\delta_{\rm edge}$ is invariant under translations by elements of $\Z \boldsymbol{\mathfrak{v}}_{\rm edge}$. Consequently, its spectral properties as an operator on $L^2(\R^2)$ can be resolved into those of its restriction onto fiber subspaces  $L^2_k(\R^2/\Z\boldsymbol{\mathfrak{v}}_{\rm edge})$, for ${k \smallin [-\pi, \, \pi]}$ (see Section \ref{sec:fb-thy-edge}), which are pseudoperiodic with wavenumber $k$ parallel to the edge $\R\boldsymbol{\mathfrak{v}}_{\rm edge}$ and square integrable transverse to the edge. Hence, for each fixed ${k \smallin [-\pi, \, \pi]}$, we consider the  spectral problem:
\begin{equation}
\label{eq:edge-evp_1}
H^\delta_{\rm edge} \Psi = E \Psi, \quad \Psi \smallin L^2_k(\R^2/\Z \boldsymbol{\mathfrak{v}}_{\rm edge}) .
\end{equation}
In the remainder of the paper, we address the agenda outlined in the introduction by studying the edge state curves ${k \mapsto E(k)}$ (see Figure \ref{fig:schem-sql-3} and Figure \ref{fig:schem-dfm-3}) via asymptotic analysis of the eigenvalue problem \eqref{eq:edge-evp_1}.

\smallskip

\section{Floquet-Bloch theory}
\label{sec:fb-thy}

\setcounter{equation}{0}
\setcounter{figure}{0}

In this section, we briefly review the spectral theory of periodic, elliptic operators; see, e.g., \cite{RS4-1978, Kuchment:12, Kuchment:16, Eastham:74}.

\subsection{Floquet-Bloch theory: 2D periodicity}
\label{sec:fb-thy-bulk}

Let ${\bm v}_1$, ${{\bm v}_2 \smallin \R^2}$ be linearly independent vectors, and define the 2D {\it (Bravais) lattice} ${\Lambda \equiv \Z {\bm v}_1 \oplus \Z {\bm v}_2}$. Further, define the {\it dual lattice} ${\Lambda^* \equiv 2 \pi \Z {\bm k}_1 \oplus 2 \pi \Z {\bm k}_2}$, where ${\bm k}_1$, ${{\bm k}_2 \smallin \R^2}$ satisfy ${{\bm k}_j \cdot {\bm v}_k = \delta_{j, k}}$ for $j$, ${k = 1}$, $2$. Let $\mathcal{B}$ denote the {\it Brillouin zone}: a choice of fundamental domain for $\R^2/\Lambda^*$. For each ${{\bm k} \smallin \mathcal{B}}$, we define
\begin{equation}
\label{eq:def-L2k}
L^2_{\bm k}(\R^2/\Lambda) \equiv \bigl \{ {\bm x} \mapsto f({\bm x}) \smallin L^2_{\rm loc}(\R^2) : f({\bm x} + {\bm v}) = e^{i {\bm k} \cdot {\bm v}} f({\bm x}), \ {\bm v} \smallin \Lambda \bigr \} .
\end{equation}
Any function ${{\bm x} \mapsto f({\bm x}) \smallin L^2(\R^2)}$ can be expressed as a superposition of functions $\smash{{\bm x} \mapsto \tilde{f}({\bm x}; {\bm k}) \smallin L^2_{\bm k}(\R^2/\Lambda)}$ over ${{\bm k} \smallin \mathcal{B}}$. We express this via the fiber decomposition $\smash{L^2(\R^2) = \int^\oplus_\mathcal{B} L^2_{\bm k}(\R^2/\Lambda) \, {\rm d}{\bm k}}$; see \cite{RS4-1978, Kuchment:16}.

Now consider a self-adjoint $\Lambda$-periodic, elliptic operator
\begin{equation}
\label{eq:def-L}
{\mathcal{L} = -\nabla_{\bm x} \cdot A({\bm x}) \nabla_{\bm x} + Q({\bm x})} ,
\end{equation}
acting in  $L^2(\R^2)$, with ${{\bm x} \mapsto A({\bm x})}$ uniformly positive definite and ${{\bm x} \mapsto Q({\bm x})}$ bounded. Since $\mathcal{L}$ is $\Lambda$-periodic, it maps each $L^2_{\bm k}(\R^2/\Lambda)$ to itself. Hence, the spectrum of $\mathcal{L}$ can be obtained via the family of {\it Floquet-Bloch eigenvalue problems}:
\begin{equation}
\label{eq:Lk-evp}
\mathcal{L} \Phi = E \Phi, \quad \Phi \smallin L^2_{\bm k}(\R^2/\Lambda), \quad {\bm k} \smallin \mathcal{B}.
\end{equation}
The $L^2(\R^2)$ spectrum of $\mathcal{L}$ is equal to the union of its $L^2_{\bm k}(\R^2/\Lambda)$ spectra over ${{\bm k} \smallin \mathcal{B}}$. For each fixed ${{\bm k} \smallin \mathcal{B}}$, the operator $\mathcal{L}$ has compact resolvent and hence real, discrete spectrum:
\begin{equation}
E_1({\bm k}) \leq E_2({\bm k}) \leq ... \leq E_b({\bm k}) \leq ... \to +\infty .
\end{equation}
The functions ${{\bm k} \mapsto E_b({\bm k})}$, ${b \geq 1}$, are known as {\it dispersion relations}, and their graphs over ${{\bm k} \smallin \mathcal{B}}$ are {\it dispersion surfaces}, or {\it bands}. For each fixed ${{\bm k} \smallin \mathcal{B}}$, the corresponding eigenstates ${{\bm x} \mapsto \Phi_b({\bm x}; {\bm k})}$, ${b \geq 1}$, form a complete orthonormal basis of $L^2_{\bm k}(\R^2/\Lambda)$. The collection of all Floquet-Bloch eigenpairs ${{\bm k} \mapsto (E_b({\bm k}), {\bm x} \mapsto \Phi_b({\bm x}; {\bm k}))}$, ${b \geq 1}$, is called the {\it band structure} of $\mathcal{L}$.

\subsection{Floquet-Bloch theory: 1D periodicity}
\label{sec:fb-thy-edge}

We now consider the setting of translation invariance in a single direction ${\boldsymbol{\mathfrak{v}}_1 \smallin \Lambda}$. For each (parallel quasi-momentum) 
${k \smallin \mathcal{B}_1 \equiv [-\pi, \, \pi]}$,  
we define
\begin{align}
\label{eq:def-L2kpar}
L^2_k(\R^2/\Z \boldsymbol{\mathfrak{v}}_1) & \equiv \left \{ {\bm x} \mapsto g({\bm x}) \smallin L^2_{\rm loc}(\R^2) : g({\bm x} + \boldsymbol{\mathfrak{v}}_1) = e^{i k} g({\bm x}), \ \int_{\R^2/\Z\boldsymbol{\mathfrak{v}}_1} |g({\bm x})|^2 \, {\rm d}{\bm x} < +\infty \right \} .
\end{align}
Any function ${{\bm x} \mapsto g({\bm x}) \smallin L^2(\R^2)}$ can be expressed as a superposition of functions ${{\bm x} \mapsto \tilde{g}({\bm x}; k) \smallin L^2_k(\R^2/\Z \boldsymbol{\mathfrak{v}}_1)}$ over ${k \smallin \mathcal{B}_1}$, namely, $\smash{L^2(\R^2) = \int^\oplus_{\mathcal{B}_1} L^2_k(\R^2/\Z \boldsymbol{\mathfrak{v}}_1) \, {\rm d}k}$.

Suppose $\mathcal{L}$, defined in \eqref{eq:def-L}, is translation invariant with respect to the 1D lattice $\Z \boldsymbol{\mathfrak{v}}_1$. Then, in analogy with the previous section, $\mathcal{L}$ maps each $L^2_k(\R^2/\Z\boldsymbol{\mathfrak{v}}_1)$ to itself, and the spectrum of $\mathcal{L}$ acting in $L^2(\R^2)$ can be obtained via the family of spectral problems:
\begin{equation}
\label{eq:kpar-evp}
\mathcal{L} \Psi = E \Psi, \quad \Psi \smallin L^2_k(\R^2/\Z\boldsymbol{\mathfrak{v}}_1), \quad k \smallin \mathcal{B}_1 = [-\pi, \, \pi] . 
\end{equation}
The $L^2(\R^2)$ spectrum of $\mathcal{L}$ is equal to the union of its $L^2_k(\R^2/\Z\boldsymbol{\mathfrak{v}}_1)$ spectra over ${k \smallin \mathcal{B}_1}$. Note that $\mathcal{L}$ acting in $L^2_k(\R^2/\Z\boldsymbol{\mathfrak{v}}_1)$ {\it does not} have a compact resolvent. Hence, for each fixed ${k \smallin \mathcal{B}_1}$, the spectrum will generally have both discrete and continuous components. 
\begin{center}
{\it Values of $E$ in the point spectrum of the spectral problem \eqref{eq:kpar-evp} are the edge state eigenvalues.}
\end{center}
A plot of the spectra of $\mathcal{L}_k$ as $k$ varies over ${\mathcal{B}_1 = [-\pi, \, \pi]}$ is called an {\it edge state diagram}. Edge states appear as curves ${k \smallin \mathcal{B}_1 \mapsto E(k)}$ within a spectral gap of $\mathcal{L}$;
see Figures \ref{fig:schem-sql-intro} and \ref{fig:schem-dfm-intro}.

Finally, we note that there exist ${\boldsymbol{\mathfrak{v}}_2 \smallin \Lambda}$ and $\boldsymbol{\mathfrak{K}}_1$, ${\boldsymbol{\mathfrak{K}}_2 \smallin \Lambda^*}$ satisfying ${\boldsymbol{\mathfrak{K}}_j \cdot \boldsymbol{\mathfrak{v}}_k = \delta_{j, k}}$ for $j$, ${k = 1}$, $2$ such that ${\Lambda = \Z \boldsymbol{\mathfrak{v}}_1 \oplus \Z \boldsymbol{\mathfrak{v}}_2}$ and ${\Lambda^* = 2\pi \Z \boldsymbol{\mathfrak{K}}_1 \oplus 2\pi \Z \boldsymbol{\mathfrak{K}}_2}$; see Section \ref{sec:edge}. Further, any function ${{\bm x} \mapsto \tilde{g}({\bm x}; k) \smallin L^2_k(\R^2/\Z\boldsymbol{\mathfrak{v}}_1)}$ can be expressed as a superposition of functions $\smash{{\bm x} \mapsto \tilde{f}({\bm x}; {\bm k}) \smallin L^2_k(\R^2/\Z\boldsymbol{\mathfrak{v}}_1)}$ over the quasimomentum slice ${q \mapsto {\bm k} = k \boldsymbol{\mathfrak{K}}_1 + q \boldsymbol{\mathfrak{K}}_2}$, where ${q \smallin \mathcal{B}_1}$; that is, ${L^2_k(\R^2/\Z\boldsymbol{\mathfrak{v}}_1) = \int^\oplus_{\mathcal{B}_1} L^2_{k \boldsymbol{\mathfrak{K}}_1 + q \boldsymbol{\mathfrak{K}}_2}(\R^2/\Lambda) \, {\rm d}q}$.

\subsection{Band gaps and spectral gaps}
\label{sec:fb-thy-gaps}

Let ${{\bm k} \mapsto E_\pm({\bm k})}$ be two consecutive dispersion relations of $\mathcal{L}$ as in \eqref{eq:def-L}. Then:

\begin{itemize}
\item We say that $\mathcal{L}$ has a ${{\bm k} \mapsto L^2_{\bm k}(\R^2/\Lambda)}$ {\it band gap} if
\begin{equation}
\min_{{\bm k} \smallin \mathcal{B}} \, (E_+({\bm k}) - E_-({\bm k})) > 0 .
\end{equation}
Equivalently, we say that the dispersion surfaces defined by ${{\bm k} \mapsto E_\pm({\bm k})}$ are {\it non-degenerate}, or {\it isolated}.

\item Denote ${\mathcal{B}_1 \equiv [-\pi, \, \pi]}$. Further, define
\begin{equation}
\mu_+(k) \equiv \min_{q \smallin \mathcal{B}_1} \, E_+(k \boldsymbol{\mathfrak{K}}_1 + q \boldsymbol{\mathfrak{K}}_2) \quad \text{and} \quad \mu_-(k) \equiv \max_{q \smallin \mathcal{B}_1} \, E_-(k \boldsymbol{\mathfrak{K}}_1 + q \boldsymbol{\mathfrak{K}}_2) .
\end{equation}
We say that $\mathcal{L}$ has a ${k \mapsto L^2_k(\R^2/\Z\boldsymbol{\mathfrak{v}}_1)}$ {\it band gap} if
\begin{equation}
\min_{k \smallin \mathcal{B}_1} \, (\mu_+(k) - \mu_-(k)) > 0 .
\end{equation}
If $\mathcal{L}$ has a ${k \mapsto L^2_k(\R^2/\Z\boldsymbol{\mathfrak{v}}_1)}$ band gap, then it necessarily has a ${{\bm k} \mapsto L^2_{\bm k}(\R^2/\Lambda)}$ band gap. In general, the converse does not hold.

\item $\mathcal{L}$ has an $L^2(\R^2)$ {\it spectral gap} if
\begin{equation}
\min_{{\bm k} \smallin \mathcal{B}} \, E_+({\bm k}) - \max_{{\bm k} \smallin \mathcal{B}} \, E_-({\bm k}) > 0 , \quad \text{or equivalently} \quad \min_{k \smallin \mathcal{B}_1} \, \mu_+(k) - \max_{k \smallin \mathcal{B}_1} \, \mu_-(k) > 0 .
\end{equation}
If $\mathcal{L}$ has an $L^2(\R^2)$ spectral gap, then it necessarily has both ${{\bm k} \mapsto L^2_{\bm k}(\R^2/\Lambda)}$ and ${k \mapsto L^2_k(\R^2/\Z\boldsymbol{\mathfrak{v}}_1)}$ band gaps. Again, the converse does not generally hold. 
\end{itemize}
In later sections, we apply analogous terminology to our effective bulk/edge Hamiltonians, which depend on quasimomentum parameters ${{\bm \kappa} \smallin \R^2}$ (bulk) and ${\kappa \smallin \R}$ (edge).

\smallskip

\section{Bulk Hamiltonians for square and deformed square lattice media}
\label{sec:bulk}

\setcounter{equation}{0}
\setcounter{figure}{0}

We introduce two classes of bulk Hamiltonians $H_{\rm bulk}$ and discuss types of  degeneracies that occur in their band structures. Near each degeneracy, the band structure of $H_{\rm bulk}$ is well-approximated by the Fourier symbol of an {\it effective (homogenized) bulk Hamiltonian}. Further, the evolution of wavepackets which are spectrally localized about a given band degeneracy is approximated by the dynamics induced by the corresponding effective Hamiltonian; see, e.g., \cite{FW:14}. Finally, we discuss the lifting of degeneracies via $\mathcal{PC}$ symmetry-breaking perturbations, which induces non-trivial {\it band topology}. This has implications for the global character of edge states, as expressed in the bulk-edge correspondence principle.

\subsection{Bulk Hamiltonians, band structure degeneracies}

\subsubsection{Square lattice bulk Hamiltonian, quadratic band degeneracies}
\label{sec:sql-bulk}

Consider the bulk Hamiltonian given by the Schr\"{o}dinger operator
\begin{equation}
\label{eq:sql-bulk-op_2}
H_{\rm bulk} = H_V = -\Delta_{\bm x} + V({\bm x}) ,
\end{equation}
acting in $L^2(\R^2)$. Here, ${{\bm x} \mapsto V({\bm x})}$ is a {\it square lattice potential} in the following sense:

\begin{definition}
\label{def:sql-pot}
{\rm (Square lattice potential.)}
We say that a smooth potential $V$ is a square lattice potential if it is $\Z^2$-periodic and invariant under all symmetry operations listed in \eqref{eq:sym-op}.
\end{definition}

\noindent Note that the symmetry operations in \eqref{eq:sym-op} generate the full symmetry group of the square (i.e., the dihedral group of order 8). In particular, since ${\mathcal{R}^2 = \mathcal{P}}$, square lattice potentials are invariant under spatial inversion. An example of a square lattice potential is plotted in \cite[Figure 1.1a]{chaban2024instability}; see also \cite{keller2018spectral}.

In \cite{keller2018spectral, keller2020erratum}, it was shown that the band structure of $H_V$ exhibits {\it quadratic band degeneracies} at the high-symmetry quasimomentum ${{\bm M} = [\pi, \pi]^\mathsf{T}}$ at one vertex of the Brillouin zone ${\mathcal{B} = [-\pi, \, \pi] \times [-\pi, \, \pi]}$ and at its other vertices, which are equivalent to ${\bm M}$ modulo $(\Z^2)^*$. Near a quadratic band degeneracy at ${(E_S, {\bm M})}$, the two degenerate dispersion surfaces are approximated by those of an effective bulk Hamiltonian with ${2 \times 2}$ matrix Fourier symbol
\begin{equation}
\label{eq:quad-eff}
H_{\rm eff}^{\bm M}({\bm \kappa}) = (1 - \alpha^{\bm M}_0) |{\bm \kappa}|^2 \sigma_0 - \alpha^{\bm M}_1 ({\bm \kappa} \cdot \sigma_1 {\bm \kappa}) \sigma_1 - \alpha^{\bm M}_2 ({\bm \kappa} \cdot \sigma_3 {\bm \kappa}) \sigma_2 .
\end{equation}
Here, ${{\bm \kappa} = {\bm k} - {\bm M}}$ is a quasimomentum displacement from ${\bm M}$. The parameters $\alpha^{\bm M}_0$, $\alpha^{\bm M}_1$, and ${\alpha^{\bm M}_2 \smallin \R}$ are defined in terms of an orthonormal basis of the degenerate Floquet-Bloch eigenspace at ${(E_S, {\bm M})}$: $\{ \Phi^{\bm M}_1, \, \Phi^{\bm M}_2 = \mathcal{PC}[\Phi^{\bm M}_1] \}$; see Appendix \ref{apx:quad}. The two quadratically touching dispersion surfaces of $H_{\rm eff}^{\bm M}({\bm \kappa})$ are plotted in \cite[Figure 1.1b]{chaban2024instability}.

\subsubsection{Deformed square lattice bulk  Hamiltonian, conical degeneracies (Dirac points)}
\label{sec:dfm-bulk}

In the recent work \cite{chaban2024instability}, the effect of linear deformations of the medium modeled by $H_V$ on its band structure was studied;  specifically, the Schr\"{o}dinger operator
\begin{equation}
\label{eq:dfm-bulk-op_2}
H_{\rm bulk} = H_{V \circ \, T^{-1}} = -\Delta_{\bm x} + V(T^{-1}{\bm x}) ,
\end{equation}
where $T$ is an invertible  ${2 \times 2}$ matrix representing a linear deformation of the square lattice medium described by  ${{\bm x} \mapsto V({\bm x})}$. The resulting potential  ${{\bm x} \mapsto V(T^{-1}{\bm x})}$ is $T\Z^2$-periodic, and may have none of its original point group symmetries. An example of a deformed square lattice potential is plotted in \cite[Figure 1.1c]{chaban2024instability}. By the change of coordinates ${{\bm y} = T^{-1}{\bm x}}$, \eqref{eq:dfm-bulk-op_2} is equivalent to the $\Z^2$-periodic operator
\begin{equation}
\label{eq:dfm-bulk-op_3}
T_* H_V = -\nabla_{\bm y} \cdot (T^\mathsf{T} T)^{-1} \nabla_{\bm y} + V({\bm y}) .
\end{equation}
This change of coordinates enables us work with a fixed lattice and dual lattice for any choice of $T$.

In \cite{chaban2024instability} it is shown that, for $T$ sufficiently close to the identity and not a purely volumetric deformation (i.e., an orthogonal transformation composed with a scaling), the quadratic band degeneracy at ${(E_S, {\bm M})}$ splits into a pair of nearby conical degeneracies, or {\it Dirac points} ${(E_D, {\bm D}^\pm)}$ at quasimomenta ${\bm D}^\pm$ inversion symmetric about ${\bm M}$; see \cite[Figure 1.1d]{chaban2024instability}. In the vicinity of each Dirac point, the two degenerate dispersion surfaces are approximated by those of effective bulk Hamiltonians with ${2 \times 2}$ matrix Fourier symbols
\begin{equation}
\label{eq:dir-eff}
T_* H_{\rm eff}^{{\bm D}^\pm} \! ({\bm \kappa}) = ({\bm \gamma}^{{\bm D}^\pm}_0 \! \cdot {\bm \kappa}) \sigma_0 + ({\bm \gamma}^{{\bm D}^\pm}_1 \! \cdot {\bm \kappa}) \sigma_1 + ({\bm \gamma}^{{\bm D}^\pm}_2 \! \cdot {\bm \kappa}) \sigma_2 .
\end{equation}
Here, ${{\bm \kappa} = {\bm k} - {\bm D}^\pm}$ is a quasimomentum displacement from ${\bm D}^\pm$. The parameters ${\bm \gamma}^{{\bm D}^\pm}_0 \!$, ${\bm \gamma}^{{\bm D}^\pm}_1 \!$, and ${{\bm \gamma}^{{\bm D}^\pm}_2 \! \smallin \R^2}$ are defined in terms of orthonormal basis $\smash{\{ \Phi^{{\bm D}^\pm}_1 \!, \, \Phi^{{\bm D}^\pm}_2 \! = \mathcal{PC}[\Phi^{{\bm D}^\pm}_1] \}}$ spanning the degenerate Floquet-Bloch eigenspaces at ${(E_D, {\bm D}^\pm)}$; see Appendix \ref{apx:dir}. With the choice of basis ${\Phi^{{\bm D}^-}_1 \! = \mathcal{P}[\Phi^{{\bm D}^+}_1]}$ and ${\Phi^{{\bm D}^-}_2 \! = \mathcal{P}[\Phi^{{\bm D}^+}_2]}$, it can be easily shown (see Proposition \ref{prop:gam-sym}) that
\begin{equation}
{\bm \gamma}^{{\bm D}^-}_\ell\! = - {\bm \gamma}^{{\bm D}^+}_\ell\!, \quad 0 \leq \ell \leq 2, \quad \textrm{and hence}\quad {T_* H_{\rm eff}^{{\bm D}^-}\!({\bm \kappa}) = - T_* H_{\rm eff}^{{\bm D}^+}\!({\bm \kappa})}.
\end{equation}
The tilted, conically touching dispersion surfaces of $T_* H_{\rm eff}^{{\bm D}^+}\!({\bm \kappa})$ and $T_* H_{\rm eff}^{{\bm D}^-}\!({\bm \kappa})$ are shown in \cite[Figure 1.2]{chaban2024instability}.

\begin{remark}
\label{rmk:vol-dfm}
It is shown in \cite{chaban2024instability} that, for purely volumetric deformations, the quadratic band degeneracy point {\rm does not} split into two Dirac points; instead, it perturbs to a nearby quadratic band degeneracy point with local character given in \eqref{eq:quad-eff}. This case is identical to the scenario without deformation, so we do not consider it further.
\end{remark}

\subsection{Breaking time-reversal symmetry, lifting degeneracies}
\label{sec:bulk-breakC}

In both cases of the previous section, $H_{\rm bulk}$ commutes with $\mathcal{PC}$, where $\mathcal{P}$ denotes spatial inversion and $\mathcal{C}$ denotes complex conjugation; see \eqref{eq:sym-op}. $\mathcal{PC}$ symmetry is central to the existence of the quadratic degeneracy of $H_V$ (Section \ref{sec:sql-bulk}) and the persistence of degeneracy under deformation, as in $H_{V \circ \, T^{-1}}$ (Section \ref{sec:dfm-bulk}). Under small perturbations which respect self-adjointness and periodicity, but break $\mathcal{PC}$ symmetry, a local ${{\bm k} \mapsto L^2_{\bm k}(\R^2/\Lambda)}$ band gap opens about ${(E_\star, {\bm k}_\star)}$; see \cite{keller2018spectral, keller2020erratum} for $H_V$ and \cite{chaban2024instability} for $H_{V \circ \, T^{-1}}$.

We next study models which break $\mathcal{PC}$ symmetry by breaking $\mathcal{C}$ while preserving $\mathcal{P}$. The complimentary scenario of breaking $\mathcal{P}$ symmetry while preserving $\mathcal{C}$ symmetry is briefly addressed in Remark \ref{rmk:breakP} below.

\subsubsection{Square lattice bulk Hamiltonian, lifting quadratic band degeneracies}
\label{sec:sql-bulk-breakC}

We introduce a $\mathcal{C}$-breaking perturbation of ${H_{\rm bulk} = H_V}$, given in \eqref{eq:sql-bulk-op_2}. For ${\delta > 0}$, define
\begin{equation}
\label{eq:sql-bulk-op-breakC}
H_{V, \, A}^\delta \equiv -\Delta_{\bm x} + V({\bm x}) + \delta^2 \nabla_{\bm x} \cdot A({\bm x}) \sigma_2 \nabla_{\bm x} ,
\end{equation}
where ${{\bm x} \mapsto A({\bm x})}$ is a non-constant, smooth,  $\Z^2$-periodic function which is invariant under the symmetry operations in \eqref{eq:sym-op}. (Note that, for $A$ constant, ${\nabla \cdot A \sigma_2 \nabla}$ vanishes.) This class of $\mathcal{C}$-breaking terms models a {\it magneto-optic effect} \cite{chaban2024instability, FLW-2d_edge:16, HR08, RH:08}. The $\delta^2$ scaling of the $\mathcal{C}$-breaking term is explained in Remark \ref{rmk:scaling}.

Recall that, for ${\delta = 0}$, $\smash{H_{V, \, A}^0 = H_V}$ has quadratic band degeneracies within its band structure; see Section \ref{sec:sql-bulk}. Suppose now that ${(E_S, {\bm M})}$ is one such quadratic band degeneracy of $\smash{H_{V, \, A}^0 = H_V}$. For ${\delta > 0}$ small, the degeneracy is lifted and the bands become non-degenerate. The Fourier symbol of the effective Hamiltonian given in \eqref{eq:quad-eff} perturbs to
\begin{equation}
\label{eq:quad-eff-breakC}
H^{\bm M}_{\rm eff}({\bm \kappa}; \delta) = (1 - \alpha^{\bm M}_0) |{\bm \kappa}|^2 \sigma_0 - \alpha^{\bm M}_1 ({\bm \kappa} \cdot \sigma_1 {\bm \kappa}) \sigma_1 - \alpha^{\bm M}_2 ({\bm \kappa} \cdot \sigma_3 {\bm \kappa}) \sigma_2 + \delta^2 \vartheta^{\bm M} \sigma_3 ,
\end{equation}
where
\begin{equation}
\label{eq:def-thetaM}
\vartheta^{\bm M} \equiv \langle \Phi_1^{\bm M}, \, \nabla \cdot A \sigma_2 \nabla \Phi^{\bm M}_1 \rangle .
\end{equation}
The two gapped dispersion surfaces of $H^{\bm M}_{\rm eff}({\bm \kappa}; \delta)$ are plotted in \cite[Figure 10.3a]{chaban2024instability}.

\subsubsection{Deformed square lattice bulk Hamiltonian, lifting Dirac points}
\label{sec:dfm-bulk-breakC}

Similarly, we introduce a $\mathcal{C}$-breaking perturbation of ${H_{\rm bulk} = H_{V \circ \, T^{-1}}}$, given in \eqref{eq:dfm-bulk-op_2}. For ${\delta > 0}$, let
\begin{equation}
\label{eq:dfm-bulk-op-breakC}
H^\delta_{V \circ \, T^{-1}, \, A \, \circ \, T^{-1}} \equiv -\Delta_{\bm x} + V(T^{-1}{\bm x}) + \delta \nabla_{\bm x} \cdot A(T^{-1}{\bm x}) \sigma_2 \nabla_{\bm x} ,
\end{equation} 
where ${{\bm x} \mapsto A({\bm x})}$ satisfies the same criteria as in the previous section. The $\delta$ scaling chosen for the $\mathcal{C}$-breaking term is discussed in Remark \ref{rmk:scaling}. By the change of coordinates ${{\bm y} = T^{-1}{\bm x}}$, \eqref{eq:dfm-bulk-op-breakC} is equivalent to
\begin{equation}
\label{eq:dfm-bulk-op-breakC_2}
T_* H^\delta_{V, \, A} = - \nabla_{\bm y} \cdot (T^\mathsf{T} T)^{-1} \nabla_{\bm y} + V({\bm y}) + \delta \nabla_{\bm y} \cdot \det(T^{-1}) A({\bm y}) \sigma_2 \nabla_{\bm y} ,
\end{equation}
which is a perturbation of \eqref{eq:dfm-bulk-op_3}. Note that, since $T$ is assumed invertible, ${\det(T^{-1}) \neq 0}$.

Recall that, for ${\delta = 0}$, the band structure of $\smash{T_* H^0_{V, \, A} = T_* H_V}$ has a pair of Dirac points; see Section \ref{sec:dfm-bulk}. Suppose that ${(E_D, \, {\bm D}^\pm)}$ are such a pair of Dirac points of $\smash{T_* H^0_{V, \, A} = T_* H_V}$, inversion symmetric about ${\bm M}$. For ${\delta > 0}$, both degeneracies are lifted and the bands become non-degenerate. The Fourier symbols of the effective Hamiltonians in \eqref{eq:dir-eff} perturb to
\begin{equation}
\label{eq:dir-eff-breakC}
T_* H_{\rm eff}^{{\bm D}^\pm}\!({\bm \kappa}; \delta) = ({\bm \gamma}_0^{{\bm D}^\pm}\! \cdot {\bm \kappa}) \sigma_0 + ({\bm \gamma}_1^{{\bm D}^\pm}\! \cdot {\bm \kappa}) \sigma_1 + ({\bm \gamma}_2^{{\bm D}^\pm}\! \cdot {\bm \kappa}) \sigma_2 + \delta \vartheta^{{\bm D}^\pm} \sigma_3 ,
\end{equation}
where 
\begin{equation}
\label{eq:def-theta-Dpm}
\vartheta^{{\bm D}^\pm} \equiv \det(T^{-1}) \langle \Phi^{{\bm D}^\pm}_1\!, \, \nabla \cdot A \sigma_2 \nabla \Phi^{{\bm D}^\pm}_1 \rangle .
\end{equation}
Using ${\Phi^{{\bm D}^-}_1\! = \mathcal{P}[\Phi^{{\bm D}^+}_1]}$ and ${\Phi^{{\bm D}^-}_2\! = \mathcal{P}[\Phi^{{\bm D}^+}_2]}$, it is easily shown that
\begin{equation}
\label{eq:theta-Dpm-sym}
\vartheta^{{\bm D}^+}\! = \vartheta^{{\bm D}^-} .
\end{equation}
The two gapped dispersion surfaces of $T_* H^{{\bm D}^+}_{\rm eff}({\bm \kappa}; \delta)$ are plotted in \cite[Figure 10.3b]{chaban2024instability}; those of $T_* H^{{\bm D}^-}_{\rm eff}({\bm \kappa}; \delta)$ are qualitatively similar.

\subsection{Bulk topology, bulk-edge correspondence}
\label{sec:bulk-edge}

\subsubsection{Bulk topological indices}
\label{sec:bti}

In Section \ref{sec:bulk-breakC}, we discussed how breaking $\mathcal{PC}$ symmetry lifts band structure degeneracies, opening a local ${{\bm k} \mapsto L^2_{\bm k}(\R^2/\Lambda)}$ band gap between now non-degenerate, or isolated bands. The collection of all eigenspaces associated with bands below the gap, indexed by their quasimomentum, defines a vector bundle over the Brillouin zone ${\mathcal{B} \sim \mathbb{T}^2}$. To this vector bundle, one can assign an integer-valued topological index: the first {\it Chern number}, expressible as an integral of the {\it Berry curvature} over $\mathcal{B}$; see, e.g., \cite{drouot2019a, fruchart2013introduction}.

Two examples of $\mathcal{PC}$ symmetry-breaking terms are those which: (1) break $\mathcal{P}$ symmetry, but preserve $\mathcal{C}$ symmetry, and (2) preserve $\mathcal{P}$ symmetry, but break $\mathcal{C}$ symmetry.

\begin{remark}
\label{rmk:breakP}
{\rm (Topologically trivial case: $\mathcal{P}$-breaking and $\mathcal{C}$-preserving.)}
Consider a $\mathcal{PC}$-invariant bulk Hamiltonian, e.g., ${H_{\rm bulk} = H_V}$ from \eqref{eq:sql-bulk-op_2} or $H_{V \circ \, T^{-1}}$ from \eqref{eq:dfm-bulk-op_2}. Under perturbation by a $\mathcal{P}$-breaking, $\mathcal{C}$-preserving term, $\mathcal{PC}$ symmetry is broken and the band structure degeneracies are lifted. However, the Berry curvature is odd in the presence of $\mathcal{C}$ symmetry, yielding a zero Chern number. 
\end{remark}

\noindent Hence, to obtain nonzero bulk topological indices, $\mathcal{C}$ symmetry must be broken. In our models, this is done by the magneto-optic term; see \eqref{eq:sql-bulk-op-breakC} and \eqref{eq:dfm-bulk-op-breakC}.

For small band gap-opening perturbations, the Berry curvature is concentrated in small quasimomentum neighborhood(s) of the gap opening(s), and the Chern number can be obtained from the perturbed effective bulk Hamiltonian(s) about the former band touching(s); see, e.g., \cite{drouot2019a, chong2008effective}. For a quadratic band degeneracy, this is \eqref{eq:quad-eff-breakC}, and we have
\begin{equation}
c_1(H^\delta; E_S) = - {\rm sgn}(\vartheta^{\bm M}) .
\end{equation} 
For a pair of Dirac points, these are both of \eqref{eq:dir-eff-breakC}. Summing their contributions:
\begin{equation}
c_1(T_* H^\delta; E_D) = - \frac{{\rm sgn}(\vartheta^{{\bm D}^+})}{2} - \frac{{\rm sgn}(\vartheta^{{\bm D}^-})}{2} = - {\rm sgn}(\vartheta^{{\bm D}^+}) ,
\end{equation}
where ${\vartheta^{{\bm D}^-}\! = \vartheta^{{\bm D}^+}}$ from \eqref{eq:theta-Dpm-sym}. For further discussion of band topology in these models, see \cite[Section 10]{chaban2024instability}.

\subsubsection{Bulk-edge correspondence principle}
\label{sec:bec}

In Section \ref{sec:edge}, we introduce a family of edge Hamiltonians $H^\delta_{\rm edge}$ which slowly interpolate between periodic bulk Hamiltonians $\smash{H^{\pm, \, \delta}_{\rm bulk}}$. The operators $\smash{H^{\pm, \, \delta}_{\rm bulk}}$ have a common ${{\bm k} \mapsto L^2_{\bm k}(\R^2/\Lambda)}$ band gap containing the band degeneracy ${(E_\star, {\bm k}_\star)}$ of $\smash{H^{\pm, \, 0}_{\rm bulk} = H_{\rm bulk}}$; see panels (b) of Figures \ref{fig:schem-sql} and \ref{fig:schem-dfm}. Note that ${{\rm spec}(H^{\pm, \, \delta}_{\rm bulk}) \subseteq}$ ${{\rm spec}_{\rm ess}(H^\delta_{\rm edge})}$. Further, when the Chern numbers of $H^{+, \, \delta}_{\rm bulk}$ and $H^{-, \, \delta}_{\rm bulk}$ differ, the ${k \mapsto L^2_k(\R^2/\Z\boldsymbol{\mathfrak{v}}_1)}$ band gap containing ${(E_\star, k_\star = {\bm k}_\star \cdot \boldsymbol{\mathfrak{v}}_1)}$ is traversed by eigenvalue curves ${k \mapsto E(k)}$, where $E(k)$ is in the discrete spectrum of the spectral problem \eqref{eq:edge-evp_1}.

The {\it bulk-edge correspondence principle}, proved in the PDE context of the models in \cite{drouot2021}, states that the difference of Chern numbers ${c_1(H^{+, \, \delta}_{\rm bulk}) - c_1(H^{-, \, \delta}_{\rm bulk})}$ is an algebraic count, or ``spectral flow'', of the number of edge state curves which traverse the the ${k \mapsto L^2_k(\R^2/\Z\boldsymbol{\mathfrak{v}}_1)}$ band gap; see \cite{doll2023spectral}. The spectral flow is also related to a measure of quantum current in non-interacting models of electrons used to describe the integer quantum Hall effect.
 
\smallskip

\section{Rational edges and edge Hamiltonians}
\label{sec:edge}

\setcounter{equation}{0}
\setcounter{figure}{0}

A {\it rational line} in ${\Lambda = \Z {\bm v}_1 \oplus \Z {\bm v}_2}$ is a line ${\R \boldsymbol{\mathfrak{v}}_1}$, where ${\boldsymbol{\mathfrak{v}}_1 \equiv m_1 {\bm v}_1 + n_1 {\bm v}_2}$ for $m_1$, ${n_1 \smallin \Z}$ relatively prime. Given ${(m_1, n_1)}$, there exists a second pair of relatively prime integers ${(m_2, n_2)}$ such that ${m_1 n_2 - m_2 n_1 = 1}$.

\begin{remark}
\label{rmk:alt-edge}
For a rational line defined by ${(m_1, n_1)}$, the solution ${(m_2, n_2)}$ to ${m_1 n_2 - m_2 n_1 = 1}$ is unique modulo an integer translate by ${(m_1, n_1)}$; namely, for any ${\ell \smallin \Z}$, ${(m_2, n_2) + \ell (m_1, n_1)}$ is also a solution. 
\end{remark}

\noindent Setting ${\boldsymbol{\mathfrak{v}}_2 \equiv m_2 {\bm v}_1 + n_2 {\bm v}_2}$, we have ${\Z \boldsymbol{\mathfrak{v}}_1 \oplus \Z \boldsymbol{\mathfrak{v}}_2 = \Lambda}$. The dual lattice vectors are given by $\boldsymbol{\mathfrak{K}}_1 \equiv n_2 {\bm k}_1 - m_2 {\bm k}_2$ and ${\boldsymbol{\mathfrak{K}}_2 \equiv - n_1 {\bm k}_1 + m_1 {\bm k}_2}$, and we recover the relations ${\boldsymbol{\mathfrak{K}}_j \cdot \boldsymbol{\mathfrak{v}}_k = \delta_{j, k}}$ for $j$, ${k = 1}$, $2$.

A rational line ${\R \boldsymbol{\mathfrak{v}}_1}$ divides $\R^2$ into two half-spaces $\Omega_+$ and $\Omega_-$. We next introduce {\it edge Hamiltonians} which model a medium which slowly and smoothly interpolates between two distinct periodic bulk Hamiltonians (of the form introduced in Section \ref{sec:bulk-breakC}) at spatial infinity in $\Omega_+$ and $\Omega_-$ transverse to a rational line defect, or rational edge.

\subsection{Square lattice edge Hamiltonian}
\label{sec:sql-edge}

Let ${H_{\rm bulk} = H_V}$ denote a Schr\"{o}dinger operator with a square lattice potential; see Section \ref{sec:sql-bulk}. Following Section \ref{sec:sql-bulk-breakC}, we introduce a pair of asymptotic bulk operators perturbed by a $\mathcal{C}$-breaking term
\begin{equation}
\label{eq:def_sql-asym}
H^{\delta, \, \pm}_{V, \, A} \equiv -\Delta_{\bm x} + V({\bm x}) \pm \delta^2 \nabla_{\bm x} \cdot A({\bm x}) \sigma_2 \nabla_{\bm x} .
\end{equation}
(Recall that $\delta^r$ for ${r = 2}$ arises for a quadratic band degeneracy of $H_{\rm bulk}$ and ${r = 1}$ for a conical degeneracies; see Remark \ref{rmk:scaling}.) Note that, since ${{\bm x} \mapsto A({\bm x})}$ is even, then ${\mathcal{C} H^{+, \, \delta}_{\rm bulk} = H^{-, \, \delta}_{\rm bulk} \mathcal{C}}$. Hence, $H^{+ , \, \delta}$ and $H^{-, \, \delta}$ have common spectral gapsFix a rational line ${\R \boldsymbol{\mathfrak{v}}_1}$ in ${\Lambda = \Z^2}$. In our construction of an edge Hamiltonian we use a {\it domain wall}.

\begin{definition}
\label{def:dwall}
{\rm (Domain wall functions.)}
A {\it domain wall function} is given by a smooth, real-valued function ${X \mapsto \chi(X)}$ with the asymptotic behavior ${\chi(X) \to \pm 1}$ as ${X \to \pm \infty}$.
\end{definition}

\noindent A simple example of a domain wall function is ${\chi(X) = \tanh(X)}$. In some of our results, we shall require that ${\chi(X)^2 \to 1}$ at a sufficiently rapid rate as ${|X| \to +\infty}$. Further, we believe that our results extend to a relaxed notion of  domain wall function, requiring only that ${\chi(-\infty) \cdot \chi(+\infty) < 0}$.

Let ${X \mapsto \chi(X)}$ denote a domain wall function, as in Definition \ref{def:dwall}. We introduce the edge Hamiltonian:
\begin{equation}
\label{eq:def_sql-edge}
H_{\rm edge}^\delta = H^\delta_{V, \, A} \equiv -\Delta_{\bm x} + V({\bm x}) + \delta^2 \nabla_{\bm x} \cdot \chi(\delta \boldsymbol{\mathfrak{K}}_2 \cdot {\bm x}) A({\bm x}) \sigma_2 \nabla_{\bm x} .
\end{equation}
Since ${\boldsymbol{\mathfrak{K}}_2 \cdot \boldsymbol{\mathfrak{v}}_1 = 0}$, the edge Hamiltonian $H^\delta_{V, \, A}$ commutes with translations parallel to the edge: ${{\bm x} \mapsto {\bm x} + \boldsymbol{\mathfrak{v}_1}}$. \\

\noindent {\bf Notation.} For $V$ and $A$ fixed, we shall denote $H^{\delta, \, \pm}_{V, \, A}$ by $H^{\delta, \, \pm}$ and $H^\delta_{V, \, A}$ by $H^\delta$. \\

\noindent The difference of bulk topological indices between the asymptotic bulk Hamiltonians $H^{+, \, \delta}$ and $H^{\delta, -}$, defined in \eqref{eq:def_sql-asym}, is equal to ${\pm 2}$. The bulk-edge correspondence principle predicts that the edge Hamiltonian $H^\delta$, defined in \eqref{eq:def_sql-edge}, has $2$ edge state curves traversing the band gap in the sense of spectral flow; see Section \ref{sec:bulk-edge}.

\subsection{Deformed square lattice edge Hamiltonian}
\label{sec:dfm-edge}

Let ${H_{\rm bulk} = H_{V \circ \, T^{-1}}}$ be a Schr\"{o}dinger operator with a deformed square lattice potential; see Section \ref{sec:dfm-bulk}. Similarly, following Section \ref{sec:dfm-bulk-breakC}, we introduce a pair of perturbed asymptotic bulk operators:
\begin{equation}
\label{eq:def_dfm-asym_a}
H^{\delta, \, \pm}_{V \circ \, T^{-1}, \, A \, \circ \, T^{-1}} \equiv -\Delta_{\bm x} + V(T^{-1} {\bm x}) \pm \delta \nabla_{\bm x} \cdot A(T^{-1} {\bm x}) \sigma_2 \nabla_{\bm x} .
\end{equation}
Making the change of variables ${{\bm y} = T^{-1} {\bm x}}$ yields:
\begin{equation}
\label{eq:def_dfm-asym}
T_* H^{\delta, \, \pm}_{V, \, A} = -\nabla_{\bm y} \cdot (T^\mathsf{T} T)^{-1} \nabla_{\bm y} + V({\bm y}) \pm \delta \nabla_{\bm y} \cdot \det(T^{-1}) A({\bm y}) \sigma_2 \nabla_{\bm y} ,
\end{equation}
where we used that ${T^{-1} \sigma_2 (T^{-1})^\mathsf{T} = \det(T^{-1}) \sigma_2}$.

Fix a rational line ${\R \tilde{\boldsymbol{\mathfrak{v}}}_1}$ in ${\Lambda = T \Z^2}$. Then, there exists a rational line ${\R \boldsymbol{\mathfrak{v}}_1}$ in $\Z^2$ such that ${\tilde{\boldsymbol{\mathfrak{v}}}_\ell = T \boldsymbol{\mathfrak{v}}_\ell}$ and ${\tilde{\boldsymbol{\mathfrak{K}}}_\ell = (T^{-1})^\mathsf{T} \boldsymbol{\mathfrak{K}}_\ell}$ for ${\ell = 1}$, $2$. Now, let ${X \mapsto \chi(X)}$ denote a domain wall function; see Definition \ref{def:dwall}. We introduce the edge Hamiltonian:
\begin{equation}
\label{eq:def_dfm-edge_a}
H^\delta_{\rm edge} = H^\delta_{V \circ \, T^{-1}, \, A \, \circ \, T^{-1}} \equiv -\Delta_{\bm x} + V(T^{-1} {\bm x}) + \delta \nabla_{\bm x} \cdot \chi(\delta (T^{-1})^\mathsf{T} \boldsymbol{\mathfrak{K}}_2 \cdot {\bm x}) A(T^{-1} {\bm x}) \sigma_2 \nabla_{\bm x} .
\end{equation}
Since $\smash{(T^{-1})^\mathsf{T} \boldsymbol{\mathfrak{K}}_2 \cdot T \boldsymbol{\mathfrak{v}}_1 = \boldsymbol{\mathfrak{K}}_2 T^{-1} T \boldsymbol{\mathfrak{v}}_1 = 0}$, the edge Hamiltonian $\smash{H^\delta_{V \circ \, T^{-1}, \, A \, \circ \, T^{-1}}}$ commutes with translations parallel to the deformed edge: $\smash{{\bm x} \mapsto {\bm x} + T \boldsymbol{\mathfrak{v}}_1}$. Again, making the change of variables ${{\bm y} = T^{-1} {\bm x}}$, we obtain:
\begin{equation}
\label{eq:def_dfm-edge}
T_* H^\delta_{V, \, A} \equiv -\nabla_{\bm y} \cdot (T^\mathsf{T} T)^{-1} \nabla_{\bm y} + V({\bm y}) + \delta \nabla_{\bm y} \cdot \chi(\delta \boldsymbol{\mathfrak{K}}_2 \cdot {\bm y}) \det(T^{-1}) A({\bm y}) \sigma_2 \nabla_{\bm y} .
\end{equation}
The edge Hamiltonian $T_* H^\delta_{V, \, A}$ commutes with translations parallel to the undeformed edge: ${{\bm y} \mapsto {\bm y} + \boldsymbol{\mathfrak{v}}_1}$. \\

\noindent {\bf Notation.} For $V$, $A$, and $T$ fixed, we shall denote $T_* H^{\delta, \, \pm}_{V, \, A}$ by $T_* H^{\delta, \, \pm}$ and $T_* H^\delta_{V, \, A}$ by $T_* H^\delta$. \\

\noindent The difference of bulk topological indices between $T_* H^{+, \, \delta}$ and $T_* H^{-, \, \delta}$, defined in \eqref{eq:def_dfm-asym}, is also equal to ${\pm 2}$. Again, the bulk-edge correspondence principle predicts that the edge Hamiltonian $T_* H^\delta$, defined in \eqref{eq:def_dfm-edge}, has $2$ edge state curves which traverse the band gap in the sense of spectral flow.

\smallskip

\section{Construction of edge states via multiple-scale analysis}
\label{sec:multi}

\setcounter{equation}{0}
\setcounter{figure}{0}

In this section, we  construct approximations to the two gap-traversing edge state eigenvalue curves anticipated by the bulk-edge correspondence principle for each of the edge Hamiltonians ${H^\delta_{\rm edge} = H^\delta}$ and ${T_* H^\delta}$, defined in \eqref{eq:def_sql-edge} and \eqref{eq:def_dfm-edge}. Strategies for extending the approximate edge states to exact, rigorously proven edge states of the eigenvalue problems \eqref{eq:sql-edge-evp} and \eqref{eq:dfm-edge-evp} are discussed in Section \ref{sec:multi-exact}.

The approximate edge states arise from the band structure degeneracies of unperturbed bulk Hamiltonians ${H_{\rm bulk} = H}$ and ${T_* H}$, defined in \eqref{eq:sql-bulk-op_2} and \eqref{eq:dfm-bulk-op_3}. Their bifurcation is governed by {\it effective edge Hamiltonians}. The family of spectra of these effective edge Hamiltonians provide blow-ups of small neighborhoods of both the ${L^2_k(\R^2/\Z \boldsymbol{\mathfrak{v}}_1)}$ spectrum of $\smash{H^\delta_{\rm edge} = H^\delta}$ near ${(E_\star, k_\star = {\bm k}_\star \cdot \boldsymbol{\mathfrak{v}}_1)}$, where ${(E_\star, {\bm k}_\star) = (E_S, {\bm M})}$ is a quadratic band degeneracy of $\smash{H^0_{\rm edge} = H}$, and that of $\smash{H^\delta_{\rm edge} = T_* H^\delta}$ near ${(E_\star, k^\pm_\star) = (E_D, {\bm k}^\pm_\star \cdot \boldsymbol{\mathfrak{v}}_1)}$, where ${(E_\star, {\bm k}^\pm_\star) = (E_D, {\bm D}^\pm)}$ are a pair of Dirac points of $\smash{H^0_{\rm edge} = T_*H}$. For ${H_{\rm edge} = H^\delta}$, the effective edge Hamiltonian is a matrix Schr\"{o}dinger operator $\mathfrak{S}(\kappa)$; see \eqref{eq:sql-edge-eff}. For deformed square lattice media, there are two effective Hamiltonians: these are a pair of Dirac operators $\cancel{\mathfrak{D}}^\pm(\kappa)$; see \eqref{eq:dfm-edge-eff}. The parameter ${\kappa \smallin \R}$ is proportional to an $O(\delta)$ parallel quasimomentum displacement from $k_\star$ or $k^\pm_\star$, respectively. Edge state eigenvalues arise as eigenvalue curves ${\kappa \mapsto \Omega(\kappa)}$ within the band gap of the relevant effective edge Hamiltonian.

\subsection{Construction of approximate edge states: square lattice case}
\label{sec:multi-sql}

The edge state eigenvalue  problem for the (undeformed) edge Hamiltonian ${H^\delta_{\rm edge} = H^\delta}$, given in \eqref{eq:def_sql-edge}, is
\begin{equation}
\label{eq:sql-edge-evp}
H^\delta \Psi = E \Psi , \quad \Psi \smallin L^2_k(\R^2/\Z \boldsymbol{\mathfrak{v}}_1) , \quad k \smallin [-\pi, \, \pi] .
\end{equation}
Let ${(E_\star, {\bm k}_\star) = (E_S, {\bm M})}$ be a quadratic band degeneracy point of ${H^0 = H}$; see Section \ref{sec:sql-bulk}. Further, denote ${k_\star = {\bm M} \cdot \boldsymbol{\mathfrak{v}}_1}$. We use the two-scale character of $\smash{H^\delta}$ to construct approximate, two-scale solutions to \eqref{eq:sql-edge-evp} with energy-parallel quasimomenta ${(E, k)}$ in a neighborhood of ${(E_\star, k_\star)}$.

First, we introduce the effective edge Hamiltonian
\begin{equation}
\label{eq:sql-edge-eff}
\mathfrak{S}(\kappa) \equiv (1 - \alpha^{\bm M}_0) {\bm P}_X(\kappa)\cdot {\bm P}_X(\kappa) \sigma_0 - \alpha^{\bm M}_1 ({\bm P}_X(\kappa) \cdot \sigma_1 {\bm P}_X(\kappa)) \sigma_1 - \alpha^{\bm M}_2 ({\bm P}_X(\kappa) \cdot \sigma_3 {\bm P}_X(\kappa)) \sigma_2 + \vartheta^{\bm M} \chi(X) \sigma_3 ,
\end{equation}
acting on $L^2(\R; \C^2)$, where
\begin{equation}
\label{eq:def_P}
{\bm P}_X(\kappa) \equiv -i \partial_X \boldsymbol{\mathfrak{K}}_2 + \kappa \boldsymbol{\mathfrak{K}}_1 .
\end{equation}
The parameters $\alpha^{\bm M}_\ell$, ${0 \leq \ell \leq 2}$, and $\vartheta^{\bm M}$ are defined in Appendix \ref{apx:quad}. The following result states that the bound state eigenpairs of $\mathfrak{S}(\kappa)$ (formally) seed edge state curves in the ${k \mapsto L^2_k(\R^2/\Z\boldsymbol{\mathfrak{v}}_1)}$ band gap of $H^\delta$, within an ${O(\delta^2) \times O(\delta)}$ energy-parallel quasimomentum neighborhood of ${(E_\star, k_\star)}$.

\begin{theorem}
\label{thm:multi-sql}
{\rm (Approximate edge states near ${(E_\star, k_\star)}$.)}
Fix $\kappa$ in a bounded interval about ${\kappa = 0}$. Then, for ${\delta > 0}$ sufficiently small, the edge state eigenvalue problem \eqref{eq:sql-edge-evp} has an approximate solution ${(E^\delta(k_\star + \delta \kappa),}$ ${{\bm x} \mapsto \Psi^\delta({\bm x}; k_\star + \delta \kappa))}$ given by
\begin{align}
E^\delta(k_\star + \delta \kappa) & = E_\star + \delta^2 E^{(2)}(\kappa) + O(\delta^3) , \\
\Psi^\delta({\bm x}; k_\star + \delta \kappa) & = e^{i \delta \kappa \boldsymbol{\mathfrak{K}}_1 \cdot {\bm x}} \bigl( a_1(\delta \boldsymbol{\mathfrak{K}}_2 \cdot {\bm x}; \kappa) \Phi^{\bm M}_1({\bm x}) + a_2(\delta \boldsymbol{\mathfrak{K}}_2 \cdot {\bm x}; \kappa) \Phi^{\bm M}_2({\bm x}) + O(\delta) \bigr) \ \ \text{as} \ \ \delta \to 0 ,
\end{align}
where ${(E^{(2)}(\kappa), X \mapsto [a_1(X; \kappa), a_2(X; \kappa)]^\mathsf{T})}$ is a bound state eigenpair of $\mathfrak{S}(\kappa)$.
\end{theorem}

\begin{example}[Vertical edge]
\label{ex:sql-edge-eff-vert}
Take the edge direction ${\boldsymbol{\mathfrak{v}}_1 = {\bm v}_2 =[0, 1]^\mathsf{T}}$. Then, ${\boldsymbol{\mathfrak{v}}_2 = [-1, 0]^\mathsf{T}}$, ${\boldsymbol{\mathfrak{K}}_1 = [0, 1]^\mathsf{T}}$, and ${\boldsymbol{\mathfrak{K}}_2 = [-1, 0]^\mathsf{T}}$. Further, we have ${k_\star = {\bm M} \cdot \boldsymbol{\mathfrak{v}}_1 = \pi}$. From \eqref{eq:def_P}, ${{\bm P}_X(\kappa) = [-P_X, \kappa]^\mathsf{T}}$, where ${P_X = -i \partial_X}$. Hence, \eqref{eq:sql-edge-eff} becomes
\begin{equation}
\mathfrak{S}(\kappa) = (1 - \alpha^{\bm M}_0) (P_X^2 + \kappa^2) \sigma_0 + 2 \alpha^{\bm M}_1 P_X \kappa \sigma_1 + \alpha^{\bm M}_2 (-P_X^2 + \kappa^2) \sigma_2 + \vartheta^{\bm M} \chi(X) \sigma_3 .
\end{equation}
\end{example}

\noindent Analytical and computational results on the spectral properties of $\mathfrak{S}(\kappa)$ are presented in Section \ref{sec:spec-sql}. Further, a strategy for extending the approximate edge states of $H^\delta$, given in Theorem \ref{thm:multi-sql}, to true solutions of the edge state eigenvalue problem \eqref{eq:sql-edge-evp} is discussed in Section \ref{sec:multi-exact}. \\

\noindent {\it Proof of Theorem \ref{thm:multi-sql}.} Consider the eigenstate ansatz
\begin{equation}
\Psi({\bm x}; k_\star + \delta \kappa) = e^{i \delta \kappa \boldsymbol{\mathfrak{K}}_1 \cdot {\bm x}} \psi({\bm x}; \kappa) , \quad \text{where} \quad {\bm x} \mapsto \psi({\bm x}; \kappa) \smallin L^2_{k_\star}(\R^2/\Z\boldsymbol{\mathfrak{v}}_1) .
\end{equation}
Substituting into the edge state eigenvalue problem \eqref{eq:sql-edge-evp} yields
\begin{align}
\label{eq:sql-edge-evp_2}
& \bigl( -(\nabla_{\bm x} + i \delta \kappa \boldsymbol{\mathfrak{K}}_1) \cdot (\nabla_{\bm x} + i \delta \kappa \boldsymbol{\mathfrak{K}}_1) + V({\bm x}) \\
& \qquad + \delta^2 (\nabla_{\bm x} + i \delta \kappa \boldsymbol{\mathfrak{K}}_1) \cdot \chi(\delta \boldsymbol{\mathfrak{K}}_2 \cdot {\bm x}) A({\bm x}) \sigma_2 (\nabla_{\bm x} + i \delta \kappa \boldsymbol{\mathfrak{K}}_1) \bigr) \psi({\bm x}; \kappa) = E \psi({\bm x}; \kappa) . \nonumber
\end{align}
For ${\delta > 0}$ small, \eqref{eq:sql-edge-evp_2} has explicitly two-scale character. We seek a two-scale solution, depending on ${\bm x}$ and ${X = \delta \boldsymbol{\mathfrak{K}}_2 \cdot {\bm x}}$, of the form:
\begin{align}
E^\delta(\kappa) & = E_\star + \delta E^{(1)}(\kappa) + \delta^2 E^{(2)}(\kappa) + O(\delta^3) , \\
\psi^\delta({\bm x}; \kappa) & = \psi^{(0)}({\bm x}, X; \kappa) + \delta \psi^{(1)}({\bm x}, X; \kappa) + \delta^2 \psi^{(2)}({\bm x}, X; \kappa) + O(\delta^3) \ \ \text{as} \ \ \delta \to 0 .
\end{align}
With this extended set of variables, we have ${\nabla_{\bm x} \to \nabla_{\bm x} + \delta \partial_X \boldsymbol{\mathfrak{K}}_2}$. To impose $\smash{{\bm x} \mapsto \psi^\delta({\bm x}; \kappa) \smallin L^2_{k_\star}(\R^2 / \Z \boldsymbol{\mathfrak{v}}_1)}$, we require that, for ${0 \leq \ell \leq 2}$,
\begin{equation}
{\bm x} \mapsto \psi^{(\ell)}({\bm x}, X; \kappa) \smallin L^2_{\bm M}(\R^2/\Z^2) \quad \text{and} \quad X \mapsto \psi^{(\ell)}({\bm x}, X; \kappa) \smallin L^2(\R) .
\end{equation}
Substituting and grouping terms with like orders of $\delta$ yields a hierarchy of equations:
\begin{subequations}
\label{eq:multi-sql}
\begin{align}
\label{eq:multi-sql-0}
\delta^0: \quad & (H - E_\star) \psi^{(0)}({\bm x}, X; \kappa) = 0 , \\
\label{eq:multi-sql-1}
\delta^1: \quad & (H - E_\star) \psi^{(1)}({\bm x}, X; \kappa) = (E^{(1)}(\kappa) + (-i \partial_X \boldsymbol{\mathfrak{K}}_2 + \kappa \boldsymbol{\mathfrak{K}}_1) \cdot 2i \nabla_{\bm x}) \psi^{(0)}({\bm x}, X; \kappa) , \\
\label{eq:multi-sql-2}
\delta^2: \quad & (H - E_\star) \psi^{(2)}({\bm x}, X; \kappa) = (E^{(1)}(\kappa) + (-i \partial_X \boldsymbol{\mathfrak{K}}_2 + \kappa \boldsymbol{\mathfrak{K}}_1) \cdot 2i \nabla_{\bm x}) \psi^{(1)}({\bm x}, X; \kappa) \\ 
& \qquad + (E^{(2)}(\kappa) - (-i \partial_X \boldsymbol{\mathfrak{K}}_2 + \kappa \boldsymbol{\mathfrak{K}}_1)^2 - \chi(X) \nabla_{\bm x} \cdot A({\bm x}) \sigma_2 \nabla_{\bm x}) \psi^{(0)}({\bm x}, X; \kappa) . \nonumber
\end{align}
\end{subequations}
We solve the equations in order with $X$, ${\kappa \smallin \R}$ fixed.

First, we solve the order $\delta^0$ equation \eqref{eq:multi-sql-0}. This homogeneous equation for ${H - E_\star}$ has the solution
\begin{equation}
{\bm x} \mapsto \psi^{(0)}({\bm x}, X; \kappa) = a^{(0)}_k(X; \kappa) \Phi^{\bm M}_k({\bm x}) .
\end{equation}
(Summation over repeated indices is implied.) The functions ${X \mapsto a^{(0)}_k(X; \kappa)}$, ${k = 1}$, $2$, are determined below and decay rapidly as ${|X| \to +\infty}$. (In the statement of Theorem \ref{thm:multi-sql}, we omit their superscripts.)

Next, we consider the $\delta^\ell$ equations, for ${\ell \geq 1}$, in the hierarchy \eqref{eq:multi-sql}. For each fixed $X$, ${\kappa \smallin \R}$, we must solve a problem of the form:
\begin{equation}
\label{eq:sql-multi-inh}
(H - E_\star) \tilde{\psi}({\bm x}, X; \kappa) = F({\bm x}, X; \kappa), \quad \text{where} \quad {\bm x} \mapsto \tilde{\psi}({\bm x}, X; \kappa), \, F({\bm x}, X; \kappa) \smallin L^2_{\bm M}(\R^2/\Z^2) .
\end{equation}
The inhomogeneous problem \eqref{eq:sql-multi-inh} is solvable if and only if ${{\bm x} \mapsto F({\bm x}, X; \kappa)}$ is orthogonal to ${\C \Phi^{\bm M}_1 \oplus \C \Phi^{\bm M}_2}$, the $L^2_{\bm M}(\R^2/\Z^2)$ kernel of ${H - E_\star}$. Any solution of \eqref{eq:sql-multi-inh} is of the form:
\begin{equation}
{\bm x} \mapsto \tilde{\psi}({\bm x}, X; \kappa) = \tilde{\psi}_p({\bm x}, X; \kappa) + \tilde{a}_k(X; \kappa) \Phi^{\bm M}_k({\bm x}) ,
\end{equation}
where ${{\bm x} \mapsto \tilde{\psi}_p({\bm x}, X; \kappa)}$ is the unique (particular) solution of \eqref{eq:sql-multi-inh} orthogonal to ${\C \Phi^{\bm M}_1 \oplus \C \Phi^{\bm M}_2}$. We write
\begin{equation}
{\bm x} \mapsto \tilde{\psi}_p({\bm x}, X; \kappa) = \mathscr{R}(E_\star) F({\bm x}, X; \kappa) ,
\end{equation}
where $\mathscr{R}(E_\star)$ is the resolvent ${(H - E_\star)^{-1}}$ acting in the $L^2_{\bm M}(\R^2/\Z^2)$ orthogonal complement of ${\C \Phi^{\bm M}_1 \oplus \C \Phi^{\bm M}_2}$.

We now turn to the order $\delta^1$ equation \eqref{eq:multi-sql-1}.

\begin{proposition}
\label{prop:multi-sql-1-sol}
The order $\delta^1$ equation \eqref{eq:multi-sql-1} is solvable if and only if ${E^{(1)}(\kappa) = 0}$.
\end{proposition}

\noindent Proposition \ref{prop:multi-sql-1-sol} is proven in Appendix \ref{apx:pf_multi-sql}. Henceforth, we set ${E^{(1)}(\kappa) = 0}$. The solution is then
\begin{align}
{\bm x} \mapsto \psi^{(1)}({\bm x}, X; \kappa) & = \mathscr{R}(E_\star) (-i \partial_X \boldsymbol{\mathfrak{K}}_2 + \kappa \boldsymbol{\mathfrak{K}}_1) \cdot 2i \nabla_{\bm x} \psi^{(0)}({\bm x}, X; \kappa) + a^{(1)}_k(X; \kappa) \Phi^{\bm M}_k({\bm x}) \\
& = \mathscr{R}(E_\star) (-i \partial_X \boldsymbol{\mathfrak{K}}_2 + \kappa \boldsymbol{\mathfrak{K}}_1) \cdot 2i \nabla_{\bm x} a^{(0)}_k(X; \kappa) \Phi^{\bm M}_k({\bm x}) + a^{(1)}_k(X; \kappa) \Phi^{\bm M}_k({\bm x}) . \nonumber
\end{align}
The functions $\smash{X \mapsto a^{(1)}_k(X; \kappa)}$, ${k = 1}$, $2$, to be determined at higher order, decay rapidly as ${|X| \to +\infty}$.

Finally, we solve the order $\delta^2$ equation \eqref{eq:multi-sql-2}. By the above discussion, \eqref{eq:multi-sql-2} is solvable if and only if the right-hand side is orthogonal to ${\C \Phi^{\bm M}_1 \oplus \C \Phi^{\bm M}_2}$. This solvability condition is equivalent to the eigenvalue problem for the effective edge Hamiltonian $\mathfrak{S}(\kappa)$, given in \eqref{eq:sql-edge-eff}:

\begin{proposition}
\label{prop:multi-sql-2-sol}
The order $\delta^2$ equation \eqref{eq:multi-sql-2} is solvable if and only if the pair ${(E^{(2)}(\kappa), X \mapsto a^{(0)}(X; \kappa))}$, with ${a^{(0)}(X; \kappa) = [a^{(0)}_1(X; \kappa), a^{(0)}_2(X; \kappa)]^\mathsf{T}}$, is a solution of
\begin{equation}
\label{eq:sql-edge-eff-evp}
\mathfrak{S}(\kappa) a^{(0)}(X; \kappa) = E^{(2)}(\kappa) a^{(0)}(X; \kappa) , \quad \text{where} \quad X \mapsto a^{(0)}(X; \kappa) \smallin L^2(\R; \C^2).
\end{equation}
\end{proposition}

\noindent Proposition \ref{prop:multi-sql-2-sol} is proven in Appendix \ref{apx:pf_multi-sql}. We have shown that the leading order term of an asymptotic solution of the edge state eigenvalue problem \eqref{eq:sql-edge-evp} near ${(E, k) = (E_\star, k_\star)}$ arises from the $L^2(\R; \C^2)$ eigenstates of $\mathfrak{S}(\kappa)$. The proof of Theorem \ref{thm:multi-sql} is complete. \qed

\subsection{Construction of approximate edge states: deformed square lattice case}
\label{sec:multi-dfm}

The edge state eigenvalue problem for the deformed edge Hamiltonian ${H^\delta_{\rm edge} = T_* H^\delta}$, given in \eqref{eq:def_dfm-edge}, is
\begin{equation}
\label{eq:dfm-edge-evp}
T_* H^\delta \Psi = E \Psi , \quad \Psi \smallin L^2_k(\R^2/\Z\boldsymbol{\mathfrak{v}}_1) , \quad k \smallin [-\pi, \, \pi] .
\end{equation}
Suppose ${(E_\star, {\bm k}^\pm_\star) = (E_D, {\bm D}^\pm)}$ is a pair of Dirac points of ${T_* H^0 = T_* H}$; see Section \ref{sec:dfm-bulk}. Further, denote ${k^\pm_\star = {\bm D}^\pm \cdot \boldsymbol{\mathfrak{v}}_1}$. We use the two-scale character of $\smash{T_* H^\delta}$ to construct approximate, two-scale solutions to \eqref{eq:dfm-edge-evp} with energy-parallel quasimomenta ${(E, k)}$ in a neighborhood of ${(E_\star, k^\pm_\star)}$.

We introduce the pair of effective edge Hamiltonians
\begin{equation}
\label{eq:dfm-edge-eff}
\cancel{\mathfrak{D}}^\pm(\kappa) \equiv ({\bm \gamma}_0^{{\bm D}^\pm} \! \cdot {\bm P}_X(\kappa)) \sigma_0 + ({\bm \gamma}_1^{{\bm D}^\pm} \! \cdot {\bm P}_X(\kappa)) \sigma_1 + ({\bm \gamma}_2^{{\bm D}^\pm} \! \cdot {\bm P}_X(\kappa)) \sigma_2 + \vartheta^{{\bm D}^\pm} \chi(X) \sigma_3 ,
\end{equation}
acting on $L^2(\R; \C^2)$, with ${\bm P}_X$ as in \eqref{eq:def_P}. The parameters ${\gamma}^{{\bm D}^\pm}_\ell$, ${0 \leq \ell \leq 2}$, and $\vartheta^{{\bm D}^\pm}$ are defined in Appendix \ref{apx:dir}. The following result states that the bound state eigenpairs of $\cancel{\mathfrak{D}}^\pm(\kappa)$ (formally) seed edge state curves in the band gap of $T_* H^\delta$, within an ${O(\delta) \times O(\delta)}$ energy-quasimomentum neighborhood of ${(E_\star, k^\pm_\star)}$.

\begin{theorem}
\label{thm:multi-dfm}
{\rm (Approximate edge states near ${(E_\star, k^\pm_\star)}$.)} 
For $\kappa$ in an interval about ${\kappa = 0}$ and ${\delta > 0}$ small, the edge state eigenvalue problem \eqref{eq:dfm-edge-evp} for ${T_* H^\delta}$ has solutions ${(E^\delta(k^\pm_\star + \delta \kappa), {\bm x} \mapsto \Psi^\delta({\bm x}; k^\pm_\star + \delta \kappa))}$ with
\begin{align}
E^\delta(k^\pm_\star + \delta \kappa) & = E_\star + \delta E^{(1)}(\kappa) + O(\delta^2) , \\
\Psi^\delta({\bm x}; k^\pm_\star + \delta \kappa) & = e^{i \delta \kappa \boldsymbol{\mathfrak{K}}_1 \cdot {\bm x}} \bigl( b_1(\delta \boldsymbol{\mathfrak{K}}_2 \cdot {\bm x}; \kappa) \Phi^{{\bm D}^+}_1\!({\bm x}) + b_2(\delta \boldsymbol{\mathfrak{K}}_2 \cdot {\bm x}; \kappa) \Phi^{{\bm D}^+}_2\!({\bm x}) + O(\delta) \bigr) \ \ \text{as} \ \ \delta \to 0 ,
\end{align}
where:
\begin{enumerate}
\item \label{itm:multi-dfm-p} 
For ${k = k^+_\star + \delta \kappa}$, ${(E^{(1)}(\kappa), X \mapsto [b_1(X; \kappa), b_2(X; \kappa)]^\mathsf{T})}$ is a bound state eigenpair of $\cancel{\mathfrak{D}}^+(\kappa)$, from \eqref{eq:dfm-edge-eff}.
\item \label{itm:multi-dfm-m} For ${k = k^-_\star + \delta \kappa}$, ${(E^{(1)}(\kappa), X \mapsto [b_1(X; \kappa), b_2(X; \kappa)]^\mathsf{T})}$ is instead a bound state eigenpair of $\cancel{\mathfrak{D}}^-(\kappa)$.
\end{enumerate}
\end{theorem}

\noindent Analytical and computational results on the spectral properties of $\cancel{\mathfrak{D}}^\pm(\kappa)$ are presented in Section \ref{sec:spec-dfm}.  \\

\noindent {\it Proof of Theorem \ref{thm:multi-dfm}.} We prove part \ref{itm:multi-dfm-p}; part \ref{itm:multi-dfm-m} proceeds analogously, replacing ${\bm D}^+$ with ${\bm D}^-$. Consider the eigenstate ansatz
\begin{equation}
\Psi({\bm x};k^+_\star + \delta \kappa) = e^{i \delta \kappa \boldsymbol{\mathfrak{K}}_1 \cdot {\bm x}} \psi({\bm x}; \kappa) \, , \quad \text{where} \quad {\bm x} \mapsto \psi({\bm x}; \kappa) \smallin L^2_{k^+_\star}(\R^2/\Z\boldsymbol{\mathfrak{v}}_1) .
\end{equation}
Substituting into the edge state eigenvalue problem \eqref{eq:dfm-edge-evp} yields 
\begin{align}
\label{eq:dfm-edge-evp_2}
& \bigl( -(\nabla_{\bm x} + i \delta \kappa \boldsymbol{\mathfrak{K}}_1) \cdot (T^\mathsf{T} T)^{-1} (\nabla_{\bm x} + i \delta \kappa \boldsymbol{\mathfrak{K}}_1) + V({\bm x}) \\
& \qquad + \delta (\nabla_{\bm x} + i \delta \kappa \boldsymbol{\mathfrak{K}}_1) \cdot \det(T^{-1}) \chi(\delta \boldsymbol{\mathfrak{K}}_2 \cdot {\bm x}) A({\bm x}) \sigma_2 (\nabla_{\bm x} + i \delta \kappa \boldsymbol{\mathfrak{K}}_1) \bigr) \psi({\bm x}; \kappa) = E \psi({\bm x}; \kappa) . \nonumber
\end{align}
For ${\delta > 0}$ small, \eqref{eq:dfm-edge-evp_2} has explicitly two-scale character. We seek a two-scale solution, depending on ${\bm x}$ and ${X = \delta \boldsymbol{\mathfrak{K}}_2 \cdot {\bm x}}$, of the form:
\begin{align}
E^\delta(\kappa) & = E_\star + \delta E^{(1)}(\kappa) + O(\delta^2) , \\
\psi^\delta({\bm x}; \kappa) & = \psi^{(0)}({\bm x}, X; \kappa) + \delta \psi^{(1)}({\bm x}, X; \kappa) + O(\delta^2) \ \ \text{as} \ \ \delta \to 0 .
\end{align}
with this extended set of variables, we have ${\nabla_{\bm x} \to \nabla_{\bm x} + \delta \partial_X \boldsymbol{\mathfrak{K}}_2}$. To impose $\smash{{\bm x} \mapsto \psi^\delta({\bm x}; \kappa) \smallin L^2_{k^+_\star}(\R^2/\Z\boldsymbol{\mathfrak{v}}_1)}$, we require that, for ${\ell = 0}$, $1$,
\begin{equation}
{\bm x} \mapsto \psi^{(\ell)}({\bm x}, X; \kappa) \smallin L^2_{{\bm D}^+}(\R^2/\Z^2) \quad \text{and} \quad X \mapsto \psi^{(\ell)}({\bm x}, X; \kappa) \smallin L^2(\R) .
\end{equation}
Substituting and grouping terms with like orders of $\delta$ yields a hierarchy of equations:
\begin{subequations}
\begin{align}
\label{eq:multi-dfm-0}
\delta^0: \quad & (T_* H - E_\star) \psi^{(0)}({\bm x}, X; \kappa) = 0 , \\
\label{eq:multi-dfm-1}
\delta^1: \quad & (T_* H - E_\star) \psi^{(1)}({\bm x}, X; \kappa) = (E^{(1)}(\kappa) + (-i \partial_X \boldsymbol{\mathfrak{K}}_2 + \kappa \boldsymbol{\mathfrak{K}}_1) \cdot (T^\mathsf{T} T)^{-1} 2i \nabla_{\bm x} \\
& \qquad - \chi(X) \nabla_{\bm x} \cdot \det(T^{-1}) A({\bm x}) \sigma_2 \nabla_{\bm x}) \psi^{(0)}({\bm x}, X; \kappa) . \nonumber
\end{align}
\end{subequations}
We solve the equations in order with $X$, ${\kappa \smallin \R}$ fixed.

First, we solve the order $\delta^0$ equation \eqref{eq:multi-dfm-0}. This homogeneous equation for ${T_* H - E_\star}$ has the solution
\begin{equation}
{\bm x} \mapsto \psi^{(0)}({\bm x}, X; \kappa) = b^{(0)}_k(X; \kappa) \Phi^{{\bm D}^+}_k({\bm x}) .
\end{equation}
(Summation over repeated indices is implied.) The functions ${X \mapsto b^{(0)}_k(X; \kappa)}$, ${k = 1}$, $2$, are determined below and decay rapidly as ${|X| \to +\infty}$. (In the statement of Theorem \ref{thm:multi-dfm}, we omit their superscripts.)

Next, we solve the order $\delta^1$ equation \eqref{eq:multi-dfm-1}. In a manner analogous to the discussion of Section \ref{sec:multi-sql}, \eqref{eq:multi-dfm-1} has a solution if and only if the right-hand side is orthogonal to $\smash{\C \Phi^{{\bm D}^+}_1 \oplus \C \Phi^{{\bm D}^+}_2}$. This solvability condition is equivalent to the eigenvalue problem for the effective edge Hamiltonian $\cancel{\mathfrak{D}}^+(\kappa)$ given in \eqref{eq:dfm-edge-eff}:

\begin{proposition}
\label{prop:multi-dfm-1-sol}
The order $\delta^1$ equation \eqref{eq:multi-dfm-1} is solvable if and only if the pair ${(E^{(1)}, X \mapsto b^{(0)}(X; \kappa))}$, where ${b^{(0)}(X; \kappa) = [b^{(0)}_1(X; \kappa), b^{(0)}_2(X; \kappa)]^\mathsf{T})}$, is a solution of
\begin{equation}
\label{eq:dfm-edge-eff-evp}
\cancel{\mathfrak{D}}^+(\kappa) b^{(0)}(X; \kappa) = E^{(1)} b^{(0)}(X; \kappa) .
\end{equation}
\end{proposition}

\noindent Proposition \ref{prop:multi-dfm-1-sol} is proven in Appendix \ref{apx:pf_multi-dfm}. We have shown that the leading order term of an asymptotic solution of the edge state eigenvalue problem \eqref{eq:dfm-edge-evp} near ${(E, k) = (E_\star, k^+_\star)}$ arises from the bound states (i.e., the $L^2(\R; \C^2)$ eigenstates) of $\cancel{\mathfrak{D}}(\kappa)$. The proof of Theorem \ref{thm:multi-dfm} is complete. \qed

\subsection{Remarks on effective Hamiltonians}
\label{sec:multi-rmk}

\begin{remark}
\label{rmk:bulk-edge-eff}
{\rm (Relation between effective bulk and edge Hamiltonians.)}
For ${|X| \to +\infty}$, the effective edge Hamiltonians $\mathfrak{S}(\kappa)$ in \eqref{eq:sql-edge-eff} and $\cancel{\mathfrak{D}}^\pm(\kappa)$ in \eqref{eq:dfm-edge-eff} are related to the perturbed effective bulk Hamiltonians $\smash{H_{\rm eff}^{\bm M}({\bm \kappa}; \delta)}$ in \eqref{eq:quad-eff-breakC} and $\smash{H_{\rm eff}^{{\bm D}^\pm}({\bm \kappa}; \delta)}$ in \eqref{eq:dir-eff-breakC}, respectively: ${\bm \kappa}$ is replaced with $\delta^r {\bm P}_X(\kappa)$, and an overall factor of $\delta^r$ is divided out. (Note that their Fourier transforms are exactly equal, up to a factor of $\delta^r$.)
\end{remark}

\begin{remark}
\label{rmk:alt-edge-eff}
{\rm (Invariance with respect to edge parameterization.)}
Consider the effective edge Hamiltonians arising from the multiple-scale procedures of Sections \ref{sec:multi-sql} - \ref{sec:multi-dfm}, analogous to $\mathfrak{S}(\kappa)$ and $\cancel{\mathfrak{D}}^\pm(\kappa)$, but obtained via the alternate edge parameterization ${(\boldsymbol{\mathfrak{v}}_1, \boldsymbol{\mathfrak{v}}_2 + j \boldsymbol{\mathfrak{v}}_1)}$, ${(\boldsymbol{\mathfrak{K}}_1 - j \boldsymbol{\mathfrak{K}}_2, \boldsymbol{\mathfrak{K}}_2)}$, with ${j \smallin \Z}$; see Remark \ref{rmk:alt-edge}. Their expressions are identical, but replacing ${\bm P}_X(\kappa)$ with
\begin{equation}
\label{eq:def-P-j}
{\bm P}_X(\kappa; j) \equiv (-i \partial_X - j \kappa) \boldsymbol{\mathfrak{K}}_2 + \kappa \boldsymbol{\mathfrak{K}}_1 = e^{i j \kappa X} {\bm P}_X(\kappa) e^{-i j \kappa X} .
\end{equation}
Hence, effective edge Hamiltonians obtained via equivalent edge parameterizations are unitarily equivalent.
\end{remark}

\subsection{Formal multiple-scale analysis and rigorous proofs}

\subsubsection{From approximate to exact edge states; spectral no-fold condition}
\label{sec:multi-exact}

Our multiple-scale construction of approximate solutions of the edge state eigenvalue problems for $\smash{H^\delta_{\rm edge} = H^\delta}$ (Section \ref{sec:multi-sql}) and $\smash{H^\delta_{\rm edge} = T_* H^\delta}$ (Section \ref{sec:multi-dfm}) can be continued to arbitrary finite order $N$ in $\delta$. For each fixed ${k \smallin [-\pi, \, \pi]}$, this yields an approximate eigenpair, or ``quasimode'' $\smash{(E^\delta_N, \Psi^\delta_N)}$ such that $\smash{\lVert \Psi^\delta_N \rVert = 1}$, where ${\lVert \cdot \rVert}$ denotes the norm on $L^2_k(\R^2/\Z\boldsymbol{\mathfrak{v}}_1)$, and 
\begin{equation}
\label{eq:quasimode}
\lVert (H^\delta_{\rm edge} - E_N^\delta) \Psi^\delta_N \rVert = O(\delta^{N}) \ \ \text{as} \ \ \delta \to 0 .
\end{equation}
Proving that ${(E^\delta_N, \Psi^\delta_N)}$ is an approximation of a true edge state eigenpair ${(E^\delta, \Psi^\delta)}$, to some order $\delta^N$, requires order $\delta^N$ upper bounds on the remainder ${(\Psi^\delta -\Psi^\delta_N, E^\delta - E^\delta_N)}$. Two strategies that can be implemented are:
\begin{enumerate}
\renewcommand{\theenumi}{\alph{enumi}}
\item Direct solution of the equation for the remainder and a bound on its norm; see, e.g., \cite{FLW-PNAS:14,FLW-MAMS:17,FLW-2d_edge:16}.
\item Obtaining a resolvent expansion for $\smash{H^\delta_{\rm edge}}$ for energies in the band gap, and constructing a quasimode whose energy $E^\delta_N$ is sufficiently isolated from any other eigenvalues within the spectral gap; see, e.g., \cite[Lemma 3.1]{drouot2020defect}, \cite{drouotweinstein2020}, \cite[Appendix C]{CW21}.
\end{enumerate}
An edge state with quasimomentum $k$ corresponds to an eigenvalue of $\smash{H^\delta_{\rm edge}}$ within a gap in its $L^2_k(\R^2/\Z\boldsymbol{\mathfrak{v}}_1)$ essential spectrum. Yet, the multiple-scale expansion defining $\Psi^\delta_N$, which involves only band structure information at the bulk band degeneracy ${(E_\star, {\bm k}_\star)}$, requires only a {\rm local} band gap. A true spectral gap is ensured by an additional ``spectral no-fold condition'' on the band structure of the unperturbed bulk operator $\smash{H^0_{\rm edge} = H_{\rm bulk}}$, which must be verified separately. This condition states, roughly, that the dispersion surfaces do not fold over the energy $E_\star$ away from the degeneracy at ${(E_\star, {\bm k}_\star)}$ in the direction of $\boldsymbol{\mathfrak{K}}_2$. Finally, strategy (b) above enables one to conclude that all eigenvalues of $H^\delta_{\rm edge}$ within the band gap arise from multiple-scale expansions, seeded by the discrete spectra of effective edge Hamiltonians. For a full implementation of these strategies in the context of 2D honeycomb media and 1D dislocations, see \cite{FLW-PNAS:14, FLW-MAMS:17, FLW-2d_edge:16, drouot2020defect, drouotweinstein2020}.

\subsubsection{Multiple-scale analysis and the bulk-edge correspondence principle}
\label{sec:multi-bec}

As remarked in Section \ref{sec:summary}, multiple-scale analysis reduces the computation of bulk and edge topological indices to corresponding calculations for effective Hamiltonians. Indeed, for $\delta$ small, the Berry curvature is concentrated near the bulk band structure degeneracies, which have been lifted by the symmetry-breaking term in $\smash{H^{\pm, \, \delta}_{\rm bulk}}$. These neighborhoods provide the dominant contribution to the integral expression for the first Chern number. Since the spectral properties of $\smash{H^{\delta,\pm}_{\rm bulk}}$, and hence the Berry curvature, are accurately approximated by the relevant effective bulk Hamiltonian, the Chern number calculations are reduced to (and simplified by) calculations involving the effective bulk Hamiltonians, which are constant-coefficient partial differential operators. Similarly, the spectral flow/quantum current associated with $H^\delta_{\rm edge}$ can be reduced to the spectral flow for the effective edge Hamiltonians. For implementations of this strategy, see, e.g., \cite{drouot2019a,drouot2020defect,drouotweinstein2020}.

\smallskip

\section{Spectral analysis of effective edge Hamiltonians $\mathfrak{S}(\kappa)$ and $\cancel{\mathfrak{D}}^\pm(\kappa)$}
\label{sec:spec}

\setcounter{equation}{0}
\setcounter{figure}{0}

In this section, we present numerical simulations of the  spectra of our effective edge Hamiltonians $\mathfrak{S}(\kappa)$ and $\cancel{\mathfrak{D}}^\pm(\kappa)$, where ${\kappa \in \R}$. Further, we study the continuous spectra of $\mathfrak{S}(\kappa)$ and $\cancel{\mathfrak{D}}^\pm(\kappa)$, and the discrete spectra of $\mathfrak{S}(0)$ and $\cancel{\mathfrak{D}}^\pm(0)$ using PDE and spectral theoretic methods. Our results are consistent with topological arguments \cite{bal-cpde-2022, bal-jmp-2023, quinn-ball-sima-2024} implying that the spectral flow (or algebraic count of gap-traversing curves) is equal to ${\pm 2}$. In this section we do not make use of arguments rooted in the topology of the space of Fredholm operators.

\subsection{Spectrum of matrix Schr\"{o}dinger operator $\mathfrak{S}(\kappa)$}
\label{sec:spec-sql}

We focus on a simple choice of rational edge whose direction is parallel to an axis of reflection symmetry. \\

\noindent {\bf Notation.} Throughout this section, we suppress ${\bm M}$ superscripts on all parameters. \\

Assume the edge direction ${\boldsymbol{\mathfrak{v}}_1 = [0, 1]^\mathsf{T}}$. From Example \ref{ex:sql-edge-eff-vert}, $\mathfrak{S}(\kappa)$ is given by 
\begin{equation}
\label{eq:sql-edge-eff_2}
\mathfrak{S}(\kappa) = (1 - \alpha_0) (P_X^2 + \kappa^2)\ \sigma_0 - 2 \alpha_1 P_X \kappa\ \sigma_1 - \alpha_2 (P_X^2 -\kappa^2)\ \sigma_2 + \vartheta \chi(X)\ \sigma_3 .
\end{equation}
Here, ${\kappa \smallin \R}$, ${P_X = -i \partial_X}$, and  ${X \mapsto \chi(X)}$ is a domain wall function; see Definition \ref{def:dwall}. The parameters $\alpha_\ell$, ${0 \leq \ell \leq 2}$, are given in \eqref{eq:def-a-par}. The non-degeneracy assumption \ref{itm:quad-dgn-5} implies $\alpha_1$, ${\alpha_2 \neq 0}$.

\begin{figure}[t]
\centering
\includegraphics[scale = 0.6]{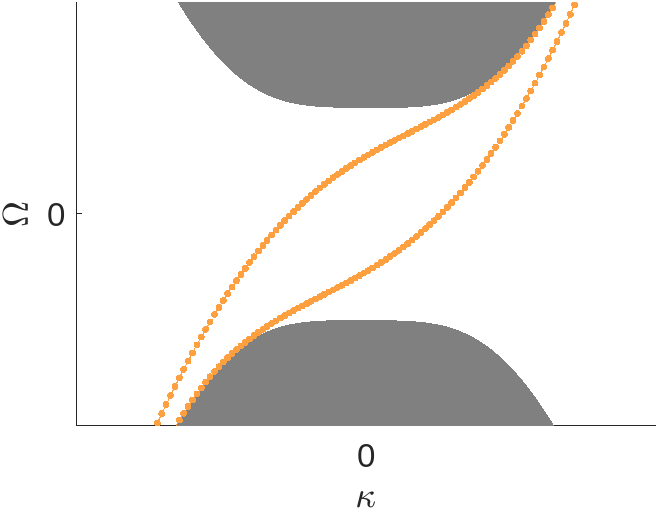}
\caption{Numerically generated spectra of the family ${\kappa \mapsto \mathfrak{S}(\kappa)}$ \eqref{eq:sql-edge-eff_2} with $\alpha_0$, $\alpha_1$, $\alpha_2$, ${\vartheta = 1}$, ${\chi(X) = \tanh(X)}$. The union of essential spectra {\it (gray)} consists of two connected regions, or bands, separated by a band gap. Two \nolinebreak eigenvalue curves {\it (orange)} traverse the gap; the spectral flow is ${+2}$.}
\label{fig:sql-eff-spec}
\end{figure}
 
Figure \ref{fig:sql-eff-spec} displays the numerically computed spectra of the family ${\kappa \mapsto \mathfrak{S}(\kappa)}$. For each fixed ${\kappa = \kappa_0}$, the spectrum of $\mathfrak{S}(\kappa_0)$ is the colored subset of the slice at ${\kappa = \kappa_0}$. It consists of two semi-infinite intervals of continuous spectrum ${(-\infty, \, \eta_-(\kappa_0)] \cup [\eta_+(\kappa_0), \, +\infty)}$ (gray), together with two simple eigenvalues $\Omega_+(\kappa_0)$ and $\Omega_-(\kappa_0)$ (orange), with
\begin{equation}
\label{eq:gap-ineq}
\eta_-(\kappa_0) < \Omega_-(\kappa_0) < \Omega_+(\kappa_0) < \eta_+(\kappa_0) .
\end{equation}
Figure \ref{fig:sql-eff-spec} shows two eigenvalue curves ${\kappa \mapsto \Omega_\pm(\kappa)}$ (orange) which traverse the band gap. This is consistent with the discussion in Section \ref{sec:sql-edge}, based on the bulk-edge correspondence principle. Further, our asymptotic analysis is in very good agreement with numerical simulations, presented in Section \ref{sec:numerics}, of the full $L^2_k(\R^2/\Z \boldsymbol{\mathfrak{v}}_1)$ spectrum of $\smash{H^\delta_{\rm edge} = H^\delta}$ near ${(E_\star, k_\star)}$.

Note that the multiple-scale analysis of Section \ref{sec:multi-sql} applies for  ${X \mapsto \chi(X)}$ belonging to a much broader class than domain wall functions. In particular:

\begin{remark}
\label{rmk:chi-cases}
{\rm (Topologically trivial edge media.)}
Consider a topologically trivial case, i.e., an edge medium which interpolates between two periodic bulk media with {\rm equal}  topological indices. One can arrange this by taking ${X \mapsto \chi(X)}$ non-constant and such that ${\chi(X) \to +1}$ as ${X \to \pm \infty}$, so that ${H^{+, \, \delta}_{\rm bulk} = H^{-, \, \delta}_{\rm bulk}}$. 
\begin{itemize}
\item Figure \ref{fig:sql-eff-spec-conj-4} displays numerical simulations of the spectrum of $\mathfrak{S}(\kappa)$ in such a case. In contrast to the domain wall case, we observe two eigenvalue curves which do not traverse the band gap; the spectral flow is equal to zero, as anticipated by the bulk-edge correspondence principle.

\item We have also numerically investigated the case where ${X \mapsto \chi(X)}$ of this scenario is perturbed by a large, localized function of $X$ with the important property that ${\int_\R (\chi(X)^2 - 1) \, {\rm d}X > 0}$.  In these simulations, there are no eigenvalue curves in the band gap; they have been deformed away into the bands.

\item Other numerical simulations show multiple eigenvalue curves within, but not traversing the band gap; again, the spectral flow is zero.
\end{itemize}
See further discussion below of the spectra of $\mathfrak{S}(\kappa)$, contrasting the topologically distinct cases: ${\chi(X) \to \pm 1}$ and ${\chi(X) \to +1}$ as ${X \to \pm \infty}$.
\end{remark}

\noindent We next present analytical results in support of the above observations.

\subsubsection{Essential spectrum of $\mathfrak{S}(\kappa)$}
\label{sec:spec-sql-ess}

That the essential spectrum of $\mathfrak{S}(\kappa)$, ${\kappa \smallin \R}$, is as displayed in Figure \ref{fig:sql-eff-spec} is a consequence of:

\begin{proposition}
\label{prop:spec-sql-ess}
{\rm (Essential spectrum of $\mathfrak{S}(\kappa)$.)}
Consider $\mathfrak{S}(\kappa)$ \eqref{eq:sql-edge-eff_2} with ${\alpha_0 = 1}$ and $\alpha_1$, $\alpha_2$, ${\vartheta \neq 0}$. Assume ${\chi^2(X) \to 1}$ sufficiently rapidly as ${|X| \to +\infty}$. Then, for each fixed ${\kappa \smallin \R}$, the essential spectrum of $\mathfrak{S}(\kappa)$ consists of two semi-infinite intervals:
\begin{equation}
\label{eqn:specwindows}
{\rm spec}_{\rm ess}(\mathfrak{S}(\kappa)) = (-\infty, \, \eta_-(\kappa)] \cup [\eta_+(\kappa), \, +\infty) ,
\end{equation}
separated by an open interval ${(\eta_-(\kappa), \, \eta_+(\kappa))}$ containing ${(-|\vartheta|, \, |\vartheta|)}$. Consequently,
\begin{equation}
\inf_{\kappa \smallin \R} \, \eta_+(\kappa) - \sup_{\kappa \smallin \R} \, \eta_-(\kappa) \geq 2 |\theta| > 0 ,
\end{equation}
implying ${\int_\R^\oplus \mathfrak{S}(\kappa) \, {\rm d}\kappa}$ has a spectral gap.  
\end{proposition}

\begin{remark}
\label{rmk:sql-edge-eff-gap}
{\rm (Band gaps, spectral gaps for ${\alpha_0 \neq 1}$.)}
We conjecture, under the relaxed hypothesis
\begin{equation}
\label{eq:a-par-gap}
-(1 - \alpha_0)^2 + \alpha_2^2 > 0 ,
\end{equation}
that ${\kappa \mapsto \mathfrak{S}(\kappa)}$ has a band gap, corresponding to the condition
\begin{equation}
\inf_{\kappa \smallin \R} (\eta_+(\kappa) - \eta_-(\kappa)) > 0 .
\end{equation}
It can be easily proven that \eqref{eq:a-par-gap} is necessary for the existence of a band gap; we have verified sufficiency in several parameter regimes. More generally, the validity of this property is supported by numerical simulations.

On the other hand, \eqref{eq:a-par-gap} does not generally imply the stronger spectral gap property: 
\begin{equation}
\inf_{\kappa \smallin \R} \, \eta_+(\kappa) - \sup_{\kappa \smallin \R} \, \eta_-(\kappa) > 0 ;
\end{equation}
see Figure \ref{f:71kappa}.
\end{remark}

\noindent For the proof of Proposition \ref{prop:spec-sql-ess}, see Appendix \ref{apx:spec-sql-ess}. For further discussion of the generic ${\alpha_0 \neq 1}$ case, see \cite{chaban2025thesis}.

\begin{figure}[t]
    \centering
    \hspace{0.3cm}
    \begin{subfigure}{0.32\textwidth}
        \centering
        \subcaption{}
        \label{fig:sql-eff-spec-conj-2}
        \vspace{0.1cm}
        \hspace{-0.8cm} \includegraphics[scale = 0.5]{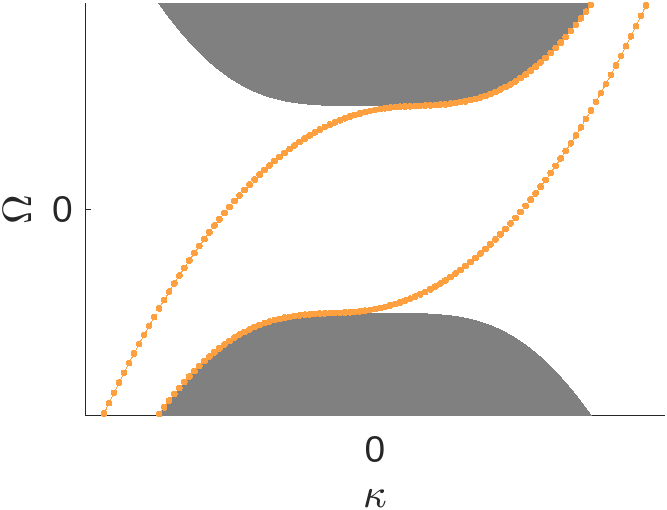}
    \end{subfigure}
    \begin{subfigure}{0.32\textwidth}
        \centering
        \subcaption{}
        \label{fig:sql-eff-spec-conj-3}
        \vspace{0.1cm}
        \hspace{-0.8cm} \includegraphics[scale = 0.5]{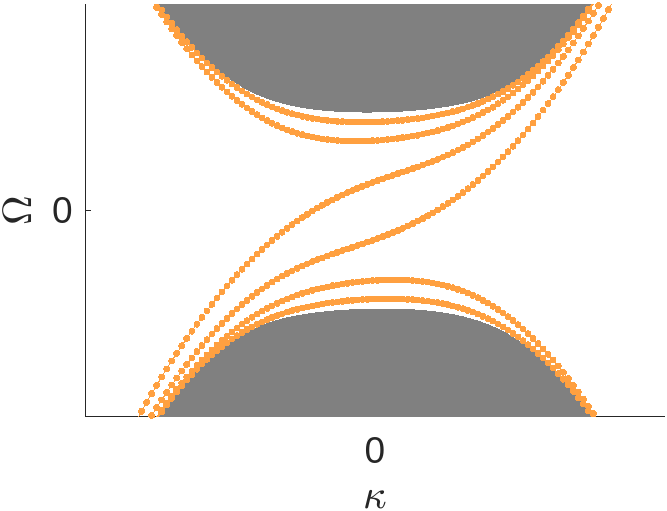}
    \end{subfigure}
    \begin{subfigure}{0.32\textwidth}
        \centering
        \subcaption{}
        \label{fig:sql-eff-spec-conj-4}
        \vspace{0.1cm}
        \hspace{-0.8cm} \includegraphics[scale = 0.5]{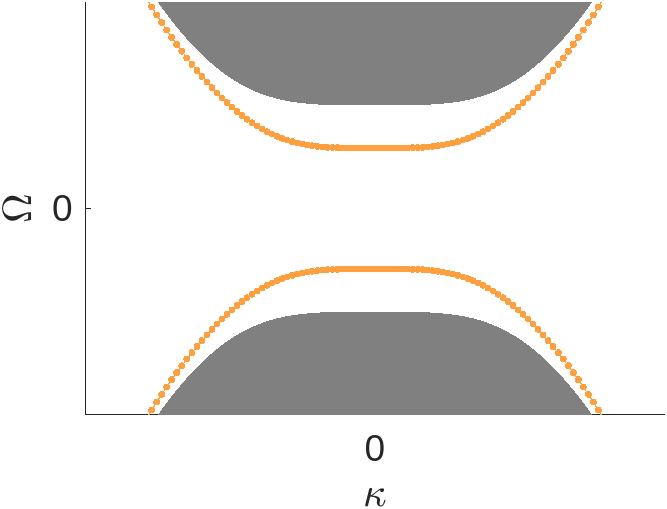}
    \end{subfigure}
\caption{Numerically generated spectra of the families ${\kappa \mapsto \mathfrak{S}(\kappa)}$ \eqref{eq:sql-edge-eff_2} with $\alpha_0$, $\alpha_1$, $\alpha_2$, ${\vartheta = 1}$, $\chi(X)$ given \nolinebreak by: {\bf (a)} ${\chi(X) = \tanh(X) + 50 \exp(-X^2)}$; the sign condition \eqref{eq:dw-varcond} is broken, but the two gap-traversing eigenvalue curves present in Figure \ref{fig:sql-eff-spec} persist. {\bf (b)} ${\chi(X) = \tanh(X/4)}$; the two gap-traversing eigenvalue curves persist, two additional pairs of non-traversing curves emerge. {\bf (c)} ${\chi(X) = 1 - \exp(-X^2)}$ (which, since ${\chi(-\infty) \cdot \chi(+\infty) > 0}$, is not a domain wall function); there are two non-traversing eigenvalue curves, but no gap-traversing curves are present.}
\label{fig:sql-eff-spec-conj}
\end{figure}

\subsubsection{Discrete spectrum of $\mathfrak{S}(0)$}
\label{sec:spec-sql-disc}

We now turn to a discussion of the discrete spectrum. Numerical simulations suggest that, if the family of effective edge Hamiltonians ${\kappa \mapsto \mathfrak{S}(\kappa)}$ has a band gap, then there are exactly two gap-traversing eigenvalue curves;
see Figure \ref{fig:sql-eff-spec}. 
These results are consistent with the bulk edge correspondence principle, based on topological arguments, discussed earlier.
In this section, we initiate a direct analytical study.

For simplicity, we shall restrict our attention to the discrete spectrum of $\mathfrak{S}(\kappa)$, defined in \eqref{eq:sql-edge-eff_2}, at ${\kappa = 0}$. We assume ${\alpha_0 = 1}$ and $\alpha_1$, $\alpha_2$, ${\vartheta \neq 0}$, which yields a spectral gap by Proposition \ref{prop:spec-sql-ess}. Hence, we consider  
\begin{align}
\label{eq:sql-edge-eff_2a}
\mathfrak{S}(0) = -\alpha_2 P_X^2 \sigma_2 + \vartheta \chi(X) \sigma_3 .
\end{align}
A short calculation shows that $\mathfrak{S}(0)$ has a gap in its essential spectrum given by ${(-|\vartheta|, \, |\vartheta|)}$. Our first result, proven in Appendix \ref{apx:sql-eff-no-zero}, rules out the possibility of zero-energy bound states.

\begin{proposition}
\label{prop:sql-eff-no-zero}
{\rm (No zero-energy bound states.)}
${\Omega = 0}$ is not an eigenvalue of $\mathfrak{S}(0)$.
\end{proposition}

We turn to the existence of nonzero eigenvalues in the spectral gap of $\mathfrak{S}(0)$. By topological arguments cited earlier, there are at least two eigenvalues in its spectral gap. It would be of interest to obtain this result, as well as detailed information about the eigenstates and eigenvalue curves by direct analytical methods; Theorems \ref{thm:sql-eff-eig_1} and \ref{thm:sql-eff-eig_2}, and an explicitly solvable case (Supplementary Material \ref{supp:sql-eff-exact}) are steps in this direction.

\begin{theorem}
\label{thm:sql-eff-eig_1}
{\rm (Result 1: Two eigenvalues in the spectral gap of $\mathfrak{S}(0)$ under a sign condition.)}
Consider $\mathfrak{S}(0)$ \eqref{eq:sql-edge-eff_2a} acting in $L^2(\R; \C^2)$ where $\alpha_1$, $\alpha_2$, ${\vartheta \neq 0}$. Assume that ${X \mapsto \chi(X)}$ is smooth with ${\chi(X)^2 \to 1}$ sufficiently rapidly as ${|X| \to +\infty}$. Further, assume
\begin{equation}
\label{eq:dw-varcond}
\int_\R (\chi(X)^2 - 1) \, {\rm d}X < 0 .
\end{equation}
Then, $\mathfrak{S}(0)$ has at least two nonzero eigenvalues $\Omega_+ > 0$ and $\Omega_- = -\Omega_+$ within its spectral gap.
\end{theorem}

\noindent The proof of Theorem \ref{thm:sql-eff-eig_1} is located in Appendix \ref{apx:sql-eff-eig_1}. The next result allows for domain wall functions with a sufficiently steep transition, extending our results to a class of domain walls which do not satisfy the sign condition \eqref{eq:dw-varcond}. 

\begin{theorem}
\label{thm:sql-eff-eig_2}
{\rm (Result 2: Two eigenvalues in the spectral gap of $\mathfrak{S}(0)$ for steep domain walls.)}
There are two scenarios where we can prove the existence of at least two eigenvalues when \eqref{eq:dw-varcond} is violated:
\begin{enumerate}
\item \label{itm:sql-eff-eig_2-1} {\rm (Steep domain wall transitions.)} Assume $\alpha_2$, ${\vartheta \neq 0}$, and let ${X \mapsto \chi(X)}$ be a domain wall function; see Definition \ref{def:dwall}. For ${\varepsilon > 0}$, define
\begin{equation}
\mathfrak{S}^\varepsilon(0) = -\alpha_2 P_X^2 \sigma_2 + \vartheta \chi_\varepsilon(X) \sigma_3 , \quad \text{where} \quad \chi_\varepsilon(X) \equiv \chi \biggl( \frac{X}{\varepsilon} \biggr) .
\end{equation}
Then, for ${\varepsilon > 0}$ sufficiently small, $\mathfrak{S}^\varepsilon(0)$ has at least two nonzero eigenvalues ${\Omega^\varepsilon_{+1} > 0}$ and ${\Omega^\varepsilon_{-1} = -\Omega^\varepsilon_{+1}}$ within its spectral gap.

\item \label{itm:sql-eff-eig_2-2} {\rm (Explicitly solvable model: ${\chi(X) = \sgn(X)}$.)} Take $\mathfrak{S}(0)$ \eqref{eq:sql-edge-eff_2a} with ${\chi(X) = {\rm sgn}(X)}$. One can solve the constant-coefficient equations for ${X > 0}$ and ${X < 0}$ and impose matching conditions across ${X = 0}$ to obtain two explicit bound state eigenpairs; see Figure \ref{fig:sql-eff-exact}. 
\end{enumerate}
\end{theorem}

\noindent The proof of Theorem \ref{thm:sql-eff-eig_2} is located in Appendix \ref{apx:sql-eff-eig_2}. Detailed calculations for the explicitly solvable case are presented in Supplementary Material \ref{supp:sql-eff-exact}. Part \ref{itm:sql-eff-eig_2-2} of Theorem \ref{thm:sql-eff-eig_2} shows that the result of part \ref{itm:sql-eff-eig_2-1} extends to the limit ${\varepsilon \to 0}$.

\begin{figure}[!t]
\centering
\includegraphics[scale = 0.55]{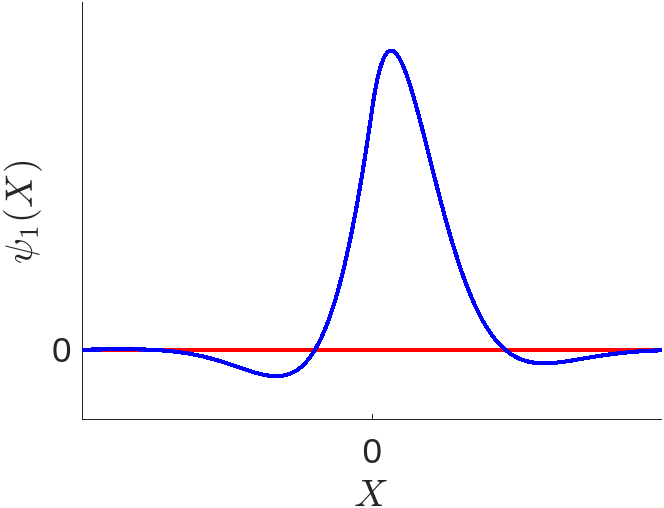}
\includegraphics[scale = 0.55]{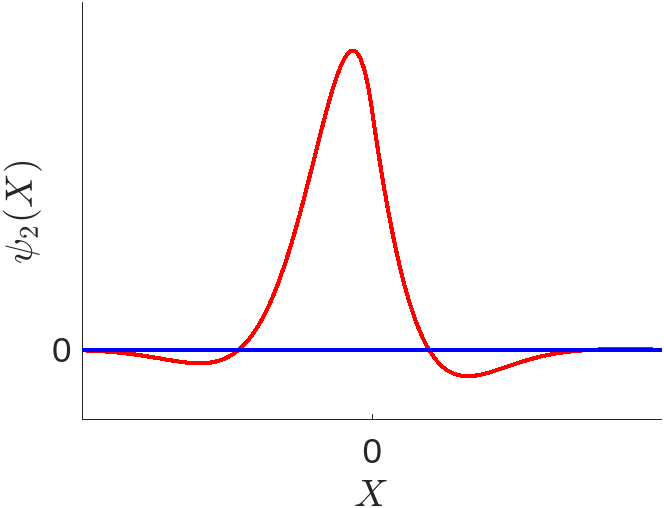}
\caption{Real {\it (red)} and imaginary {\it (blue)} parts of first {\it (left)} and second {\it (right)} components of the bound state of $\mathfrak{S}(0)$ \eqref{eq:sql-edge-eff_2a} with eigenvalue ${\Omega_+ = 1/\sqrt{2}}$. Here, $\alpha_2$, ${\vartheta = 1}$, ${\chi(X) = {\rm sgn}(X)}$; see Theorem \ref{thm:sql-eff-eig_2}, part \ref{itm:sql-eff-eig_2-2}. The bound state with eigenvalue ${\Omega_- = -1/\sqrt{2}}$ {\it (not pictured)} is obtained by interchanging the components.}
\label{fig:sql-eff-exact}
\end{figure}

The proofs of Theorems \ref{thm:sql-eff-eig_1} and \ref{thm:sql-eff-eig_2} are based on the following: Note that the squared operator $\mathfrak{S}(0)^2$ is non-negative, and hence bounded below. We define
\begin{align}
\label{eq:def-calL}
\mathcal{L} = \mathfrak{S}(0)^2 - \vartheta^2 & = (\alpha_2^2 P_X^4 + \vartheta^2 (\chi(X)^2 - 1)) I - i \alpha_2 \vartheta (P^2_X \chi(X) - \chi(X) P^2_X) \sigma_1 \\
& = (\alpha_2^2 P_X^4 + \vartheta^2 (\chi(X)^2 - 1)) I - \alpha_2 \vartheta (P_X \chi'(X) + \chi'(X) P_X) \sigma_1 , \nonumber
\end{align}
with domain $H^4(\R; \C^2)$. For ${\chi(X)^2 \to 1}$ sufficiently rapidly as ${|X| \to +\infty}$, the essential spectrum of $\mathcal{L}$ is
\begin{equation}
{\rm spec}_{\rm ess}(\mathcal{L}) = {\rm spec}(P_X^4) = [0, \, +\infty) .
\end{equation}
Since $\mathfrak{S}(0)^2$ is non-negative, the spectrum of $\mathcal{L}$ is a subset of ${[-\vartheta^2, \, +\infty)}$; the gap in the essential spectrum of $\mathfrak{S}(0)$ is mapped to ${[-\vartheta^2, \, 0)}$. Hence, $\mathcal{L}$ has discrete spectrum if we can produce a function in its domain for which the associated quadratic form is strictly negative. Further, an element of the discrete spectrum of $\mathcal{L}$ gives rise to two eigenvalues of $\mathfrak{S}(0)$. In summary:

\begin{proposition}
\label{prop:sql-eff-exist}
{\rm (Condition for the existence of two eigenvalues of $\mathfrak{S}(0)$.)}
The operator $\mathfrak{S}(0)$ \eqref{eq:sql-edge-eff_2a} has a pair of eigenvalues ${\Omega_+ > 0}$ and ${\Omega_- = - \Omega_+}$ in its spectral gap, i.e.,
\begin{equation}
-|\vartheta| < \Omega_- = -\Omega_+ < 0 < \Omega_+ < |\vartheta| ,
\end{equation}
if and only if there exists ${\psi \smallin L^2(\R; \C^2)}$ such that ${\langle \psi, \, \mathcal{L} \psi \rangle < 0}$, where ${\mathcal{L} = \mathfrak{S}(0)^2 - \vartheta^2}$ \eqref{eq:def-calL}.
\end{proposition}

\noindent The proof of Proposition \ref{prop:sql-eff-exist} is located in Appendix \ref{apx:sql-eff-exist}.

\subsubsection{Eigenvalue curves of ${\kappa \mapsto \mathfrak{S}(\kappa)}$}

Perturbation theory of self-adjoint linear operators \cite{kato1995perturbation} ensures that, for $|\kappa|$ sufficiently small, there exist real analytic eigenvalue curves in the band gap of ${\kappa \mapsto \mathfrak{S}(\kappa)}$. In the case where $\Omega_\pm$ are simple eigenvalues, there are two curves: ${\kappa \mapsto \Omega_\pm(\kappa)}$, with ${\Omega_\pm(0) = \Omega_\pm}$. The slope of each eigenvalue curve at ${\kappa = 0}$ is obtained from the expansion of $\Omega_\pm(\kappa)$ to first order. We contrast the cases: (1) ${\chi(X) \to \pm 1}$, and (2) ${\chi(X) \to +1}$ as ${X \to \pm \infty}$.

\begin{proposition}
\label{prop:sql-eig-local}
Consider $\mathfrak{S}(\kappa)$, defined in \eqref{eq:sql-edge-eff_2}, with ${\alpha_0 = 1}$ and $\alpha_2$, ${\vartheta \neq 0}$. Then:
\begin{enumerate}
\item \label{itm:sql-eig-local_1} If ${\chi(X) \to \pm 1}$ as ${X \to \pm \infty}$ and ${\chi(-X) = -\chi(X)}$,
\begin{equation}
\Omega'_-(0) = \Omega'_+(0) =  2 \alpha_1 \langle \psi_+, \, P_X \sigma_1 \psi_+ \rangle .
\end{equation}

\item \label{itm:sql-eig-local_2} If ${\chi(X) \to +1}$ as ${X \to \pm \infty}$ and ${\chi(-X) = \chi(X)}$,
\begin{equation}
\Omega'_+(0) = 0 \quad \text{and} \quad \Omega'_-(0) = 0 .
\end{equation}
\end{enumerate}
\end{proposition}

\noindent The proof of Proposition \ref{prop:sql-eig-local} is located in Supplementary Material \ref{apx:sql-eig-local}. This result is consistent with  Figure \ref{fig:sql-eff-spec}, Figure \ref{fig:sql-eff-spec-conj-2} - \ref{fig:sql-eff-spec-conj-3}, and Figure \ref{fig:sql-eff-spec-conj-4}, respectively.

\subsection{Spectra of Dirac operators $\cancel{\mathfrak{D}}^\pm(\kappa)$}
\label{sec:spec-dfm}

We next discuss the effective edge Hamiltonians $\cancel{\mathfrak{D}}^\pm(\kappa)$ derived in Section \ref{sec:multi-dfm}. Let us define parameters $a_\ell$, ${b_\ell \smallin \R}$, ${0 \leq \ell \leq 2}$, given by the edge-dependent expressions
\begin{equation}
\label{eqn:abdefs}
a_\ell = {\bm \gamma}_\ell^{{\bm D}^+} \! \cdot \boldsymbol{\mathfrak{K}}_2 \quad \text{and} \quad b_\ell = {\bm \gamma}_\ell^{{\bm D}^+} \! \cdot \boldsymbol{\mathfrak{K}}_1 , \quad 0 \leq \ell \leq 2 ,
\end{equation}
where ${\bm \gamma}^{{\bm D}^\pm}_\ell$, ${0 \leq \ell \leq 2}$, are given in \eqref{eq:def-gam-par}. The parameter ${c = \vartheta^{{\bm D}^+}}$ is real and non-zero, and independent of the choice of edge. For notational convenience, we introduce the triples ${a \equiv (a_0, a_1, a_2)}$ and ${b \equiv (b_0, b_1, b_2)}$, and denote ${\sigma \equiv [\sigma_0, \sigma_1, \sigma_2]^\mathsf{T}}$; see \eqref{eq:def-pauli}. We write
\begin{equation}
\label{eqn:sigma-dot}
u \cdot \sigma = u_0 \sigma_0 + u_1 \sigma_1 + u_2 \sigma_2 ,
\end{equation} 
defined for any ${u = (u_0, u_1, u_2) \smallin \R^3}$.

\begin{remark}
\label{rmk:gamma}
Since ${{\bm \gamma}^{{\bm D}^+}_1 \!}$, ${{\bm \gamma}^{{\bm D}^+}_2 \! \smallin \R^2}$ are nonzero (see the non-degeneracy condition \ref{itm:dir-pt-4}), and ${\{ \boldsymbol{\mathfrak{K}_1}, \boldsymbol{\mathfrak{K}_2} \}}$ spans $\R^2$, it follows that $a_1$ and $b_1$, and similarly $a_2$ and $b_2$, cannot be simultaneously zero.
\end{remark}

Our class of Dirac operators $\cancel{\mathfrak{D}}^\pm(\kappa)$ is of the general form:
\begin{align}
\label{eq:dfm-edge-eff_2}
\cancel{\mathfrak{D}}^\pm(\kappa) = \pm \bigl( (a \cdot \sigma) P_X + (b \cdot \sigma) \kappa \bigr)  + c \chi(X) \sigma_3 ,
\end{align}
where ${\kappa \smallin \R}$, ${P_X = -i \partial_X}$, and $\chi(X)$ is a domain wall function; see Definition \ref{def:dwall}.

\begin{remark}
\label{rmk:honey-zigzag}
Setting $a_0$, $b_0$, $b_1$, and ${a_2 = 0}$ yields
\begin{equation}
\cancel{\mathfrak{D}}^\pm(\kappa) = \pm (a_1 P_X \sigma_1 + b_2 \kappa \sigma_2) + c \chi(X) \sigma_3 .
\end{equation}
For appropriate parameters $a_1$, $b_2$, and $c$, this is the Dirac operator obtained as an effective edge Hamiltonian in \cite[Proposition 4.3]{drouotweinstein2020} in the context of 2D honeycomb media.
\end{remark}

\begin{figure}[!t]
\centering
\includegraphics[scale = 0.55]{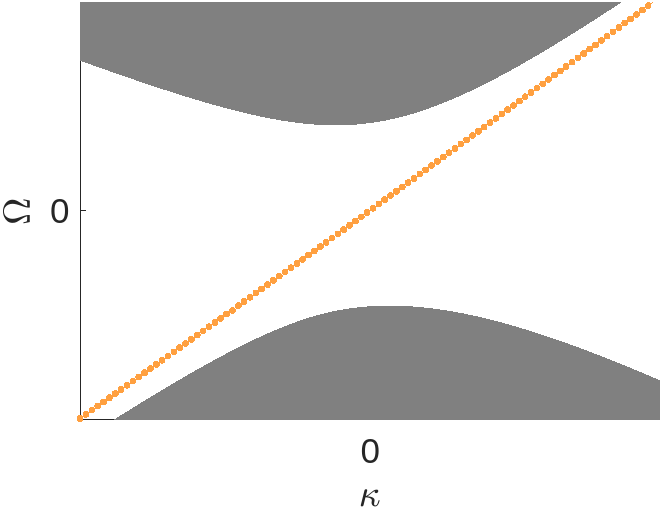} 
\includegraphics[scale = 0.55]{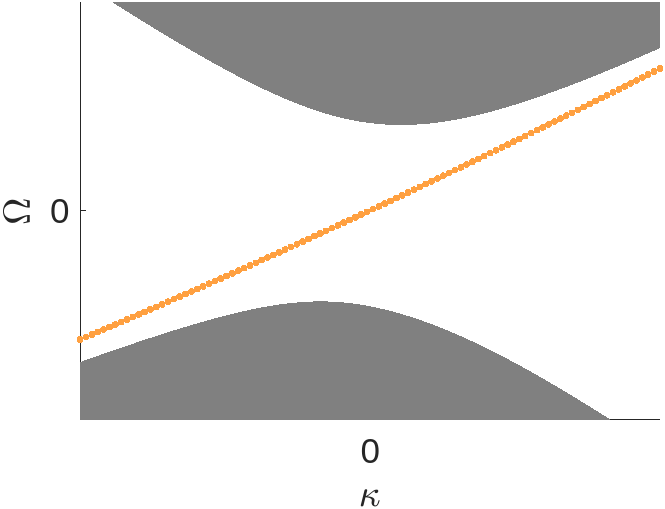} 
\caption{Numerically generated spectra of the families ${\kappa \mapsto \cancel{\mathfrak{D}}^-(\kappa)}$ {\it (left)} and ${\kappa \mapsto \cancel{\mathfrak{D}}^-(\kappa)}$ {\it (right)} \eqref{eq:dfm-edge-eff_2} with ${a_0 = 0.3}$, $a_1$, $a_2$, ${b_2 = 1}$, ${c = 0.5}$, ${\chi(X) = \tanh(X)}$; all other parameters are zero. In each plot, the union of essential spectra {\it (gray)} consists of two connected regions (i.e., bands) separated by a band gap. One eigenvalue curve {\it (orange)} traverses the gap, contributing $+1$ to the spectral flow; summing both contributions, the spectral flow is $+2$.}
\label{fig:effDirac-spec}
\end{figure}

Figure \ref{fig:effDirac-spec} displays computed spectra of the families ${\kappa \mapsto \cancel{\mathfrak{D}}^\pm(\kappa)}$, analogous to Figure \ref{fig:sql-eff-spec} in the previous section. Here, we observe a single eigenvalue curve traversing the band gap for each operator, consistent with the bulk-edge correspondence principle as discussed in Section \ref{sec:dfm-edge}. The plots are in excellent agreement with numerical simulations of the $L^2_k(\R^2/\Z\boldsymbol{\mathfrak{v}}_1)$ spectrum of ${H^\delta_{\rm edge} = T_* H^\delta}$, located in Section \ref{sec:numerics}. Our analytical results, presented in the subsequent sections, explain these features.

\subsubsection{Essential spectra of $\cancel{\mathfrak{D}}^\pm(\kappa)$}
\label{sec:spec-dfm-ess}

Fourier analysis of the constant-coefficient asymptotic Hamiltonians arising for ${|X| \to +\infty}$ yields a characterization of the essential spectrum of $\cancel{\mathfrak{D}}^\pm(\kappa)$. 

\begin{proposition}
\label{prop:spec-dfm-gap}
{\rm (Essential spectra of $\cancel{\mathfrak{D}}^\pm(\kappa)$.)}
Consider $\cancel{\mathfrak{D}}^+(\kappa)$ \eqref{eq:dfm-edge-eff_2}, where we assume ${\chi(X)^2 \to 1}$ sufficiently rapidly as ${|X| \to +\infty}$. Then, the following statements are equivalent:
\begin{enumerate}
\item ${{-a_0^2} + a_1^2 + a_2^2 > 0}$ and ${c \neq 0}$,
\item For each ${\kappa \smallin \R}$, the essential spectrum of $\cancel{\mathfrak{D}}^+(\kappa)$ consists of two semi-infinite intervals
\begin{equation}
\label{eqn:specwindowsDirac}
{\rm spec}_{\rm ess}(\cancel{\mathfrak{D}}^+(\kappa)) = (-\infty, \, \nu_-(\kappa)] \cup [\nu_+(\kappa), \, +\infty) ,
\end{equation}
separated by an open interval ${(\nu_-(\kappa), \, \nu_+(\kappa))}$ containing ${\Omega = 0}$. Further, ${\kappa \mapsto \cancel{\mathfrak{D}}^+(\kappa)}$ has a band gap:
\begin{equation}
\label{eqn:DirBandGap}
\inf_{\kappa \smallin \R} \bigl( \nu_+(\kappa) - \nu_-(\kappa) \bigr) > 0 ;
\end{equation}
see Section \ref{sec:fb-thy-gaps}.
\end{enumerate}
The same holds when $\cancel{\mathfrak{D}}^+(\kappa)$ is replaced by $\cancel{\mathfrak{D}}^-(\kappa)$.
\end{proposition}

\begin{remark}
{\rm (Spectral gaps.)}
In general, the families ${\kappa \mapsto \cancel{\mathfrak{D}}^\pm(\kappa)}$ do not yield a spectral gap. (However, there may be a spectral gap, e.g., for ${a_0 = 0}$.)
\end{remark}

\noindent The proof of Proposition \ref{prop:spec-dfm-gap} is presented in Appendix \ref{apx:dfm-edge-eff-gap}. 

\subsubsection{Discrete spectrum of $\cancel{\mathfrak{D}}^\pm(0)$}
\label{sec:spec-dfm-disc}

Let $\cancel{\mathfrak{D}}^\pm(\kappa)$ be given by \eqref{eq:dfm-edge-eff_2}. Numerical simulations indicate that, for each ${\kappa \smallin \R}$, $\cancel{\mathfrak{D}}^\pm(\kappa)$ have at least one eigenvalue $\Omega_\pm(\kappa)$ in the gap, and that the curves ${\kappa \mapsto \Omega_\pm(\kappa)}$ traverse the gap.

Assume the band gap condition ${{-a_0}^2 + a_1^2 + a_2^2 > 0}$ of Proposition \ref{prop:spec-dfm-gap}. We first establish the existence of distinguished bound states of both $\cancel{\mathfrak{D}}^-(0)$ and $\cancel{\mathfrak{D}}^+(0)$.

\begin{proposition}
\label{prop:zero-energy}
{\rm (Zero-energy bound states of $\cancel{\mathfrak{D}}^-(0)$ and $\cancel{\mathfrak{D}}^+(0)$.)}
Assume ${{-a_0}^2 + a_1^2 + a_2^2 > 0}$ and ${c \neq 0}$, so that the essential spectra of $\cancel{\mathfrak{D}}^\pm(0)$ have a gap containing ${\Omega = 0}$; see Proposition \ref{prop:spec-dfm-gap}. Then, ${\Omega = 0}$ is a simple eigenvalue of both $\cancel{\mathfrak{D}}^-(0)$ and $\cancel{\mathfrak{D}}^+(0)$ with corresponding eigenstates
\begin{equation}
\label{eq:psi-star}
\psi^\star_\pm(X) = \exp \biggl( {\frac{-|c|}{\sqrt{{-a_0^2} + a_1^2 + a_2^2}}} \int_{[0, \, X]} \chi(X^\prime) \, {\rm d}X^\prime \biggr) \psi_\pm ,
\end{equation}
where ${\psi_\pm \smallin \C^2}$ are eigenvectors of the matrix
\begin{equation}
\label{eqn:RHSmatrix}
\begin{bmatrix}
i a_0 & i a_1 + a_2 \\
{-i a_1} + a_2 & {-i a_0}
\end{bmatrix}
\end{equation}
associated with eigenvalues ${\lambda_\pm \equiv \mp \sqrt{{-a_0}^2 + a_1^2 + a_2^2}}$, respectively.
\end{proposition}

\begin{proof}
First, consider $\cancel{\mathfrak{D}}^+(0)$. The ${\Omega = 0}$ eigenvalue equation is
\begin{equation}
\label{eq:0energy-state}
\! \begin{bmatrix}
a_0 & a_1 - i a_2 \\
a_1 + i a_2 & a_0
\end{bmatrix}
\! (-i \partial_X) \psi^\star(X) = {-
\! \begin{bmatrix}
c & 0 \\
0 & -c
\end{bmatrix}} \!
\chi(X) \psi^\star(X) .
\end{equation}
The determinant of the matrix on the left-hand side is ${-(-a_0^2 + a_1^2 + a_2^2)}$, which is strictly negative by the gap condition. Hence, the matrix is invertible and
\begin{equation}
\label{eq:direigmateq}
\partial_X \psi^\star(X) = \frac{c}{-a_0^2 + a_1^2 + a_2^2}
\begin{bmatrix}
i a_0 & i a_1 + a_2 \\
{-i a_1} + a_2 & {-i a_0}
\end{bmatrix} \!
\chi(X) \psi^\star(X) .
\end{equation}
The matrix on the right-hand side has eigenvalues ${\lambda_\pm \equiv \mp \sqrt{{-a_0}^2 + a_1^2 + a_2^2} \smallin \R}$. Let ${\psi_\pm \smallin \C^2}$ denote eigenvectors corresponding to the eigenvalues $\lambda_\pm$. It follows that any solution ${\psi^\star \smallin L^2(\R; \C^2)}$ of \eqref{eq:0energy-state} is of the form \eqref{eq:psi-star}. 

The analogous result for $\cancel{\mathfrak{D}}^-(0)$ holds by a similar argument: Here, $c$ is replaced by ${-c}$, and consequently the eigenpair ${(\lambda_-, \psi_-)}$ is chosen instead of ${(\lambda_+, \psi_+)}$.
\end{proof}

\subsubsection{Eigenvalue curves of ${\kappa \mapsto \cancel{\mathfrak{D}}^\pm(\kappa)}$}

The following results provide information about the eigenvalue curve ${\kappa \mapsto \Omega_\pm(\kappa)}$ of $\cancel{\mathfrak{D}}^\pm(\kappa)$ passing through ${\Omega_\pm(0) = 0}$; see Proposition \ref{prop:zero-energy} and Figure \ref{fig:effDirac-spec}.

\begin{proposition}
\label{prop:dfm-eig-local}
{\rm (On the eigenvalue curve through zero energy.)}
Consider $\cancel{\mathfrak{D}}^\pm(\kappa)$ \eqref{eq:dfm-edge-eff_2}. Let 
${\kappa \mapsto \Omega_\pm(\kappa)}$ denote the eigenvalue curve of $\cancel{\mathfrak{D}}^\pm(\kappa)$, which  passes through ${\Omega_\pm(0) = 0}$.
\begin{enumerate}
\item \label{itm:dfm-eig-local-1} 
Let ${(\Omega, \psi) = (0, \psi^\star_\pm)}$ be the normalized zero-energy bound state eigenpair of $\cancel{\mathfrak{D}}^\pm(0)$ (Proposition \ref{prop:zero-energy}).
For $|\kappa|$ small, we have the expansion
\begin{equation}
\label{eq:dfm-omega}
\Omega_\pm(\kappa) = \kappa \Omega^{(1)}_\pm + \kappa^2 \Omega^{(2)}_\pm + O(\kappa^3) \ \ \text{as} \ \ |\kappa| \to 0 ,
\end{equation}
where
\begin{align}
\label{eq:dfm-omega-1}
\Omega^{(1)}_\pm & \equiv \langle \psi^\star_\pm, \, (b \cdot \sigma) \psi^\star_\pm \rangle , \\
\label{eq:dfm-omega-2}
\Omega^{(2)}_\pm & \equiv - \langle (\Omega^{(1)}_\pm - (b \cdot \sigma)) \psi^\star_\pm, \, \cancel{\mathfrak{D}}^{\pm}(0)^{-1} (\Omega^{(1)}_\pm - (b \cdot \sigma)) \psi^\star_\pm \rangle .
\end{align}
\item \label{itm:dfm-eig-local-2} If the domain wall function $\chi(X)$ is odd, i.e., ${\chi(-X) = -\chi(X)}$, then ${\kappa \mapsto \Omega_\pm(\kappa)}$ satisfies
\begin{equation}
\label{eq:dfm-omega-odd}
\Omega_\pm(-\kappa) = -\Omega_\pm(\kappa) .
\end{equation}
Hence, the expansion \eqref{eq:dfm-omega} contains only odd powers of $\kappa$.
\end{enumerate}
\end{proposition}

\noindent Part \ref{itm:dfm-eig-local-1} of Proposition \ref{prop:dfm-eig-local}, is proven in Supplementary Material \ref{supp:dfm-eig-local-1}. Part \ref{itm:dfm-eig-local-2}, proven in Supplementary Material \ref{supp:dfm-eig-local-2}, is based on a global symmetry: if $\Omega$ is an eigenvalue of $\cancel{\mathfrak{D}}^+(\kappa)$, then $-\Omega$ is an eigenvalue of $\cancel{\mathfrak{D}}^+(-\kappa)$; see Proposition \ref{prop:dfm-odd-sym}. 

We conclude the section by identifying a high-symmetry subclass of Dirac operators.

\begin{proposition}
\label{prop:dfm-eig-lin}
{\rm (Dirac operators with linear eigenvalue curves.)}
If the parameters $a$, $b$ of $\cancel{\mathfrak{D}}^\pm(\kappa)$ in \eqref{eq:dfm-edge-eff_2} satisfy
\begin{equation}
\label{eq:dfm-eig-lin}
a_1 b_1 - a_2 b_2 = 0 \quad \text{and} \quad a_0 = 0 ,
\end{equation}
then the eigenpair ${(\Omega_\pm(\kappa), \psi_\pm(\kappa))}$, for which ${(\Omega_\pm(0), \psi^\star_\pm(0)) = (0, \psi^\star_\pm)}$, is given by
\begin{equation}
(\Omega_\pm(\kappa), \psi_\pm(\kappa)) = (\kappa \Omega^{(1)}_\pm, \psi_\pm^\star) ,
\end{equation}
with $\Omega^{(1)}_\pm$ as in \eqref{eq:dfm-omega-1}. Note, in this case, that $\psi_\pm(\kappa)$ is independent of $\kappa$ and ${\kappa \mapsto \Omega_\pm(\kappa)}$ is linear.
\end{proposition}

\noindent The proof of Proposition \ref{prop:dfm-eig-lin} is given in Supplementary Material \ref{supp:dfm-eig-lin}. The subclass of Dirac operators in Proposition \ref{prop:dfm-eig-lin} includes the well-known cases arising for edge states in 2D honeycomb structures; see, e.g., \cite[Proposition 4.4]{drouotweinstein2020}.

\smallskip

\section{Computational experiments}
\label{sec:numerics}

\setcounter{equation}{0}
\setcounter{figure}{0}

In this section, we report on numerical simulations of the spectral problems associated with our 2D Hamiltonians and compare with predictions from our asymptotic analysis in Sections \ref{sec:multi} and \ref{sec:spec}. Specifically, we compute the band structures of bulk Hamiltonians ${H_{\rm bulk} = H_V}$, defined in \eqref{eq:sql-bulk-op_2}, with $V$ a simple square lattice potential, and ${H_{\rm bulk} = H_{V \circ \, T^{-1}}}$, defined in \eqref{eq:dfm-bulk-op_2}, for the same $V$ under a shear-like deformation $T$. We additionally compute the band structures of corresponding perturbed bulk Hamiltonians $\smash{H^{\pm, \, \delta}_{\rm bulk}}$, as defined in \eqref{eq:def_sql-asym} (square lattice case) and \eqref{eq:def_dfm-asym_a} (deformed square lattice case). Finally, we compute edge state diagrams for the above bulk Hamiltonians, which consist only of continuous spectrum determined by their band structures, and for our edge Hamiltonians, defined in \eqref{eq:def_sql-edge} (square lattice case) and \eqref{eq:def_dfm-edge_a} (deformed square lattice case), which feature edge state eigenvalue curves.

\subsection{Numerical methods}

\subsubsection{Computing band structures}

The band structure of a periodic, elliptic operator is determined by the collection of solutions to its family of Floquet-Bloch eigenvalue problems, as defined in \eqref{eq:Lk-evp}; see Section \ref{sec:fb-thy-bulk} for a review. To compute our band structures, we use a first-order finite difference scheme based on a discretization of the unit square ${\R^2/\Z^2 \simeq}$ ${[-1/2, \, 1/2] \times [-1/2, \, 1/2]}$ by a grid of uniform spacing ${h = 1/20}$. Specifically, for ${H_{\rm bulk} = H_V}$ and $H_{V \circ \, T^{-1}}$:
\begin{itemize}
\item We discretize the Laplacian ${-\Delta}_{\bm x}$ using centered differences.

To impose pseudoperiodic boundary conditions at $x_1$, ${x_2 = \pm 1/2}$, we equivalently consider the operator ${- e^{-i{\bm k} \cdot {\bm x}} \Delta_{\bm x} \,  e^{i{\bm k} \cdot {\bm x}} = {-\Delta}_{\bm x} - 2 i {\bm k} \cdot \nabla + |{\bm k}|^2}$, where ${{\bm k} = [k_1, k_2]^\mathsf{T}}$ varies over the Brillouin zone ${\mathcal{B} = [-\pi, \, \pi] \times}$ ${[-\pi, \, \pi]}$, and impose periodic boundary conditions. 

\item To construct the potential:

For ${H_{\rm bulk} = H_V}$, defined in \eqref{eq:sql-bulk-op_2}, we use a square lattice potential ${{\bm x} \mapsto V({\bm x})}$ constructed as follows: First, we define a smooth, localized ``atomic'' potential, given by a standard compactly supported bump function centered at the origin with amplitude ${-150}$ and width $1/4$. Taking its periodic extension (i.e., summing over integer translates by square lattice basis vectors), we obtain a square lattice potential.

For $H_{\rm bulk} = H_{V \circ \, T^{-1}}$, defined in \eqref{eq:dfm-bulk-op_2}, we use the same square lattice potential as above, and consider {\it tilt} deformations:
\begin{equation}
\label{eqn:numTmatrix}
T(\phi) =
\begin{bmatrix}
\cos(\phi) & -\sin(\phi) \\
-\sin(\phi) & \cos(\phi)
\end{bmatrix} \! .
\end{equation}
By changing coordinates, we instead work with $T_* H_V$, defined in \eqref{eq:dfm-bulk-op_3}, where
\begin{equation}
(T(\phi)^\mathsf{T} T(\phi))^{-1} = \sec^2(2 \phi) I + \sec(2 \phi) \tan(2 \phi) \sigma_1 .
\end{equation}
We also record ${\det(T(\phi)^{-1}) = \sec(2 \phi)}$.
\end{itemize}

\noindent The Floquet-Bloch eigenvalues and eigenstates themselves are obtained using a solver for sparse matrices.

\begin{figure}[!t]
    \centering
    \begin{subfigure}{0.25\textwidth}
        \centering
        \subcaption{}
        \includegraphics[scale = 0.2]{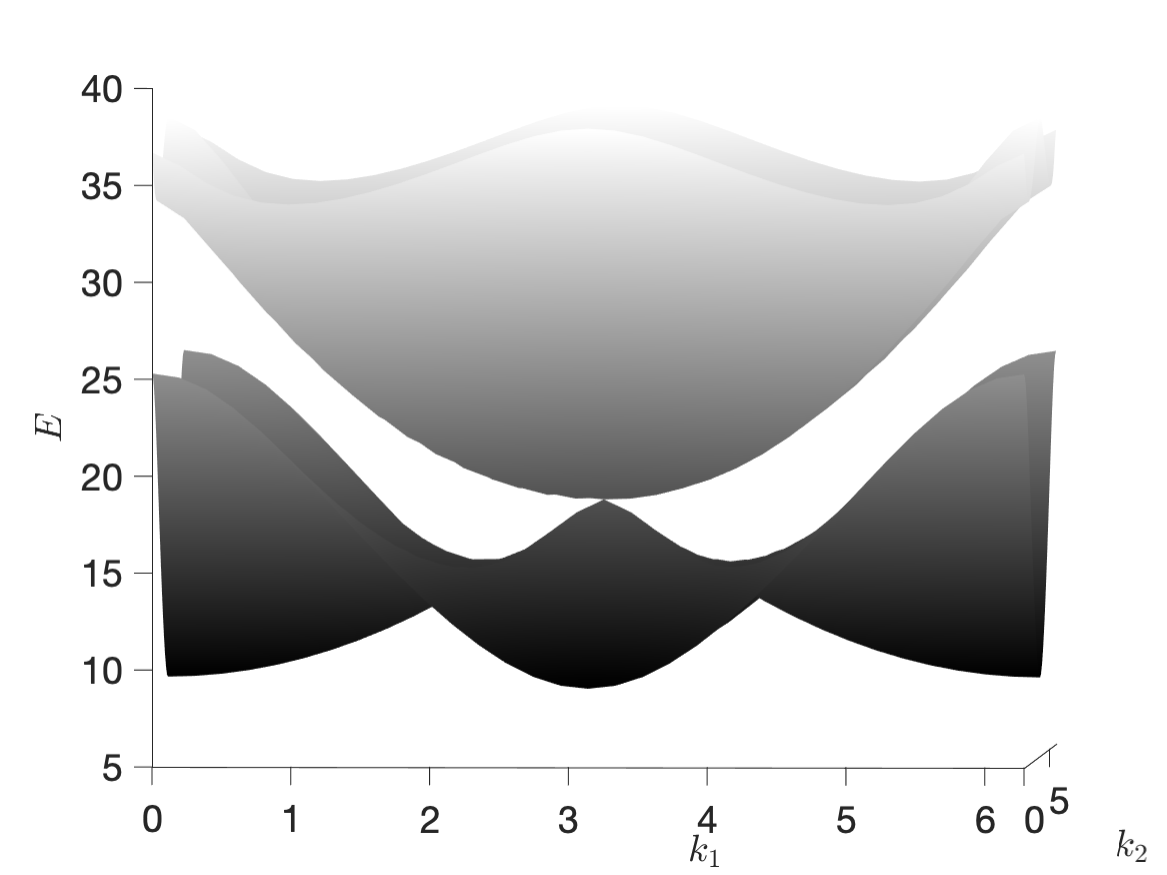} \\
        \includegraphics[scale = 0.2]{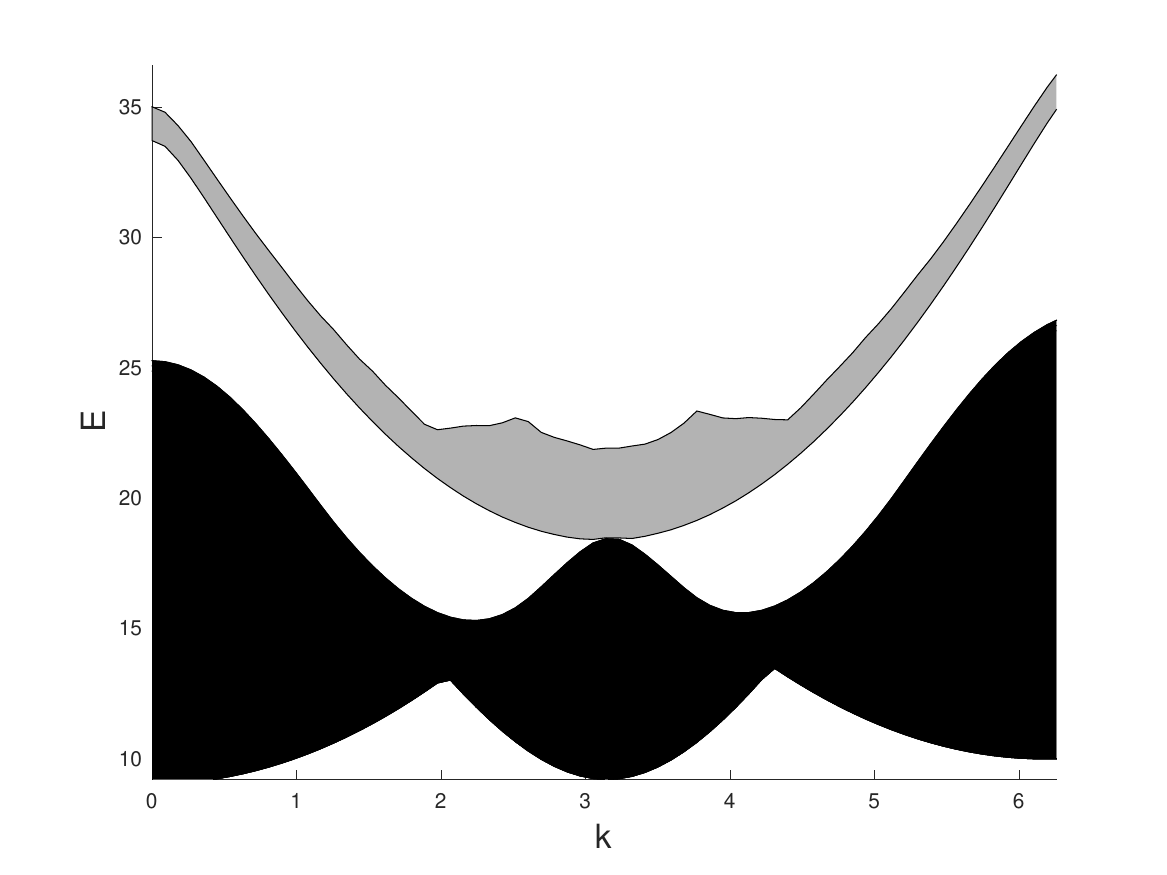}
    \end{subfigure}
    \begin{subfigure}{0.25\textwidth}
        \centering
        \subcaption{}
        \includegraphics[scale = 0.2]{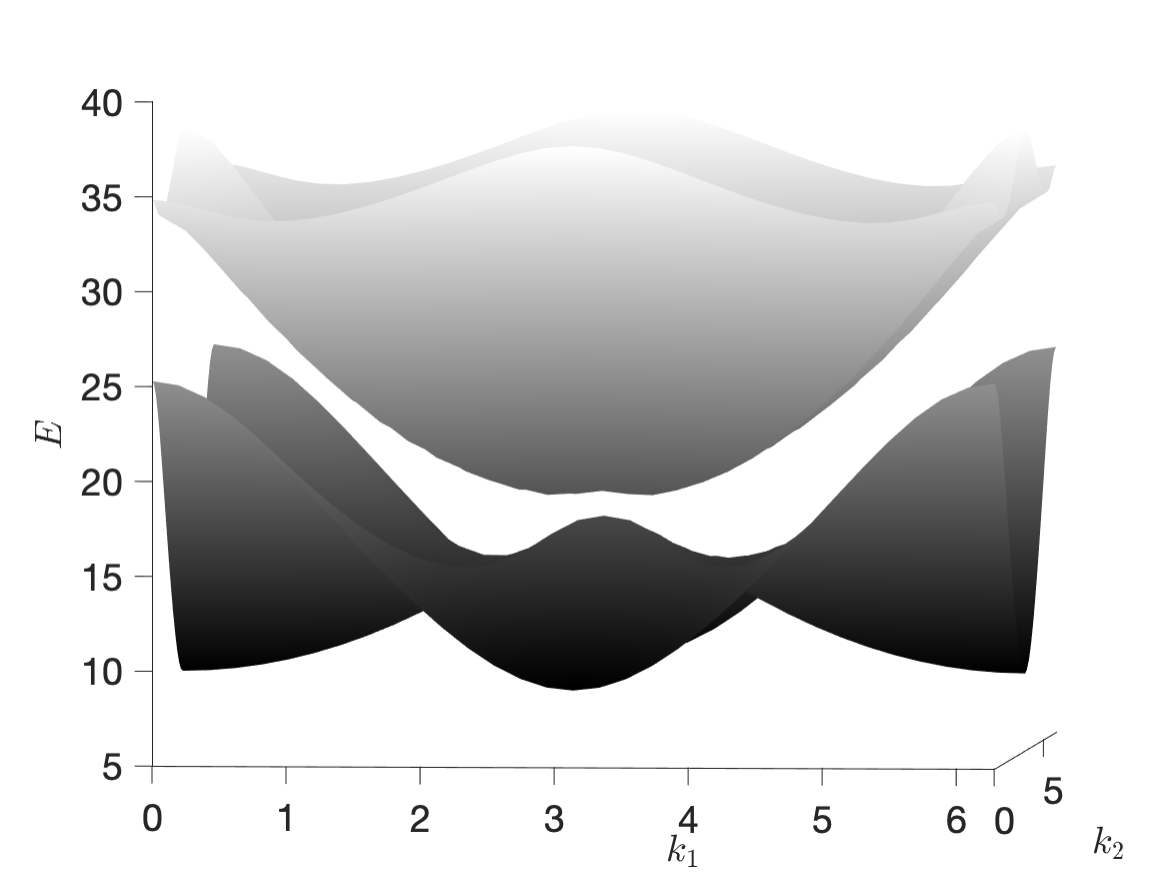} \\
        \includegraphics[scale = 0.2]{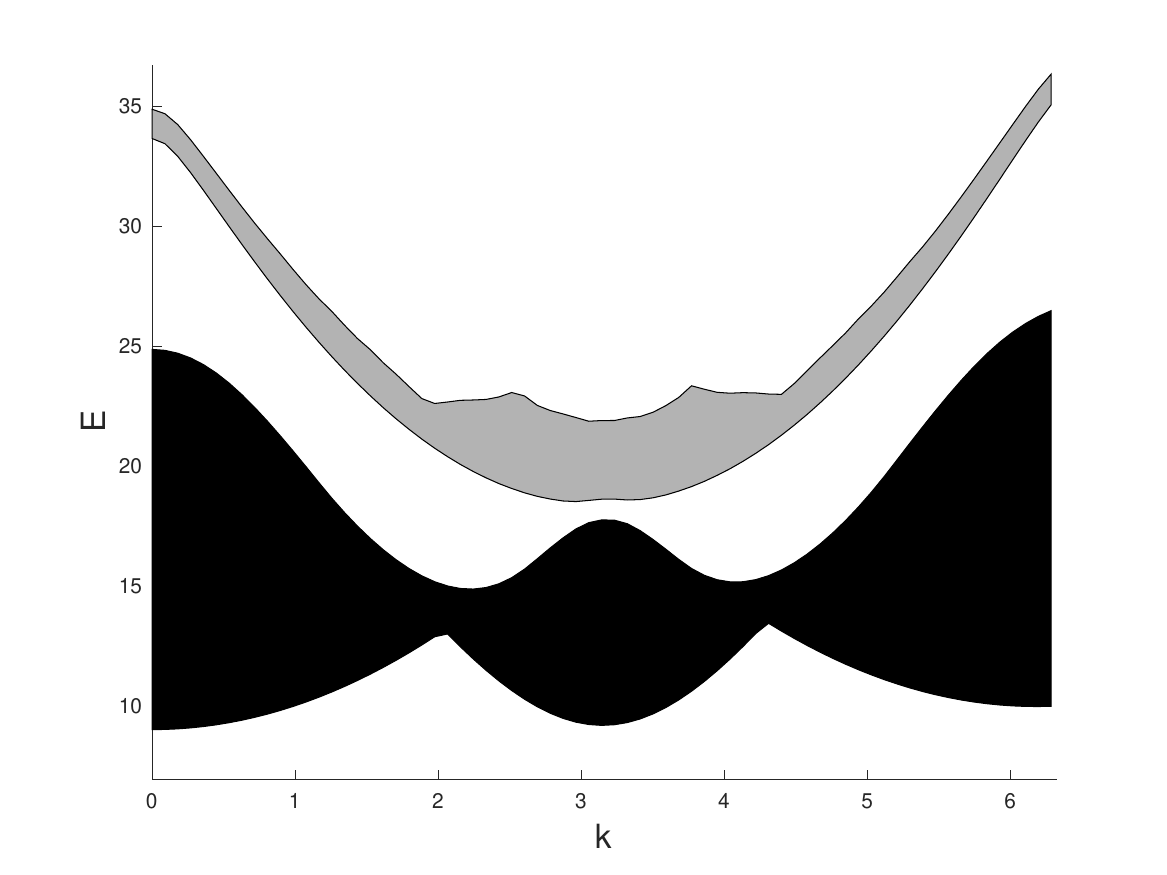}
    \end{subfigure}
    \begin{subfigure}{0.35\textwidth}
        \centering
        \subcaption{}
        \raisebox{1cm}{\includegraphics[scale = 0.28]{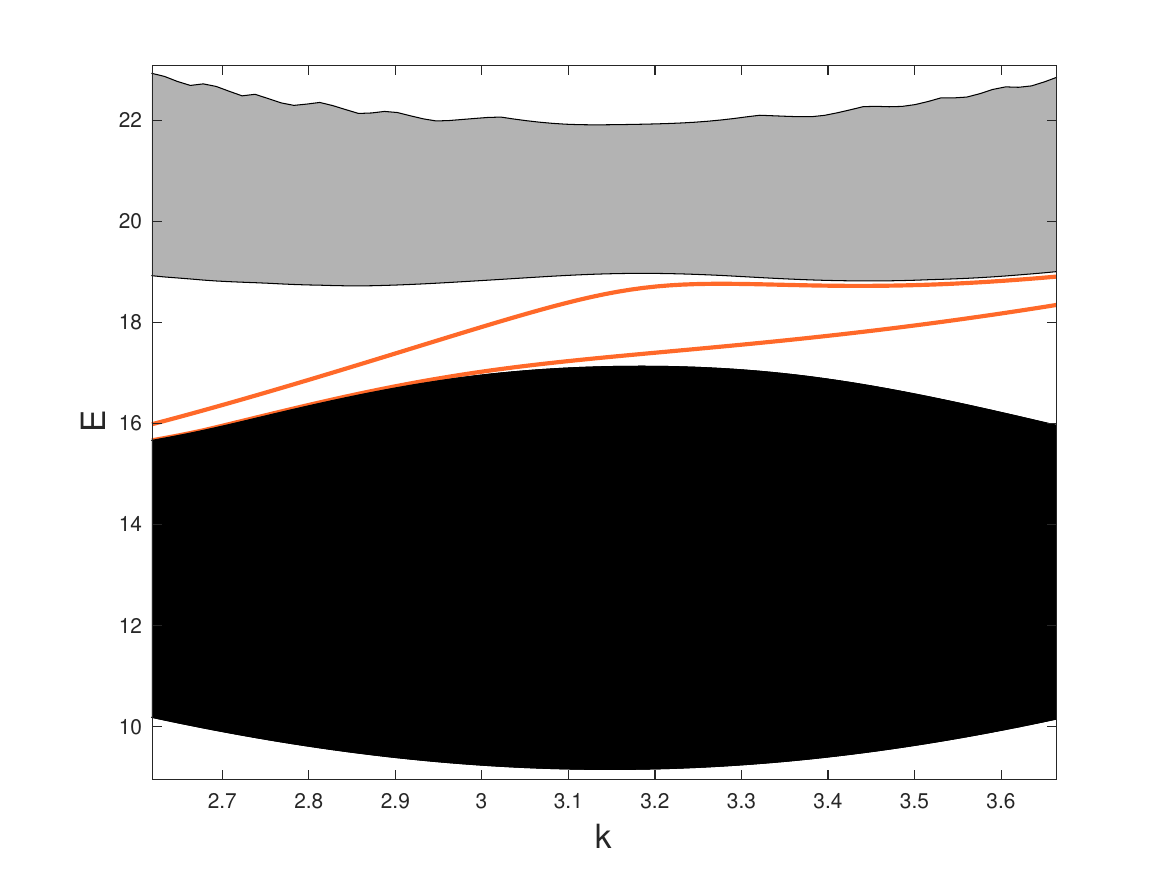}}
    \end{subfigure}
    \caption{(a) Bulk dispersion surfaces of the undeformed ($\phi=0$) square lattice Hamiltonian \eqref{eq:sql-bulk-op_2} (top) and ungapped bulk spectrum (with quadratic degeneracy) in a cylinder (bottom).  (b) Bulk band-gapped spectrum (no defect edge) of $H^{\delta,+}$ from \eqref{eq:def_sql-asym} with $\delta = .1$ and bulk gapped spectrum in a cylinder. (c) Blowup near $k=\pi$ showing two edge state curves traversing the bulk band gap for Hamiltonian $H^\delta$ from \eqref{eq:def_sql-edge} with $\delta = .1$, as anticipated by Theorem \ref{thm:multi-sql}.
    }
    \label{fig:newksweep_untilted}
\end{figure}

We also compute the band structures of bulk Hamiltonians perturbed by $\mathcal{C}$-breaking terms. In particular, we consider $\smash{H^{\delta, \, +}_{V, \, A}}$, defined in \eqref{eq:sql-bulk-op-breakC}, and $\smash{H^{\delta, \, +}_{V \circ \, T^{-1}, \, A \, \circ \, T^{-1}}}$, defined in \eqref{eq:dfm-bulk-op-breakC}. By again changing coordinates, we instead work with $T_* H^{\delta, \, \pm}_{V, \, A}$, defined in \eqref{eq:dfm-bulk-op-breakC_2}.
Here:
\begin{itemize}
\item We use the same square lattice potential $V$ and deformation $T$ as described above.

\item We discretize ${\nabla_{\bm x} \cdot A({\bm x}) \sigma_2 \nabla_{\bm x}}$ using centered differences, but work with ${e^{-i{\bm k} \cdot {\bm x}} \nabla_{\bm x} \cdot A({\bm x}) \sigma_2 \nabla_{\bm x} \,  e^{i{\bm k} \cdot {\bm x}} =}$ $({\nabla_{\bm x} + i{\bm k}) \cdot A({\bm x}) \sigma_2 (\nabla_{\bm x} + i{\bm k})}$ with periodic boundary conditions, as above.

Our choice of ${{\bm x} \mapsto A({\bm x})}$ is given by the finite Fourier series
\begin{equation}
A(x_1, x_2) = 5 \bigl( \cos(2 \pi x_1) + \cos(2 \pi x_2) + \cos (2 \pi (x_1 + x_2) + \cos (2 \pi (x_1 - x_2) \bigr) .
\end{equation}
\end{itemize}

\subsubsection{Computing edge state diagrams}
\label{sec:spectra}

An edge state diagram for a 2D operator with 1D translation invariance is the collection of its $L^2_k(\R^2/\Z\boldsymbol{\mathfrak{v}}_1)$ spectra, as defined in the spectral problem \eqref{eq:kpar-evp}, as a function of ${k \smallin B_1 = [-\pi, \, \pi]}$; see Section \ref{sec:fb-thy-edge} for details. We take the vertical edge ${\boldsymbol{\mathfrak{v}}_1 = [0, 1]^\mathsf{T}}$. To compute edge state diagrams, we again make use of a finite difference scheme: this time, using a discretization of the strip ${[-30, \, 30] \times [0, \, 1]}$ with a grid of uniform spacing ${h = 1/20}$.

\begin{itemize}
\item Again, we discretize our differential operators using centered differences.

To implement, we work with the operator ${-\Delta_{\bm x} - 2i k \partial_{x_2} + k^2}$, where ${k = k_2}$ varies over the one-dimensional Brillouin zone ${\mathcal{B}_1 = [-\pi, \, \pi]}$, and consider periodic boundary conditions in $x_2$.

To simulate square-integrability in $x_1$, we impose hard truncation (i.e., Dirichlet boundary conditions) at ${x_1 = \pm 30}$. Doing so is known to introduce spurious eigenstates, localized at the computational (artificial) boundaries.

\item We use the same square lattice potential $V$, deformation $T$, and magnetic term $A$, as above.

\item We use the domain wall function ${\chi(X) = \tanh(10 X)}$ throughout to ensure that the transition region is well-contained inside our computational domain for a range of small $\delta$ values.
\end{itemize}

\begin{figure}[!t]
    \centering
    \begin{subfigure}{0.25\textwidth}
        \centering
        \subcaption{}
        \includegraphics[scale = 0.2]{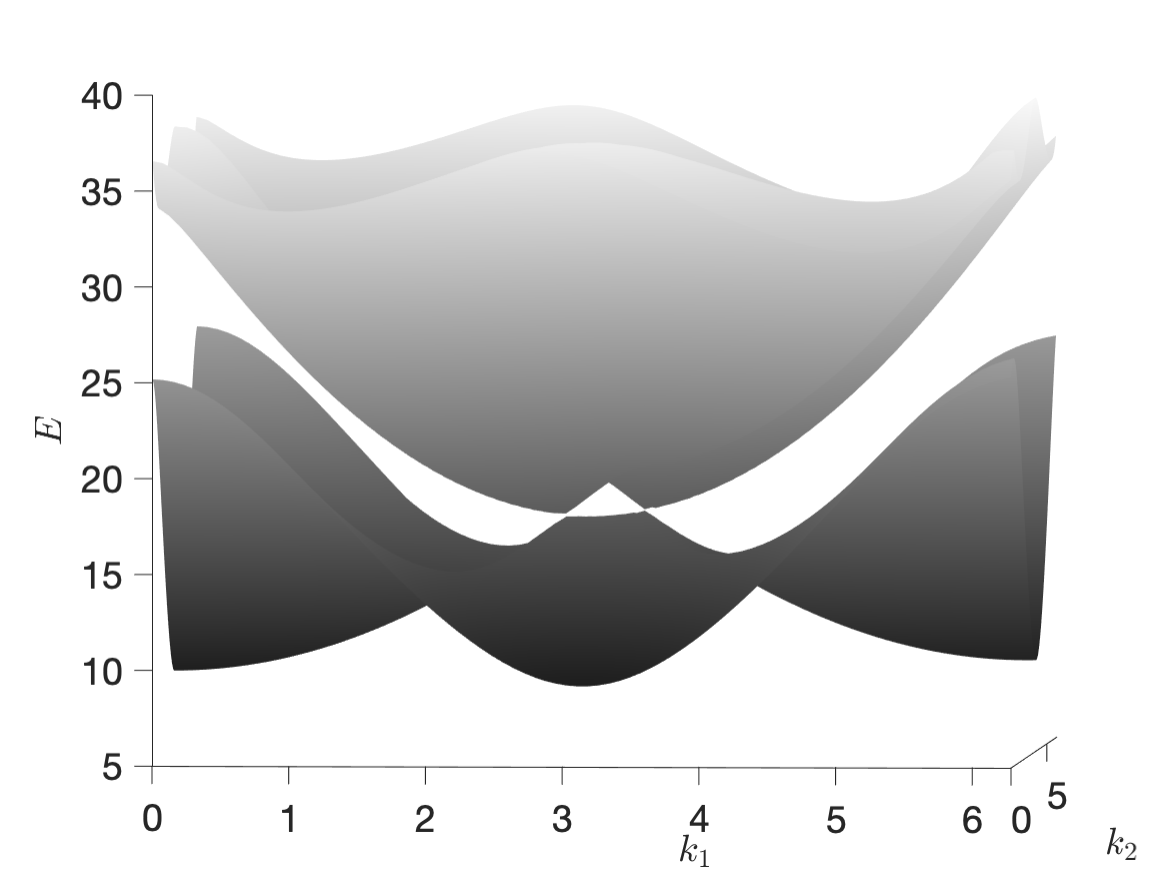} \\
        \includegraphics[scale = 0.2]{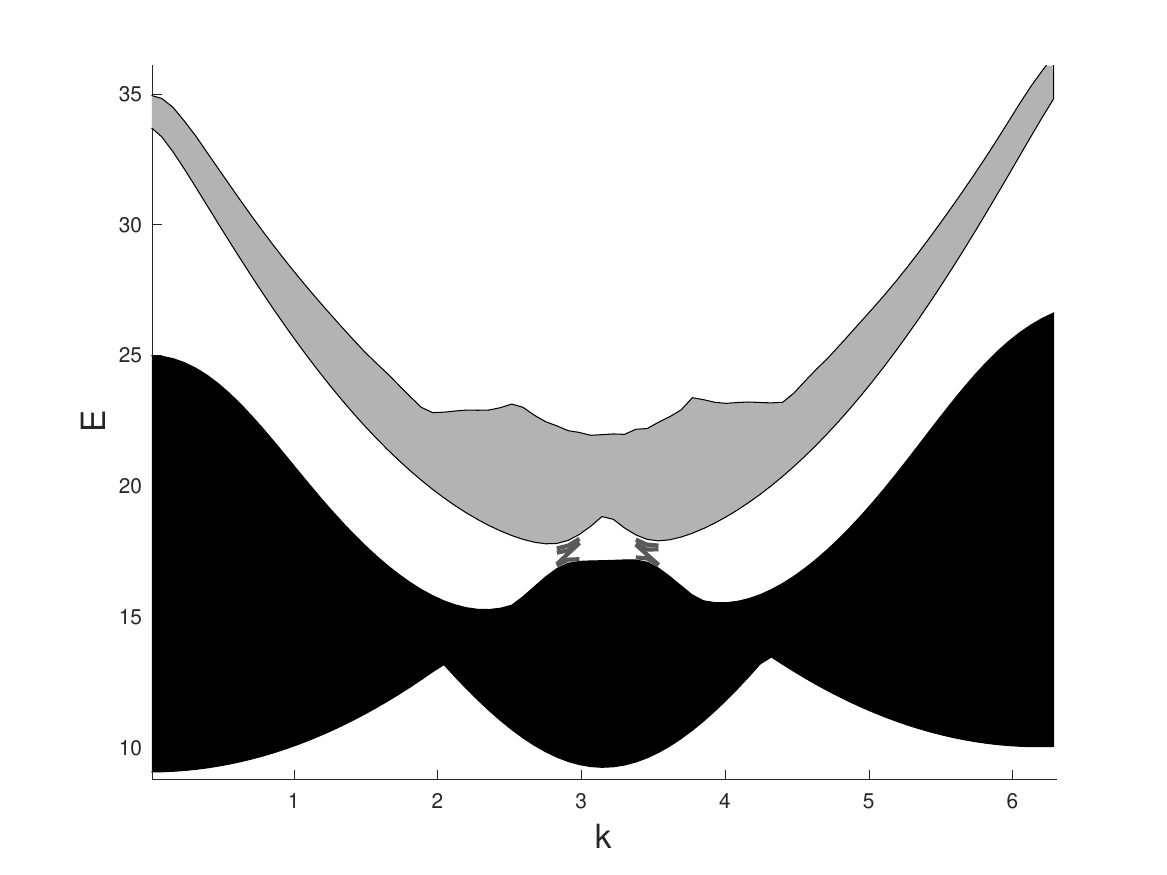}
    \end{subfigure}
    \begin{subfigure}{0.25\textwidth}
        \centering
        \subcaption{}
        \includegraphics[scale = 0.2]{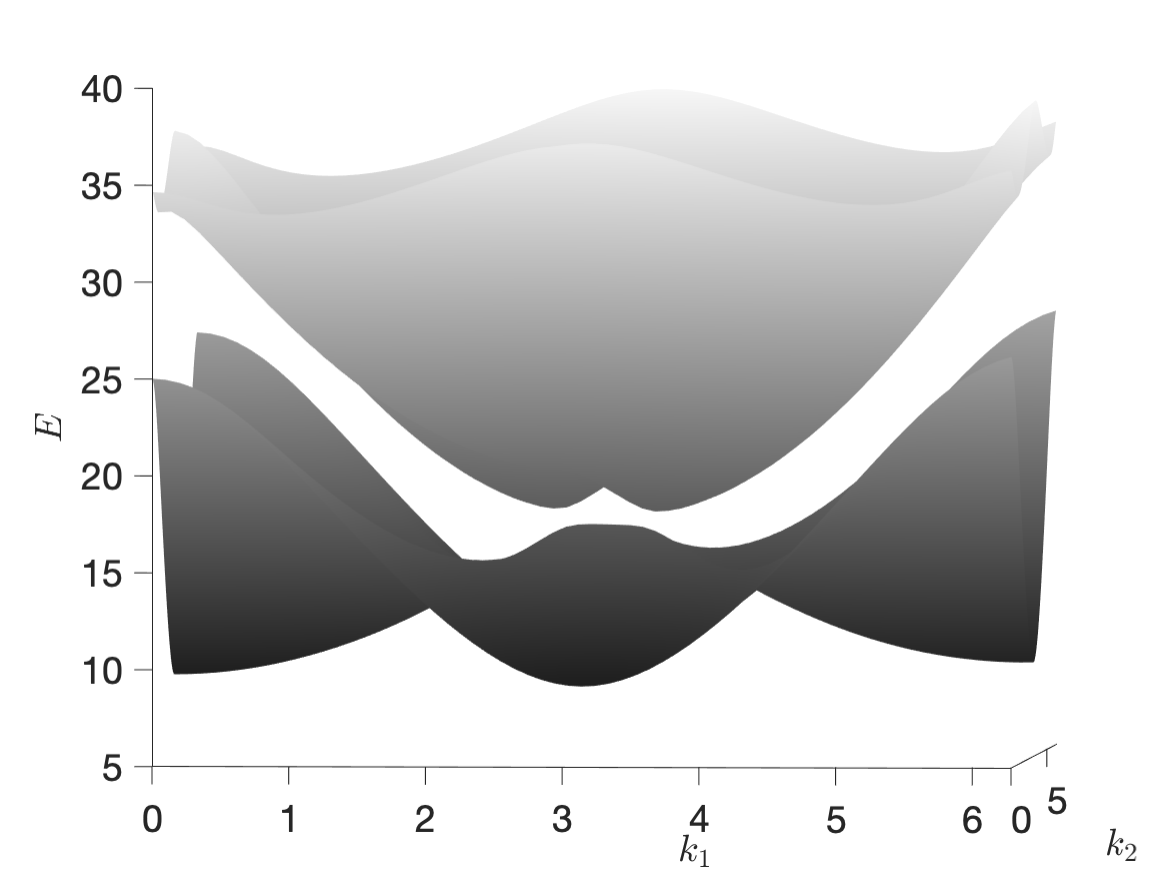} \\
        \includegraphics[scale = 0.2]{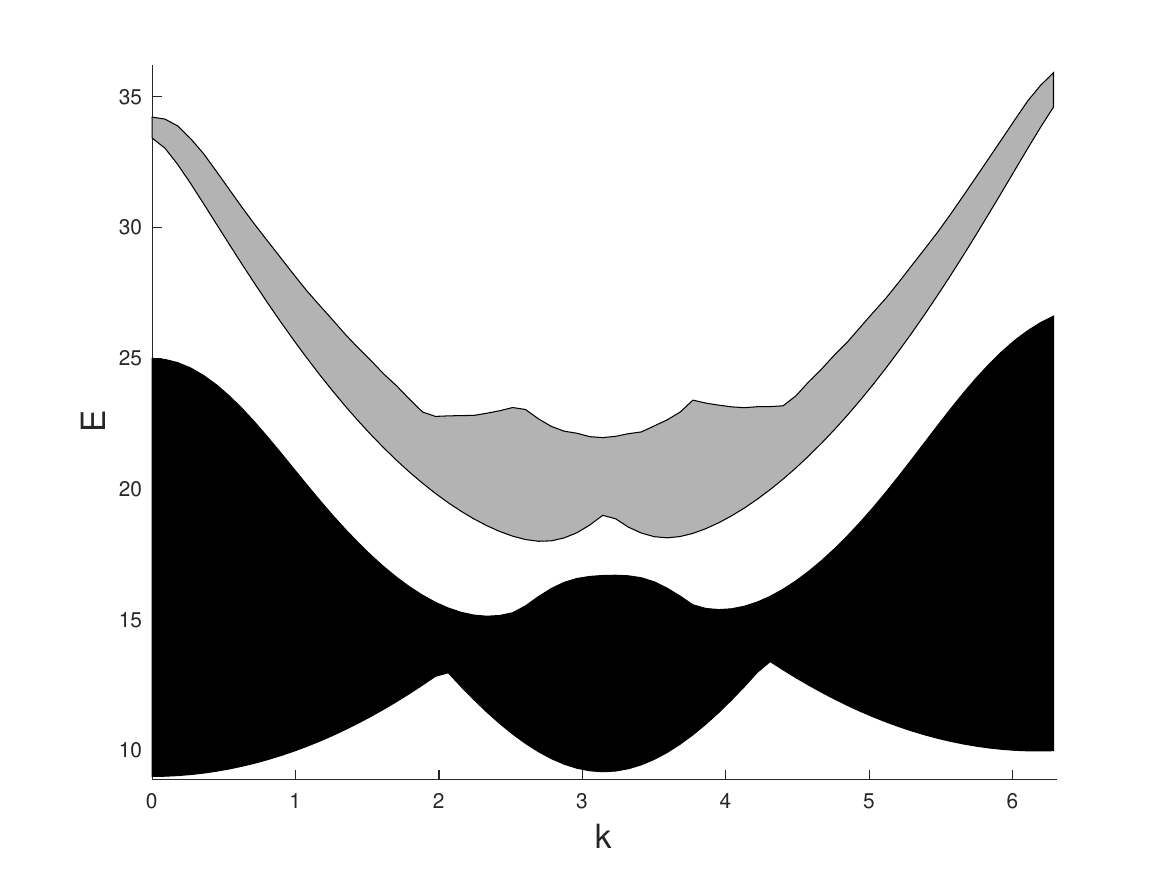}
    \end{subfigure}
    \begin{subfigure}{0.35\textwidth}
        \centering
        \subcaption{}
        \raisebox{1cm}{\includegraphics[scale = 0.28]{tilted_phi_pi_100_spec_flow_alt3-eps-converted-to.pdf}}
    \end{subfigure}
\caption{Weakly deformed square lattice version of Figure \ref{fig:newksweep_untilted};  (a) Band structure of ungapped Hamiltonian \eqref{eq:dfm-bulk-op_2} with $\phi = \pi/100$ showing two nearby Dirac points; bulk (top) and cylinder calculations (bottom). 
(b) Computations of bulk band-gapped $L^2(\R^2)-$ spectrum and bulk spectrum in a cylinder associated to Hamiltonian $T_* H^{\delta,+}$ from \eqref{eq:def_dfm-asym} for parameters $\delta = .01$, and tilt-angle $\phi = \pi/100$.  
(c) Blowup near $k=\pi$ of the spectrum of Hamiltonian $T_* H^\delta$ from \eqref{eq:def_dfm-edge} with $\delta = .01$, and tilt-angle $\phi = \pi/100$ showing two edge state curves , which traverse the bulk band gap, as anticipated by Theorem \ref{thm:multi-dfm}.}
\label{fig:newksweeptilted_1}
\end{figure}

\subsection{Qualitative confirmation of results: Schr\"{o}dinger operators on cylinders}

Figure \ref{fig:newksweep_untilted} displays computations related to (undeformed) square lattice bulk and edge media. We observe, in {\bf (a)}, that the square lattice bulk band structure has a quadratic band degeneracy; the dispersion surfaces pictured locally have the form of two touching paraboloids with curvatures of opposite sign. With the edge chosen, the locally quadratic band edges are also clearly visible in the corresponding edge state diagram. In {\bf (b)}, we see that perturbation by a $\mathcal{C}$-breaking term opens a band gap about the quadratic band degeneracy, reflected both in the band structure and edge state diagram. (Note that a full $L^2(\R^2)$ spectral gap does not open.) Finally, in {\bf (c)}, we observe the emergence of two gap-traversing eigenvalue curves, consistent with both the discussion of the bulk-edge correspondence in Section \ref{sec:bec}, as well as with the local behavior described by the family of effective edge Hamiltonians in Section \ref{sec:spec-sql}. We have checked that the corresponding eigenstates are localized near the transition region of the domain wall. Note that we have omitted eigenvalue branches corresponding to spurious eigenstates localized at computational boundaries, arising due to the hard truncation (i.e., Dirichlet boundary condition) at ${x_1 = \pm 30}$.

Figure \ref{fig:newksweeptilted_1} consists of analogous computations for the media in Figure \ref{fig:newksweep_untilted}, now deformed by the linear transformation $T(\phi)$, defined in \eqref{eqn:numTmatrix}, with ${\phi = \pi/100}$. In {\bf (a)}, the bulk band structure features two tilted conical degeneracies (Dirac points), which are then gapped due to the $\mathcal{C}$-breaking perturbation in {\bf (b)}. The edge state diagram in {\bf (c)} now depicts a pair of separated eigenvalue curves crossing the gap, again consistent with both the bulk-edge correspondence and the local behavior predicted by the effective Hamiltonians of Section \ref{sec:spec-dfm}. We have similarly omitted spurious eigenvalue branches, arising from Dirichlet boundary conditions imposed at the computational boundaries. We note that, while the tilt angle $\phi$ is quite small, the two Dirac points splitting from the quadratic band degeneracy at ${\phi = 0}$ are already well-separated. This is consistent with their predicted $O(\sqrt{\phi})$ separation; see \cite[Theorem 6.5]{chaban2024instability}.

Figure \ref{fig:newdeltasweep} demonstrates, for both square lattice and deformed square lattice edge media (as in Figures \ref{fig:newksweep_untilted} - \ref{fig:newksweeptilted_1}) the emergence of edge state eigenvalues at a fixed parallel quasimomentum, but for a range of small $\delta$. We take ${k = \pi}$ in the undeformed setting and ${k = 2.09}$ (the approximate location of one of the Dirac cones) in the deformed setting. Note that, consistent with our asymptotic analysis, we use a $\delta^2$ scaling of the symmetry breaking perturbation for the undeformed square lattice and a $\delta$ scaling for the deformed version. On the left, we observe two eigenvalues inside the opening gap; on the right, we observe just one. (Note that the second Dirac point is not pictured here.)

Finally, in Figure \ref{fig:newksweep_tilted}, we compute the edge state curves for the undeformed square lattice (${\phi = 0}$) edge medium, as well as for deformed media with tilt angles $\pi/200$, $\pi/100$, and $\pi/50$. The figure demonstrates that the pair of Dirac points, as well as the edge state curves they locally describe, become more separated as the structure is deformed. Individually, each case is consistent with both the bulk-edge correspondence and the results of our multiple-scale analysis. We note that small deformations of the original square lattice medium correspond to gap-preserving perturbations; we observe, accordingly, that all of the deformed media possess the same topological characteristics as the undeformed medium. The scenario of varying deformations motivates questions about the relationship between our two classes of effective edge Hamiltonians; we discuss this further in Section \ref{sec:an_open_prob} below.

\begin{figure}[!t]
    \centering
    \includegraphics[scale = 0.25]{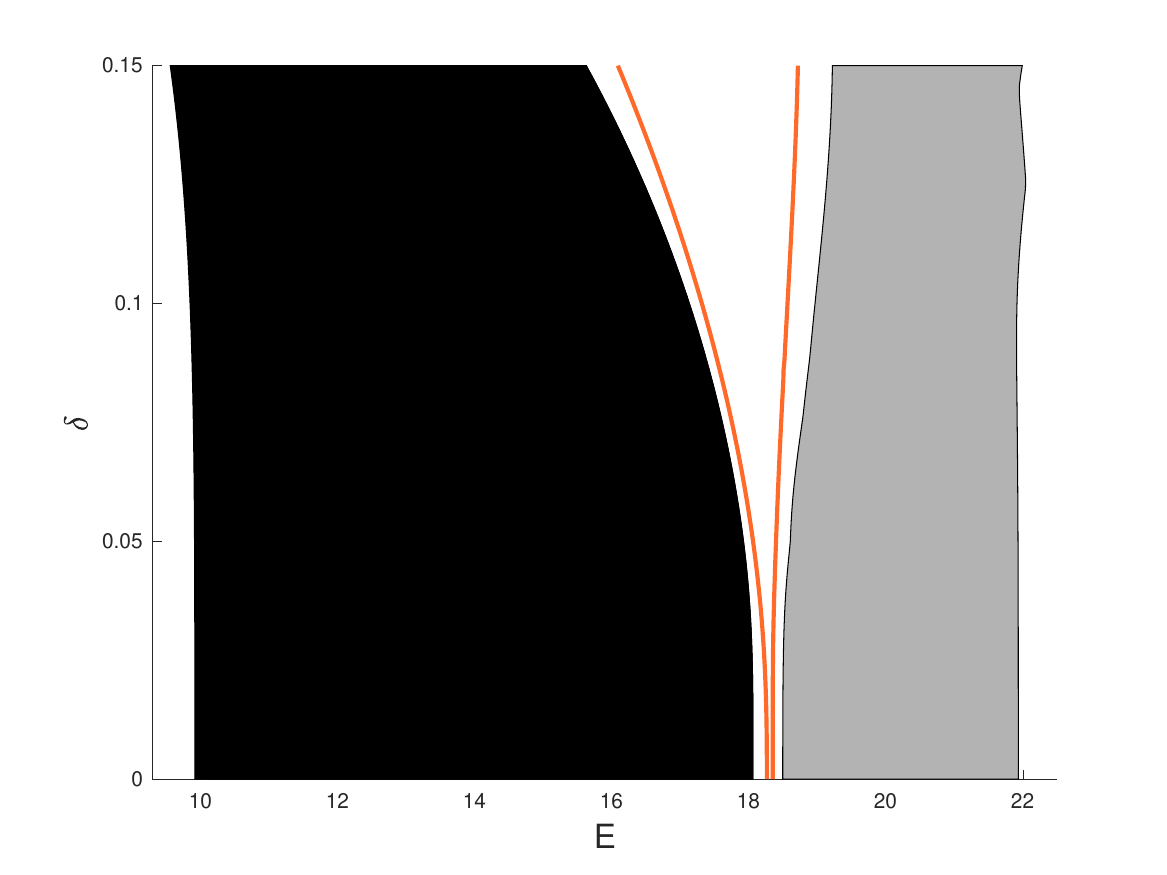}
    \hspace{0.2cm}
    \includegraphics[scale = 0.25]{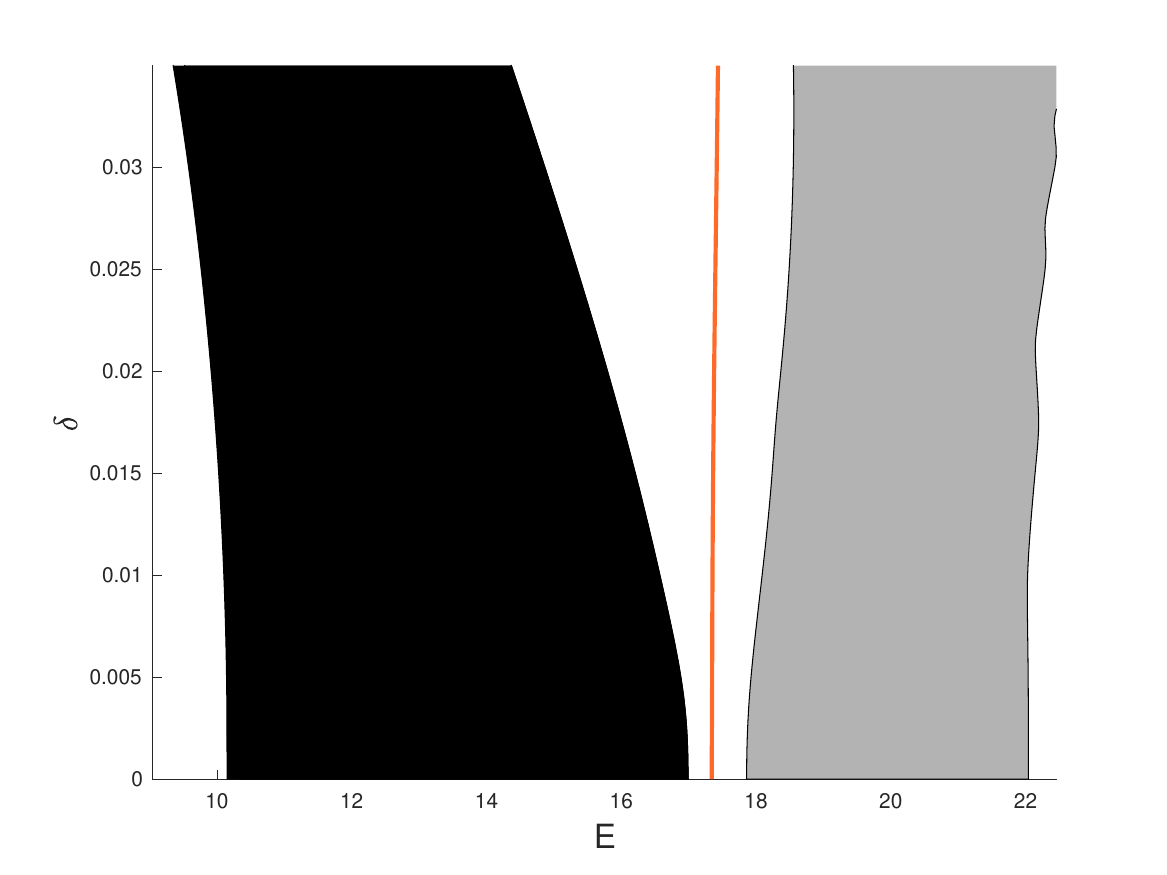}
    \caption{Plots varying $\delta$ for fixed $k$ values:
     (Left) Spectrum near $k=\pi$ for undeformed square lattice associated to Hamiltonian $H^\delta$ from \eqref{eq:def_sql-edge} (tilt angle $\phi=0$); $\delta = 0$ to $\delta = .15$). The gap states on the defect edge are colored in red, while bulk states are colored in blue. (Right) The same plot in the tilted square lattice associated to $T_* H^\delta$ from \eqref{eq:def_dfm-edge} with $\phi = \pi/100$ for $k = 2.9$ with $\delta = 0$ to $\delta = .035$.}
    \label{fig:newdeltasweep}
\end{figure}

\begin{figure}[!t]
    \centering
    \includegraphics[scale = 0.25]{quadratic_spec_flow_alt_2-eps-converted-to.pdf}
    \hspace{0.2cm}
    \includegraphics[scale = 0.25]{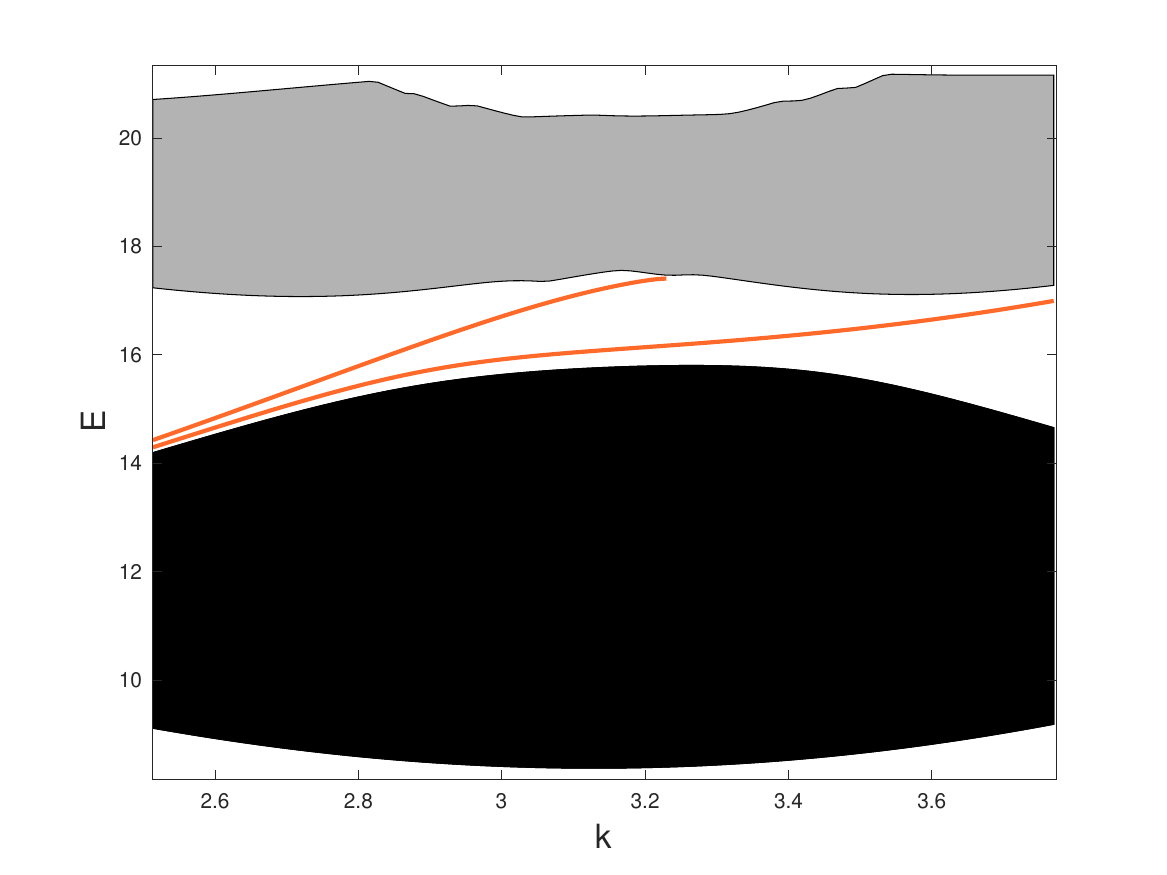} \\
    \includegraphics[scale = 0.25]{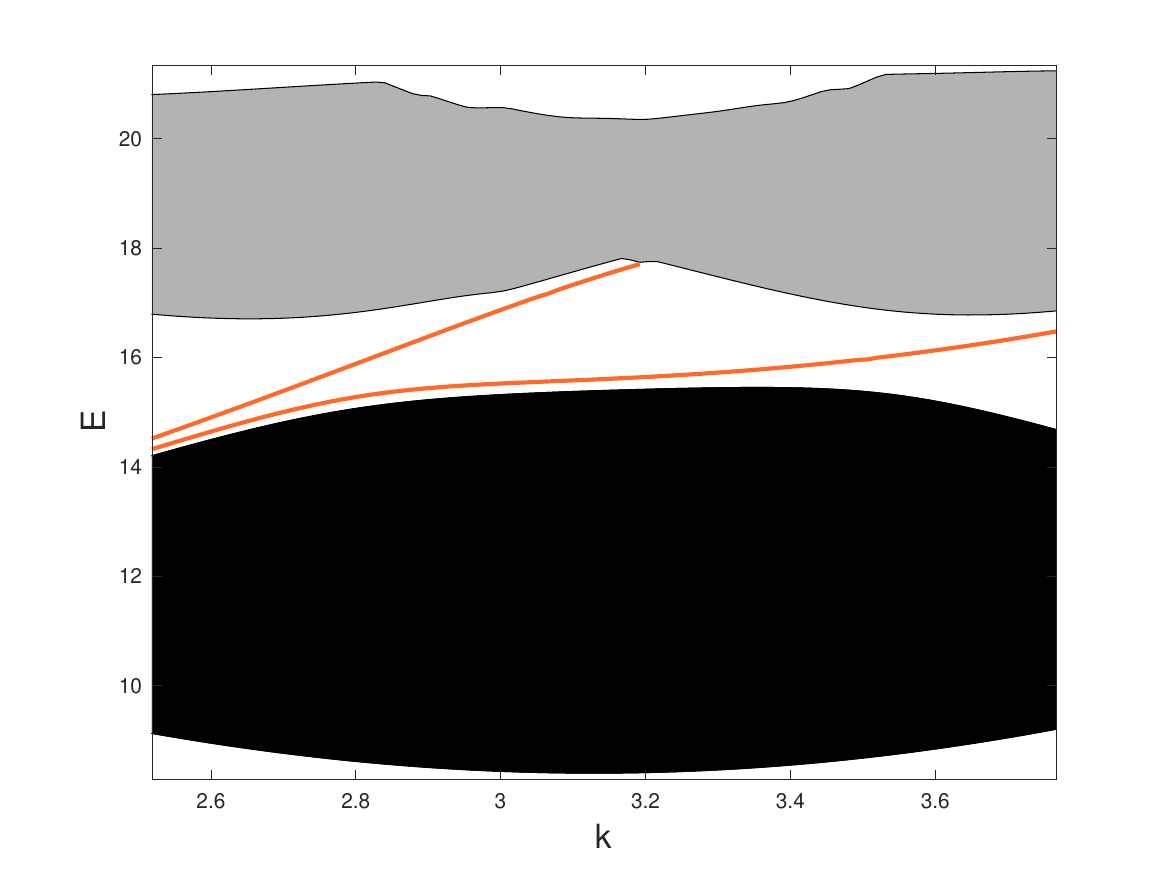}
    \hspace{0.2cm}
    \includegraphics[scale = 0.25]{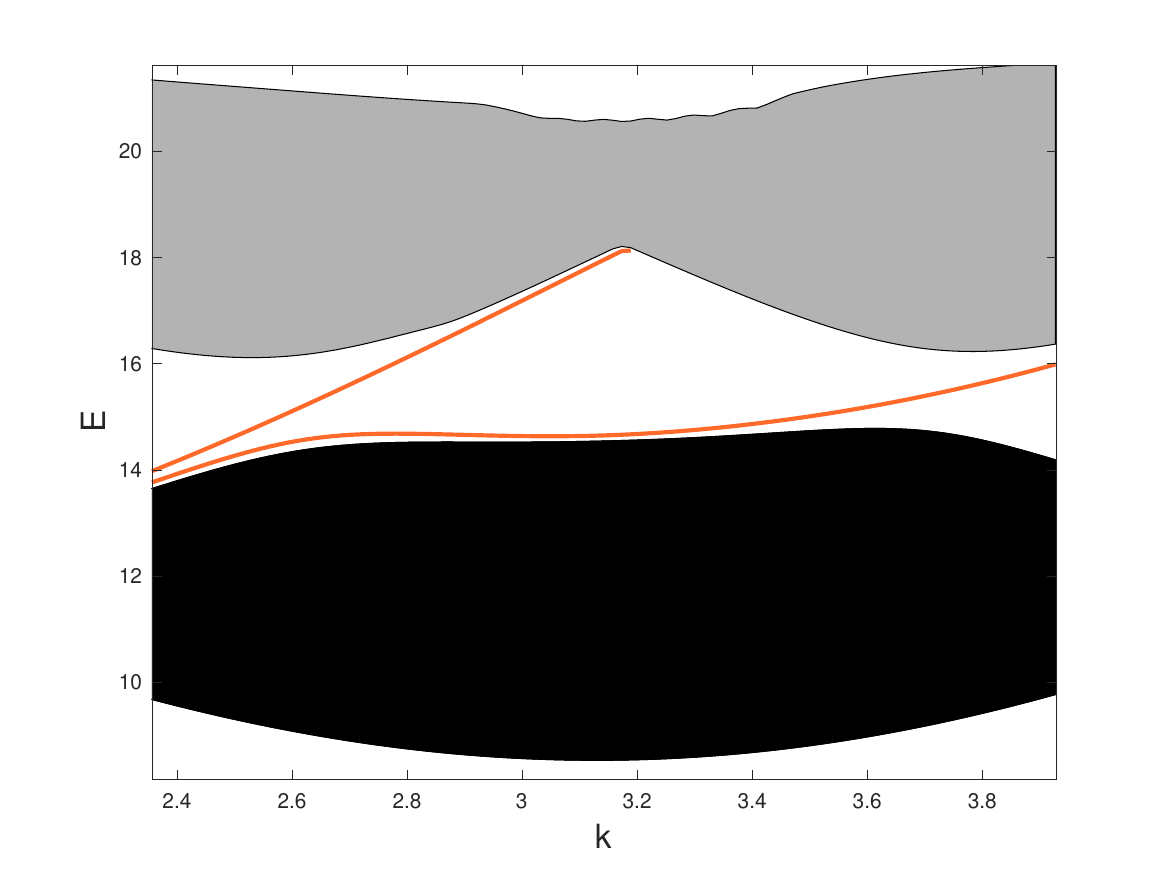}
    \caption{Clockwise from top left: Edge state curves which traverse the bulk band gap for Hamiltonian with undeformed square lattice potential $H^\delta$ \eqref{eq:def_sql-edge} ($\phi = 0$) and for the deformed Hamiltonian $T_* H^\delta$ from \eqref{eq:def_dfm-edge} with tilt angles ${\phi = \pi/200}$, $\pi/100$, and $\pi/50$ in a neighborhood of $k=\pi$. For $\phi = 0$, we take $\delta = .1$ and $A_\delta = \delta^2 \tanh( 10 \delta x_1) A(x_1,x_2)$.  For $\phi > 0$, we have take $\delta = .01$ and correspondingly $A_\delta = \delta \tanh(10 \delta x_1) A(x_1,x_2)$ in order to have the defect edge fully in our discretized region $x_1 \in [-30,30]$.}
    \label{fig:newksweep_tilted}
\end{figure}

\smallskip

\section{Perspective and an open problem on effective Hamiltonians}
\label{sec:an_open_prob}

We have presented an asymptotic construction of edge states in square lattice media and weakly deformed square lattice media for a class of self-adjoint Schr\"odinger Hamiltonians in which complex-conjugation (time-reversal) symmetry is broken. In both undeformed and deformed cases, we construct multiple-scale approximations of two edge state eigenvalue curves which traverse the bulk band gap and the associated eigenstates. These results are consistent with the bulk-edge correspondence principle, stating that the number of such gap-traversing curves is equal to  ${\pm 2}$, the difference of bulk Chern numbers.

Our multiple-scale construction shows that these edge state curves are seeded by discrete spectra of effective edge Hamiltonians, whose spectra are blow-ups of energy-parallel quasimomentum neighborhoods of band degeneracies. For the undeformed square lattice case, where there is an isolated quadratic band degeneracy, the effective edge Hamiltonian is a matrix Schr\"odinger operator $\mathfrak{S}(\kappa)$, whith ${\kappa \smallin \R}$, having two gap-traversing eigenvalue curves ${\kappa \mapsto \Omega_\pm(\kappa)}$. For the deformed case, where there are two nearby distinct conical degeneracies, there are two 
effective edge Hamiltonians, the Dirac operators ${\mathfrak{D}}^\pm(\kappa)$, each contributing one eigenvalue curve traversing the band gap (in the same direction as in the undeformed case).

While there is a good understanding of the gap-traversing eigenvalue curves of the family ${\kappa \mapsto \cancel{\mathfrak{D}}^\pm(\kappa)}$, much less is understood about those of ${\kappa \mapsto \mathfrak{S}(\kappa)}$. Topological arguments discussed earlier ensure the existence of two gap-traversing curves, provided a band gap exists. Our analytical results: multiple-scale approximations, and analysis of effective Hamiltonians, together with numerical simulations, provide detailed complementary information (e.g., the character of eigenstates, apparent monotonicity of edge curves, etc.) which is not accessible to arguments rooted in the topology of the space of Fredholm operators.
 
A deeper analytical understanding would more directly bridge the topological perspective of the bulk-edge correspondence principle with constructive analytical methods. This could be achieved, for example, by  ``carrying'' the topological information through an asymptotic analysis of resolvents (and spectral projections) of $\mathfrak{S}(\kappa)$ and $\cancel{\mathfrak{D}}^\pm(\kappa)$, and to see such information reflected in the gap-traversing eigenvalue curves of both undeformed and deformed effective edge Hamiltonians. However, deforming between the respective asymptotic regimes in which $\mathfrak{S}(\kappa)$ and $\cancel{\mathfrak{D}}^\pm(\kappa)$ govern is subtle due to the different scalings under which these effective operators arise.

\smallskip

\appendix

\section{Band structure degeneracies of bulk Hamiltonians}
\label{apx:par-bulk}

\setcounter{equation}{0}
\setcounter{figure}{0}

\subsection{Quadratic band degeneracies}
\label{apx:quad}

Denote $L^2_{\bm M} \equiv L^2_{\bm M}(\R^2/\Z^2)$. Since $\mathcal{R}$ maps $L^2_{\bm M}$ to itself, and $\mathcal{R}$ is a normal operator, $L^2_{\bm M}$ decomposes into a direct sum of eigenspaces of $\mathcal{R}$. We make this decomposition explicit: Since ${\mathcal{R}^4 = I}$, its eigenvalues $\varsigma$ satisfy ${\varsigma^4 = 1}$ and are given by $\{ 1, \, i, \, -1, \, -i\}$. Thus, we have the orthogonal decomposition
\begin{equation}
\label{eq:L2M-dsum}
L^2_{\bm M} = L^2_{{\bm M}, \, 1} \oplus L^2_{{\bm M}, \, i} \oplus L^2_{{\bm M}, \, -1} \oplus L^2_{{\bm M}, \, -i},
\end{equation}
where
\begin{equation}
\label{eq:def-L2M-rot}
L^2_{{\bm M}, \, \varsigma} \equiv \bigl \{ f \smallin L^2_{\bm M} : \mathcal{R}[f]({\bm x}) = \varsigma f({\bm x}) \ {\rm a.e.} \ {\bm x} \smallin \R^2 \bigr \} .
\end{equation}
Note that $\mathcal{PC}$ maps $L^2_{{\bm M}, i}$ to $L^2_{{\bm M}, -i}$.

A {\it quadratic band degeneracy point} is an energy-quasimomentum pair ${(E_S, {\bm M})}$ for which the following structure of the corresponding Floquet-Bloch eigenspace holds:

\begin{enumerate}
\renewcommand{\theenumi}{Q\arabic{enumi}}
\item \label{itm:quad-dgn-1} $E_S$ is a multiplicity two $L^2_{\bm M}$ eigenvalue of $H_V$,
\item \label{itm:quad-dgn-2} $E_S$ is a simple $L^2_{{\bm M},+i}$ eigenvalue of $H_V$ with normalized eigenstate $\Phi^{\bm M}_1$,
\item \label{itm:quad-dgn-3} $E_S$ is a simple $L^2_{{\bm M},-i}$ eigenvalue of $H_V$ with normalized eigenstate ${\Phi^{\bm M}_2 = \mathcal{PC}[\Phi^{\bm M}_1]}$,
\item \label{itm:quad-dgn-4} $E_S$ is neither an $L^2_{{\bm M},+1}$ eigenvalue nor an $L^2_{{\bm M},-1}$ eigenvalue of $H_V$.
\end{enumerate}

\noindent Introduce orthogonal projections ${\Pi^\parallel = \Phi^{\bm M}_1 \langle \Phi^{\bm M}_1, \, \cdot \rangle + \Phi^{\bm M}_2 \langle \Phi^{\bm M}_2, \, \cdot \rangle }$ and ${\Pi^\perp = 1 - \Pi^\parallel}$, and define the resolvent ${\mathscr{R}(E_S) = \Pi^\perp (H_V - E_S)^{-1} \Pi^\perp}$, which maps ${{\ker}_{L^2_{\bm M}}(H_V - E_S)^\perp \to H^2_{\bm M}}$. Further, define the parameters
\begin{equation}
\label{eq:def-a-par}
\begin{aligned}
\alpha^{\bm M}_0 & \equiv \langle -2i \partial_1 \Phi^{\bm M}_1, \, \mathscr{R}(E_S) (-2i \partial_1) \Phi^{\bm M}_1 \rangle , \\
\alpha^{\bm M}_1 & \equiv \langle -2i \partial_1 \Phi^{\bm M}_1, \, \mathscr{R}(E_S) (-2i \partial_2) \Phi^{\bm M}_2 \rangle , \\
\alpha^{\bm M}_2 & \equiv \langle -2i \partial_1 \Phi^{\bm M}_1, \, \mathscr{R}(E_S) (-2i \partial_1) \Phi^{\bm M}_2 \rangle .
\end{aligned}
\end{equation}
In addition to properties \ref{itm:quad-dgn-1} - \ref{itm:quad-dgn-4}, we assume the following non-degeneracy condition:

\begin{enumerate}
\renewcommand{\theenumi}{Q\arabic{enumi}}
\setcounter{enumi}{4}
\item \label{itm:quad-dgn-5} (Non-degeneracy condition.) ${\alpha^{\bm M}_1 \neq 0}$ and ${\alpha^{\bm M}_2 \neq 0}$.
\end{enumerate}

\noindent In \cite{keller2018spectral, keller2020erratum}, it was shown that properties \ref{itm:quad-dgn-1} - \ref{itm:quad-dgn-5} are sufficient to obtain the locally quadratic character of the degeneracy. We conclude this section by recording several identities convenient for local analysis.

\begin{proposition}
\label{prop:quad-1-zero}
For any ${{\bm u} \smallin \R^2}$ and $j$, ${k = 1}$, $2$, we have
\begin{equation}
{\bm u} \cdot \langle \Phi^{\bm M}_j, \, -2i \nabla \Phi^{\bm M}_k \rangle = 0 .
\end{equation}
\end{proposition}

\begin{proof}
This follows directly from \cite[Proposition 4.10]{keller2018spectral}.
\end{proof}

\begin{proposition}
\label{prop:quad-2}
For any ${{\bm u} \smallin \R^2}$ and $j$, ${k = 1}$, $2$, we have
\begin{equation}
{\bm u} \cdot \langle -2i \nabla \Phi^{\bm M}_j , \, \mathscr{R}(E_S) (-2i \nabla) \Phi^{\bm M}_k \rangle {\bm u} = \alpha^{\bm M}_0 |{\bm u}|^2 (\sigma_0)_{j, k} + \alpha^{\bm M}_1 ({\bm u} \cdot \sigma_1 {\bm u}) (\sigma_1)_{j, k} + \alpha^{\bm M}_2 ({\bm u} \cdot \sigma_3 {\bm u}) (\sigma_2)_{j, k} ,
\end{equation}
where $\alpha^{\bm M}_0$, $\alpha^{\bm M}_1$, and ${\alpha^{\bm M}_2 \smallin \R}$ are defined in \eqref{eq:def-a-par}.
\end{proposition}

\begin{proof}
This follows as a summary of the discussion in \cite[Section 4.1.3]{keller2018spectral} and \cite{keller2020erratum}.
\end{proof}

\begin{proposition}
\label{prop:quad-breakC}
Suppose $W$ is self-adjoint and anticommutes with $\mathcal{PC}$. Then
\begin{equation}
\langle \Phi^{\bm M}_j, \, W \Phi^{\bm M}_k \rangle = \vartheta^{\bm M} (\sigma_3)_{j, k}
\end{equation}
where ${\vartheta^{\bm M} \equiv \langle \Phi^{\bm M}_1, \, W \Phi^{\bm M}_1 \rangle}$.
\end{proposition}

\begin{proof}
Since $W$ is self-adjoint, the matrix with entries ${\langle \Phi^{\bm M}_j, \, W \Phi^{\bm M}_k \rangle}$, where $j$, ${k = 1}$, $2$, is self-adjoint. Hence
\begin{equation}
\langle \Phi^{\bm M}_1 \, W \Phi^{\bm M}_2 \rangle = \langle \mathcal{PC}[\Phi^{\bm M}_1] \, \mathcal{PC}[W \Phi^{\bm M}_2] \rangle^* = \langle \Phi^{\bm M}_2 \, -W \Phi^{\bm M}_1 \rangle^* = - \langle \Phi^{\bm M}_1 , \, W \Phi^{\bm M}_2 \rangle ,
\end{equation}
implying ${\langle \Phi^{\bm M}_1 , \, W \Phi^{\bm M}_2 \rangle = 0}$, and
\begin{equation}
\langle \Phi^{\bm M}_1 \, W \Phi^{\bm M}_1 \rangle = \langle \mathcal{PC}[\Phi^{\bm M}_1] \, \mathcal{PC}[W \Phi^{\bm M}_1] \rangle^* = \langle \Phi^{\bm M}_2 \, -W \Phi^{\bm M}_2 \rangle^* = - \langle \Phi^{\bm M}_2 , \, W \Phi^{\bm M}_2 \rangle .
\end{equation}
The result then follows from the definition \eqref{eq:def-pauli}.
\end{proof}

\subsection{Conical degeneracies (Dirac points)}
\label{apx:dir}

Further, $L$ commutes with $\mathcal{PC}$ (Definition \ref{eq:sym-op}) and, moreoever, for any ${{\bm k} \smallin \mathcal{B}}$, $\mathcal{PC}$ maps $L^2_{\bm k}$ to itself. Since $\mathcal{PC}$ satisfies ${(\mathcal{PC})^2 = 1}$ its eigenvalues are $\pm 1$. This induces the orthogonal decomposition:
\begin{align}
\label{eq:L2k-PC}
L^2_{\bm k} & = Y_{{\bm k},+1} \oplus Y_{{\bm k},-1}, \quad \text{where} \\
Y_{{\bm k},\varsigma} & \equiv \{ f \smallin L^2_{\bm k} : \mathcal{PC}[f] = \varsigma f \}, \ \ \varsigma = \pm 1.
\end{align}

A {\it Dirac point} is a locally conical touching of consecutive dispersion surfaces. In analogy with the case of quadratic band degeneracy points, Dirac points are a consequence of the following structure of the Floquet-Bloch eigenspace at an energy-quasimomentum pair ${(E_D, \, {\bm D})}$, which we shall assume:
\begin{enumerate}
\renewcommand{\theenumi}{D\arabic{enumi}}
\item \label{itm:dir-pt-1} $E_D$ is a multiplicity two $L^2_{\bm D}$ eigenvalue of $L$, 
\item \label{itm:dir-pt-2} $E_D$ is a simple $Y_{{\bm D}, +1}$ eigenvalue of $L$ with normalized eigenstate $\tilde{\Phi}^{\bm D}_1$,
\item \label{itm:dir-pt-3} $E_D$ is a simple $Y_{{\bm D}, -1}$ eigenvalue of $L$ with normalized eigenstate $\tilde{\Phi}^{\bm D}_2$.
\end{enumerate}
To obtain an orthonormal basis ${\{ \Phi^{\bm D}_1, \, \Phi^{\bm D}_2 = \mathcal{PC}[\Phi^{\bm D}_1] \}}$, we set
\begin{equation}
\Phi^{\bm D}_1 \equiv \frac{1}{\sqrt{2}} (\tilde{\Phi}^{\bm D}_1 + \tilde{\Phi}^{\bm D}_2) \quad \text{and} \quad \Phi^{\bm D}_2 \equiv \frac{1}{\sqrt{2}} (\tilde{\Phi}^{\bm D}_1 - \tilde{\Phi}^{\bm D}_2).
\end{equation}
Further, we define
\begin{equation}
\label{eq:def-gam-par}
\begin{aligned}
{\bm \gamma}^{\bm D}_0 & \equiv \langle \Phi^{\bm D}_1, \, (T^\mathsf{T} T)^{-1} (-2i \nabla) \Phi^{\bm D}_1 \rangle , \\
{\bm \gamma}^{\bm D}_1 & \equiv {\rm Re} \, \langle \Phi^{\bm D}_1, \, (T^\mathsf{T} T)^{-1} (-2i \nabla) \Phi^{\bm D}_2 \rangle , \\
{\bm \gamma}^{\bm D}_2 & \equiv -{\rm Im} \, \langle \Phi^{\bm D}_1, \, (T^\mathsf{T} T)^{-1} (-2i \nabla) \Phi^{\bm D}_2 \rangle .
\end{aligned}
\end{equation}

\begin{enumerate}
\renewcommand{\theenumi}{D\arabic{enumi}}
\setcounter{enumi}{3}
\item \label{itm:dir-pt-4} (Nondegeneracy condition.) ${{\bm \gamma}^{\bm D}_1 \neq {\bm 0}}$ and ${{\bm \gamma}^{\bm D}_2 \neq {\bm 0}}$.
\end{enumerate}

\noindent In \cite[Theorem 4.2]{chaban2024instability}, it is shown that properties \ref{itm:dir-pt-1} - \ref{itm:dir-pt-4} imply the locally (tilted, elliptical) conical behavior associated with Dirac points.

\begin{proposition}
\label{prop:dir-breakC}
Suppose $W$ is self-adjoint and anticommutes with $\mathcal{PC}$. Then
\begin{equation}
\langle \Phi^{\bm D}_j, \, W \Phi^{\bm D}_k \rangle = \vartheta^{\bm D} (\sigma_3)_{j, k}
\end{equation}
where ${\vartheta^{\bm D} \equiv \langle \Phi^{\bm D}_1, \, W \Phi^{\bm D}_1 \rangle}$.
\end{proposition}

\noindent The proof of Proposition \ref{prop:dir-breakC} is identical to that of Proposition \ref{prop:quad-breakC}.

Now consider a pair of Dirac points ${(E_D, {\bm D}^\pm)}$ with quasimomenta inversion symmetric about ${\bm M}$, i.e., ${\mathcal{P} L^2_{{\bm D}^+}(\R^2/\Lambda) = L^2_{{\bm D}^-}(\R^2/\Lambda)}$. The following symmetry argument yields relationships between the parameters defined above.

\begin{proposition}
\label{prop:gam-sym}
Assume ${\Phi^{{\bm D}^-}_1\! = \mathcal{P}[\Phi^{{\bm D}^+}_1]}$, ${\Phi^{{\bm D}^-}_2\! = \mathcal{P}[\Phi^{{\bm D}^+}_2]}$. Then
\begin{equation}
{\bm \gamma}^{{\bm D}^-}_\ell = -{\bm \gamma}^{{\bm D}^+}_\ell , \quad \ell = 0, \, 1, \, 2 .
\end{equation}
\end{proposition}

\begin{proof}
Since $\mathcal{P}$ is unitary, it follows that
\begin{align}
\langle \Phi^{{\bm D}^-}_j, \, (T^\mathsf{T} T)^{-1} (-2i \nabla) \Phi^{{\bm D}^-}_k \rangle & = \langle \mathcal{P}[\Phi^{{\bm D}^+}_j], \, \mathcal{P}[(T^\mathsf{T} T)^{-1} (-2i \nabla) \Phi^{{\bm D}^-}_k] \rangle \\
& = - \langle \mathcal{P}[\Phi^{{\bm D}^-}_j], \, (T^\mathsf{T} T)^{-1} (-2i \nabla) \mathcal{P}[\Phi^{{\bm D}^-}_k] \rangle \nonumber \\
& = - \langle \Phi^{{\bm D}^+}_j, \, (T^\mathsf{T} T)^{-1} (-2i \nabla) \Phi^{{\bm D}^+}_k \rangle , \nonumber
\end{align}
for any $j$, ${k = 1}$, $2$. The result then follows immediately from the definitions \eqref{eq:def-gam-par}.
\end{proof}

\smallskip

\section{Proofs of Propositions in Sections \ref{sec:multi-sql} and \ref{sec:multi-dfm}}
\label{apx:pf_multi}

\setcounter{equation}{0}
\setcounter{figure}{0}

In Sections \ref{sec:multi-sql} and \ref{sec:multi-dfm}, we construct approximate eigenpair solutions to the edge state eigenvalue problems for ${H_{\rm edge} = H^\delta}$ \eqref{eq:sql-edge-evp} and ${H_{\rm edge} = T_* H^\delta}$ \eqref{eq:dfm-edge-evp} via multiple-scale analysis. In both cases, our eigenpair ansatz consists of expansions in the small parameter ${\delta > 0}$, which generates a hierarchy of equations, each associated with a power of $\delta$, to be solved in order. In this appendix, we prove equivalent solvability conditions for each of these equations.

\subsection{Proof of Propositions in Section \ref{sec:multi-sql}.}
\label{apx:pf_multi-sql}

The order $\delta^1$ equation \eqref{eq:multi-sql-1} is solvable if and only if the inhomogeneous term is orthogonal to the subspace ${\C \Phi^{\bm M}_1 \oplus \C \Phi^{\bm M}_2}$. Proposition \ref{prop:multi-sql-1-sol} establishes an equivalent solvability condition. \\

\noindent {\it Proof of Proposition \ref{prop:multi-sql-1-sol}.} For ${j = 1}$, $2$,
\begin{align}
0 & = \langle \Phi^{\bm M}_j, \, (E^{(1)}(\kappa) + (-i \partial_X \boldsymbol{\mathfrak{K}}_2 + \kappa \boldsymbol{\mathfrak{K}}_1) \cdot 2i \nabla) \psi^{(0)}(\cdot, X; \kappa) \rangle \\
& = \langle \Phi^{\bm M}_j, \, (E^{(1)}(\kappa) + (-i \partial_X \boldsymbol{\mathfrak{K}}_2 + \kappa \boldsymbol{\mathfrak{K}}_1) \cdot 2i \nabla) a^{(0)}_k(X; \kappa) \Phi^{\bm M}_k \rangle \nonumber \\
& = E^{(1)}(\kappa) \langle \Phi^{\bm M}_j, \, \Phi^{\bm M}_k \rangle a^{(0)}_k(X; \kappa) - {\bm P}_X(\kappa) \cdot \langle \Phi^{\bm M}_j, \, -2i \nabla \Phi^{\bm M}_k \rangle a^{(0)}_k(X; \kappa) , \nonumber
\end{align}
where summation over repeated indices is implied. By Proposition \ref{prop:quad-1-zero}, for $j$, ${k = 1}$, $2$,
\begin{equation}
{\bm P}_X(\kappa) \cdot \langle \Phi^{\bm M}_j, \, -2i \nabla \Phi^{\bm M}_k \rangle = 0 .
\end{equation}
Therefore, for ${a^{(0)}_j(X; \kappa) \neq 0}$, we must have ${E^{(1)}(\kappa) = 0}$. \qed \\

The order $\delta^2$ equation \eqref{eq:multi-sql-1} is again solvable if and only if the inhomogeneous term is orthogonal to ${\C \Phi^{\bm M}_1 \oplus \C \Phi^{\bm M}_2}$. Proposition \ref{prop:multi-sql-2-sol} establishes an equivalent solvability condition. \\

\noindent {\it Proof of Proposition \ref{prop:multi-sql-2-sol}.} For ${j = 1}$, $2$, 
\begin{align}
0 & = \langle \Phi^{\bm M}_j, \, (E^{(1)}(\kappa) + (-i \partial_X \boldsymbol{\mathfrak{K}}_2 + \kappa \boldsymbol{\mathfrak{K}}_1) \cdot 2i \nabla) \psi^{(1)}(\cdot, X; \kappa) \\ 
& \qquad + (E^{(2)}(\kappa) - (-i \partial_X \boldsymbol{\mathfrak{K}}_2 + \kappa \boldsymbol{\mathfrak{K}}_1)^2 - \nabla \cdot \chi(X) A \sigma_2 \nabla) \psi^{(0)}(\cdot, X; \kappa) \rangle \nonumber \\
& = \langle \Phi^{\bm M}_j, \, (E^{(1)}(\kappa) + (-i \partial_X \boldsymbol{\mathfrak{K}}_2 + \kappa \boldsymbol{\mathfrak{K}}_1) \cdot 2i \nabla) (\mathscr{R}(E_\star) (-i \partial_X \boldsymbol{\mathfrak{K}}_2 + \kappa \boldsymbol{\mathfrak{K}}_1) \cdot 2i \nabla a^{(0)}_k(X; \kappa) \Phi^{\bm M}_k + a^{(1)}_k(X; \kappa) \Phi^{\bm M}_k) \nonumber \\
& \qquad + (E^{(2)}(\kappa) - (-i \partial_X \boldsymbol{\mathfrak{K}}_2 + \kappa \boldsymbol{\mathfrak{K}}_1)^2 - \nabla \cdot \chi(X) A \sigma_2 \nabla) a^{(0)}_k(X; \kappa) \Phi^{\bm M}_k \rangle \nonumber \\
& = {\bm P}_X(\kappa) \cdot \langle -2i \nabla \Phi^{\bm M}_j, \, \mathscr{R}(E_\star) (-2i \nabla) \Phi^{\bm M}_k \rangle {\bm P}_X(\kappa) a^{(0)}_k(X; \kappa) - {\bm P}_X(\kappa) \cdot \langle \Phi^{\bm M}_j, \, -2i \nabla \Phi^{\bm M}_k \rangle a^{(1)}_k(X; \kappa) \nonumber \\
& \qquad + E^{(2)}(\kappa) \langle \Phi^{\bm M}_j, \, \Phi^{\bm M}_k \rangle a^{(0)}_k(X; \kappa) - {\bm P}_X(\kappa)^2 \langle \Phi^{\bm M}_j, \, \Phi^{\bm M}_k \rangle a^{(0)}_k(X; \kappa) - \langle \Phi^{\bm M}_j, \, \nabla \cdot A \sigma_2 \nabla \Phi^{\bm M}_k \rangle \chi(X) a^{(0)}_k(X; \kappa) . \nonumber
\end{align}
By Proposition \ref{prop:quad-2}, for $j$, ${k = 1}$, $2$,
\begin{align}
& {\bm P}_X(\kappa) \cdot \langle -2i \nabla \Phi^{\bm M}_j, \, \mathscr{R}(E_\star) (-2i \nabla) \Phi^{\bm M}_k \rangle {\bm P}_X(\kappa) \\
& \qquad = \alpha^{\bm M}_0 {\bm P}_X(\kappa)^2 I_{j, k} + \alpha^{\bm M}_1 ({\bm P}_X(\kappa) \cdot \sigma_1 {\bm P}_X(\kappa)) (\sigma_1)_{j, k} + \alpha^{\bm M}_2({\bm P}_X(\kappa) \cdot \sigma_3 {\bm P}_X(\kappa)) (\sigma_2)_{j, k} . \nonumber
\end{align}
Further, by Proposition \ref{prop:quad-1-zero},
\begin{equation}
{\bm P}_X(\kappa) \cdot \langle \Phi^{\bm M}_j, \, 2i \nabla \Phi^{\bm M}_k \rangle = 0 .
\end{equation}
Finally, applying Proposition \ref{prop:quad-breakC}, we have
\begin{equation}
\langle \Phi^{\bm M}_j, \, \nabla \cdot A \sigma_2 \nabla \Phi^{\bm M}_k \rangle = \vartheta^{\bm M} (\sigma_3)_{j, k} .
\end{equation}
Hence,
\begin{align}
& E^{(2)} a^{(0)}_j(X; \kappa) = ((1 - \alpha^{\bm M}_0) {\bm P}_X(\kappa)^2 I_{j, k} - \alpha^{\bm M}_1 ({\bm P}_X(\kappa) \cdot \sigma_1 {\bm P}_X(\kappa)) (\sigma_1)_{j, k} \\
& \qquad - \alpha^{\bm M}_2({\bm P}_X(\kappa) \cdot \sigma_3 {\bm P}_X(\kappa)) (\sigma_2)_{j, k} + \vartheta^{\bm M} \chi(X) (\sigma_3)_{j, k}) a^{(0)}_k(X; \kappa) , \nonumber
\end{align}
which is the eigenvalue problem \eqref{eq:sql-edge-eff-evp}. \qed

\subsection{Proofs of Propositions in Section \ref{sec:multi-dfm}.}
\label{apx:pf_multi-dfm}

The order $\delta^1$ equation \eqref{eq:multi-dfm-1} is solvable if and only if the inhomogeneous term is orthogonal to the kernel of $T_* H - E_*$, which is ${\C \Phi^{{\bm D}^+}_1 \oplus \C \Phi^{{\bm D}^+}_2}$. Proposition \ref{prop:multi-dfm-1-sol} establishes an equivalent solvability condition. \\

\noindent {\it Proof of Proposition \ref{prop:multi-dfm-1-sol}.} For ${j = 1}$, $2$, 
\begin{align}
0 & = \langle \Phi^{{\bm D}^+}_j, \, (E^{(1)}(\kappa) + (-i \partial_X \boldsymbol{\mathfrak{K}}_2 + \kappa \boldsymbol{\mathfrak{K}}_1) \cdot (T^\mathsf{T} T)^{-1} (2i \nabla) - \nabla \cdot \chi(X) \det(T^{-1}) A \sigma_2 \nabla) \psi^{(0)}(\cdot, X; \kappa) \rangle \\
& = \langle \Phi^{{\bm D}^+}_j, \, (E^{(1)}(\kappa) + (-i \partial_X \boldsymbol{\mathfrak{K}}_2 + \kappa \boldsymbol{\mathfrak{K}}_1) \cdot (T^\mathsf{T} T)^{-1} (-2i \nabla) - \nabla \cdot \chi(X) \det(T^{-1}) A \sigma_2 \nabla) b^{(0)}_k(X; \kappa) \Phi^{{\bm D}^+}_k \rangle \nonumber \\
& = E^{(1)}(\kappa) \langle \Phi^{{\bm D}^+}_j, \, \Phi^{{\bm D}^+}_k \rangle b^{(0)}_k(X; \kappa) - {\bm P}_X(\kappa) \cdot \langle \Phi^{{\bm D}^+}_j, \, (T^\mathsf{T} T)^{-1} (-2i \nabla) \Phi^{{\bm D}^+}_k \rangle b^{(0)}_k(X; \kappa) \nonumber \\
& \qquad - \det(T^{-1}) \langle \Phi^{{\bm D}^+}_j, \, \nabla \cdot A \sigma_2 \nabla \Phi^{{\bm D}^+}_k \rangle \chi(X) b^{(0)}_k(X; \kappa).\nonumber 
\end{align}
For $j$, ${k = 1}$, $2$, by Proposition \ref{prop:gam-sym},
\begin{align}
& {\bm P}_X(\kappa) \cdot \langle \Phi^{{\bm D}^+}_j, \, (T^\mathsf{T} T)^{-1} (-2i \nabla) \Phi^{{\bm D}^+}_k \rangle \\
& \qquad = ({\bm P}_X(\kappa) \cdot {\bm \gamma}^{{\bm D}^+}_0) I_{j, k} + ({\bm P}_X(\kappa) \cdot {\bm \gamma}^{{\bm D}^+}_1) (\sigma_1)_{j, k} + ({\bm P}_X(\kappa) \cdot {\bm \gamma}^{{\bm D}^+}_2) (\sigma_2)_{j, k} . \nonumber
\end{align}
Further, by Proposition \ref{prop:dir-breakC}, we have
\begin{equation}
\det(T^{-1}) \langle \Phi^{{\bm D}^+}_j, \, \nabla \cdot A \sigma_2 \nabla \Phi^{{\bm D}^+}_k \rangle = \vartheta^{{\bm D}^+} (\sigma_3)_{j, k} .
\end{equation}
Hence, 
\begin{align}
E^{(1)}(\kappa) b^{(0)}_j(X; \kappa) & = (({\bm P}_X(\kappa) \cdot {\bm \gamma}^{{\bm D}^+}_0) I_{j, k} + ({\bm P}_X(\kappa) \cdot {\bm \gamma}^{{\bm D}^+}_1) (\sigma_1)_{j, k} + ({\bm P}_X(\kappa) \cdot {\bm \gamma}^{{\bm D}^+}_2) (\sigma_2)_{j, k} \\
& \qquad + \vartheta^{{\bm D}^+} \chi(X) (\sigma_3)_{j, k}) b^{(0)}_k(X; \kappa) , \nonumber
\end{align}
which is the eigenvalue problem \eqref{eq:dfm-edge-eff-evp}. \qed

\smallskip

\section{Essential spectrum of $\mathfrak{S}(\kappa)$; Proof of Proposition \ref{prop:spec-sql-ess}}
\label{apx:spec-sql-ess}

\setcounter{equation}{0}
\setcounter{figure}{0}

In this appendix, we prove Proposition \ref{prop:spec-sql-ess}, which characterizes the essential spectrum of $\mathfrak{S}(\kappa)$ for ${\kappa \smallin \R}$. \\

\noindent {\it Proof of Proposition \ref{prop:spec-sql-ess}.} Recall the matrix Schr\"{o}dinger operator $\mathfrak{S}(\kappa)$ \eqref{eq:sql-edge-eff_2}:
\begin{equation}
\label{eq:sql-edge-eff_2_supp}
\mathfrak{S}(\kappa) = (1 - \alpha_0) (P_X^2 + \kappa^2) I + \alpha_1 (2 P_X \kappa) \sigma_1 + \alpha_2 (-P_X^2 + \kappa^2) \sigma_2 + \vartheta \chi(X) \sigma_3 .
\end{equation}
For $\chi(X)$ a domain wall function (i.e., ${\chi(X) \to \pm 1}$ sufficiently rapidly as ${X \to \pm \infty}$), its essential spectrum is determined by the spectrum of the constant-coefficient operators
\begin{equation}
\label{eq:sql-edge-eff_2_mod_supp}
\mathfrak{S}_\pm(\kappa) = (1 - \alpha_0) (P_X^2 + \kappa^2) I + \alpha_1 (2 P_X \kappa) \sigma_1 + \alpha_2 (-P_X^2 + \kappa^2) \sigma_2 \pm \vartheta \sigma_3 .
\end{equation}
(Note, for ${\chi(X) \to +1}$ as ${X \to \pm \infty}$, only the spectrum of $\mathfrak{S}_+(\kappa)$ is relevant.) Their Fourier symbols are
\begin{equation}
\label{eq:sql-edge-eff_2_FT_supp}
\widehat{\mathfrak{S}}_\pm(\xi; \kappa) = (1 - \alpha_0) (\xi^2 + \kappa^2) I + \alpha_1 (2 \xi \kappa) \sigma_1 + \alpha_2 (-\xi^2 + \kappa^2) \sigma_2 \pm \vartheta \sigma_3 .
\end{equation}
For each fixed ${\kappa \smallin \R}$, both matrices $\widehat{\mathfrak{S}}_\pm(\xi; \kappa)$ have the same eigenvalues ${\xi \mapsto \Omega_\pm(\xi; \kappa)}$, given by
\begin{equation}
\Omega_\pm(\xi; \kappa) \equiv (1 - \alpha_0) (\xi^2 + \kappa^2) \pm \sqrt{\alpha_1^2 (2 \xi \kappa)^2 + \alpha_2^2 (-\xi^2 + \kappa^2)^2 + \vartheta^2} .
\end{equation}
Hence, the essential spectrum of $\mathfrak{S}(\kappa)$ is the union of the images of ${\xi \mapsto \Omega_\pm(\xi; \kappa)}$:
\begin{equation}
{\rm spec}_{\rm ess}(\mathfrak{S}(\kappa)) = \Omega_-(\R; \kappa) \cup \Omega_+(\R; \kappa) .
\end{equation}
By continuity, $\Omega_-(\R; \kappa)$ and $\Omega_+(\R; \kappa)$ are intervals. 

Now take ${\alpha_0 = 1}$, ${\vartheta \neq 0}$. Fix ${\kappa \smallin \R}$, and observe that, as ${|\xi| \to +\infty}$,
\begin{equation}
\label{eq:sql-omega-asym}
\Omega_\pm(\xi; \kappa) \to \pm |\alpha_2| \xi^2 .
\end{equation}
Since ${\alpha_2 \neq 0}$, the images of the functions in \eqref{eq:sql-omega-asym} yield the unbounded components of ${\rm spec}_{\rm ess}(\mathfrak{S}(\kappa))$ \eqref{eqn:specwindows}. Further, by \eqref{eq:sql-omega-asym}, $\Omega_-(\xi; \kappa)$ is bounded above and $\Omega_+(\xi; \kappa)$ is bounded below. We therefore define
\begin{equation}
\eta_-(\kappa) \equiv \sup_{\xi \smallin \R} \, \Omega_-(\xi; \kappa) \quad \text{and} \quad \eta_+(\kappa) \equiv \inf_{\xi \smallin \R} \, \Omega_+(\xi; \kappa) .
\end{equation}
Since ${\Omega_-(\xi; \kappa) \leq -|\vartheta|}$, ${\Omega_+(\xi; \kappa) \geq |\vartheta|}$, and ${\vartheta \neq 0}$, we have finally
\begin{equation}
\min_{\kappa \smallin \R} \, \eta_+(\kappa) - \max_{\kappa \smallin \R} \, \eta_-(\kappa) \geq 2 |\vartheta| > 0 .
\end{equation}
The proof is complete. \qed

\smallskip

\section{Proofs of Propositions and Theorems in Section \ref{sec:spec-sql-disc}}
\label{apx:sql-edge-eff-disc}

\setcounter{equation}{0}
\setcounter{figure}{0} 

\subsection{No zero energy bound states; Proof of Proposition \ref{prop:sql-eff-no-zero}}
\label{apx:sql-eff-no-zero}

{\it Proof of Proposition \ref{prop:sql-eff-no-zero}.} We argue by contradiction: Suppose that, for some nontrivial ${\psi \smallin H^2(\R; \C^2)}$,
\begin{equation}
\mathfrak{S}(0) \psi(X) = (-\alpha_2 P^2_X \sigma_2 + \vartheta \chi(X) \sigma_3) \psi(X) = 0 .
\end{equation}
Left-multiplying this equation by $\sigma_2$ yields
\begin{equation}
(-\alpha_2 P^2_X I + i \vartheta \chi(X) \sigma_1) \psi(X) = 0 .
\end{equation}
Taking the $L^2(\R;\C^2)$ inner product with $\psi$, and integrating by parts, we have
\begin{equation}
\label{eq:en-id-0}
-\alpha_2 \langle P_X \psi, \, P_X \psi \rangle + i \vartheta \langle \psi, \, \chi \sigma_1 \psi \rangle = 0 .
\end{equation}
Since the multiplication operator $\chi \sigma_1$ is self-adjoint, the second term of \eqref{eq:en-id-0} is purely imaginary. Hence, taking the real part of \eqref{eq:en-id-0} yields ${\lVert P_X \psi \rVert^2 = 0}$, implying $\psi = 0$. \qed

\subsection{Condition for the existence of bound states; Proof of Proposition \ref{prop:sql-eff-exist}}
\label{apx:sql-eff-exist}

{\it Proof of Proposition \ref{prop:sql-eff-exist}.} Suppose there exists ${{\bm \psi} \smallin H^2(\R; \C^2)}$ such that ${\langle {\bm \psi}, \, \mathcal{L} {\bm \psi} \rangle < 0}$. By the spectral theorem for self-adjoint operators, it follows that ${{\rm spec}(\mathcal{L}) \cap (-\infty, \, 0)}$ is not empty. Since ${{\rm spec}_{\rm ess}(\mathcal{L}) = [0, \, +\infty)}$, it follows that $\mathcal{L}$ has at least one negative eigenvalue of finite multiplicity.

Let $-\mu^2_\star$ denote a negative eigenvalue of ${\mathcal{L} = \mathfrak{S}(0)^2 - \vartheta^2}$ (by the argument above, there is at least one). The corresponding eigenstate ${\bm \psi}_\star$ satisfies
\begin{equation}
\mathcal{L} {\bm \psi}_\star = (\mathfrak{S}(0)^2 - \vartheta^2) {\bm \psi}_\star = -\mu^2_\star {\bm \psi}_\star .
\end{equation}
Denote ${\Omega_\star \equiv \sqrt{\vartheta^2 - \mu^2_\star}}$, and observe that ${0 < \Omega_\star < |\vartheta|}$. Then
\begin{equation}
\mathfrak{S}(0)^2 {\bm \psi}_\star = (\vartheta^2 - \mu^2_\star) {\bm \psi}_\star = \Omega^2_\star {\bm \psi}_\star, \quad \text{implying} \quad (\mathfrak{S}(0) + \Omega_\star)(\mathfrak{S}(0) - \Omega_\star) {\bm \psi}_\star = 0 .
\end{equation}
Hence, either ${(\Omega_\star, {\bm \psi}_\star)}$ is an eigenpair of $\mathfrak{S}(0)$, or ${(-\Omega_\star, (\mathfrak{S}(0) - \Omega_\star) {\bm \psi}_\star)}$ is an eigenpair of $\mathfrak{S}(0)$.

Without loss of generality, suppose ${(\Omega_+, {\bm \psi}_{+1}) = (\Omega_\star, {\bm \psi}_\star)}$ is an eigenpair of $\mathfrak{S}(0)$. Then, since ${\mathfrak{S}(0) \sigma_1 =}$ ${-\sigma_1 \mathfrak{S}(0)}$, it follows that ${(\Omega_-, {\bm \psi}_{-1}) = (-\Omega_+, \sigma_1 {\bm \psi}_{+1})}$ is also an eigenpair of $\mathfrak{S}(0)$. We have thus produced two eigenvalues $\Omega_{\pm 1}$ satisfying
\begin{equation}
-|\vartheta| < \Omega_- = -\Omega_+ < 0 < \Omega_+ < |\vartheta| .
\end{equation}
The proof is complete. \qed

\subsection{At least two eigenvalues in the spectral gap of $\mathfrak{S}(0)$; Proofs of Theorems \ref{thm:sql-eff-eig_1} and \ref{thm:sql-eff-eig_2}}
\label{apx:sql-eff-eig}

By Proposition \ref{prop:sql-eff-exist}, $\mathfrak{S}(0)$ has at least two eigenvalues in its spectral gap provided there exists ${\psi \smallin H^2(\R; \C^2)}$ such that ${\langle \psi, \, \mathcal{L} \psi \rangle < 0}$, where ${\mathcal{L} = \mathfrak{S}(0)^2 -\vartheta^2}$.

The proofs of Theorems \ref{thm:sql-eff-eig_1} and \ref{thm:sql-eff-eig_2} make use of the following: For arbitrary ${{\bm \psi} \smallin H^2(\R; \C^2)}$,
\begin{equation}
\label{eq:psiLpsi}
\langle {\bm \psi}, \, \mathcal{L} {\bm \psi} \rangle_{L^2(\R; \, \C^2)} = \alpha_2^2 \int_\R |P^2_X \psi(X)|^2 \, {\rm d}X + \vartheta^2 \int_\R v(X) |\psi(X)|^2 \, {\rm d}X - \alpha_2 \vartheta \int_\R \overline{\psi(X)} \cdot \mathcal{A}(X) {\bm \psi}(X) \, {\rm d}X ,
\end{equation}
where 
$ v(X) = \chi(X)^2 - 1$ and $
    \mathcal{A}(X) = (P_X \partial_X\chi(X) + \partial_X\chi(X) P_X) \sigma_1$.
    Below, we'll evaluate $\langle {\bm \psi}, \, \mathcal{L} {\bm \psi} \rangle_{_{L^2(\R; \, \C^2)}}$ on the family of functions: $ {\bm \psi}_h(X) = \sqrt{h} {\bm \phi}(h X)$, where $h>0$, ${\bm \phi}\in \mathcal{S}(\mathbb R)$ and $\|{\bm \phi}\|_{L^2}=1$. 

\subsubsection{Proof of Theorem \ref{thm:sql-eff-eig_1}}
\label{apx:sql-eff-eig_1}

Here we assume $\int V = \int (\chi^2-1)<0$.  Then,
\begin{align}
\langle {\bm \psi}_h, \, \mathcal{L} {\bm \psi}_h \rangle_{L^2(\R; \C^2)} & = \alpha_2^2 h^4 \int_\R |P^2_Y {\bm \phi}(Y)|^2 \, {\rm d}Y + \vartheta^2 h \int_\R \big(\chi(X)^2 - 1\big) |{\bm \phi}(h X)|^2 \, {\rm d}X \nonumber \\
& \qquad - 2 \alpha_2 \vartheta h^2 \, {\rm Re} \biggl( \int_\R \partial_X\chi(X) \overline{{\bm \phi}(h X)} \cdot \sigma_1 \big(P_Y {\bm \phi}\big)(h X) \, {\rm d}X \biggr) . \nonumber \end{align}
The second term is dominant for $h$ small and we have 
\begin{align}
\langle {\bm \psi}_h, \, \mathcal{L} {\bm \psi}_h \rangle_{L^2(\R; \C^2)} = h \biggl( \vartheta^2 |\phi(0)|^2 \int_\R \big(\chi(X)^2 - 1\big) \, {\rm d}X + \mathcal{O}(h) \biggr) \ \ \text{as} \ \ h \to 0, 
\end{align}
which is negative for ${h > 0}$ and sufficiently small. The result now follows from  Proposition \ref{prop:sql-eff-exist}. \qed

\subsubsection{Proof of Theorem \ref{thm:sql-eff-eig_2}}
\label{apx:sql-eff-eig_2}

Here we study  $\mathfrak{S}^\varepsilon(0) = -\alpha_2 P_X^2 \sigma_2 + \vartheta \chi_\varepsilon(X) \sigma_3$, where
 $\chi_\varepsilon(X) \equiv \chi \big(X/\varepsilon\big)$. 
 The operator $\mathcal{L}^\varepsilon=\mathfrak{S}^\varepsilon(0)^2-\vartheta^2$ is given by \eqref{eq:def-calL} with the replacements: $v(X)\to v(X/\varepsilon) =\chi^2(X/\varepsilon)-1$ and $\partial_X\chi(X)\to \varepsilon^{-1}\chi^\prime(X/\varepsilon)$. 
Then,  
\begin{align}
\langle {\bm \psi}_h, \, \mathcal{L}^\varepsilon {\bm \psi}_h \rangle_{L^2(\R; \C^2)} & = \alpha_2^2 h^4 \int_\R |P^2_X {\bm \phi}(X)|^2 \, {\rm d}X + \vartheta^2 h \int_\R \big(\chi^2(X/\varepsilon)-1\big) |{\bm \phi}(h X)|^2 \, {\rm d}X \nonumber \\
& \qquad - 2 \alpha_2 \vartheta h^2 \, {\rm Re} \biggl( \int_\R \varepsilon^{-1}\chi^\prime(X/\varepsilon) \overline{{\bm \phi}(h X)} \cdot \sigma_1 \left(P_Y {\bm \phi}\right)(h X) \, {\rm d}X \biggr) . \nonumber
\end{align}

Note that for $g\in\mathcal{S}(\mathbb  R)$,  fixed $h\in(0,h_0)$ and $\varepsilon\in(0,\varepsilon_0)$ we have
\begin{align*}
   \int_\R \big(\chi^2(X/\varepsilon)-1\big) g(hX)\ dX\ & = \ \varepsilon\Big( g(0)\int_\R \big(\chi^2(Y)-1\big)dY  + \mathcal{O}(\varepsilon h)\Big), \\ 
   \int_\R \frac{1}{\varepsilon}\chi^\prime(X/\varepsilon)\ g(hX)\ dX\ &=\ 2\ g(0) + \mathcal{O}(\varepsilon h).
\end{align*}
Therefore, 
\begin{align}
\langle {\bm \psi}_h, \, \mathcal{L}^\varepsilon {\bm \psi}_h \rangle_{L^2(\R; \C^2)} & = \alpha_2^2 h^4 \int_\R |P^2_X {\bm \phi}(X)|^2 \, {\rm d}X + \vartheta^2 h \varepsilon\Big( |{\bm \phi}(0)|^2\int_\R \big(\chi^2(Y)-1\big)dY  + \mathcal{O}(\varepsilon h)\Big) \nonumber \\
& \qquad - 2 \alpha_2 \vartheta h^2 \, {\rm Re} \biggl( 2\ \overline{{\bm \phi}(0)} \cdot \sigma_1 \left(P_Y {\bm \phi}\right)(0)\ + \mathcal{O}(\varepsilon h) \biggr) . \nonumber
\end{align}

Now fix $\varepsilon\in(0,\varepsilon_0)$. This fixes a domain wall $\chi(X/\varepsilon)$. Choose $h=\varepsilon^{1/2}$, Then, 
\begin{equation}
\langle {\bm \psi}_h, \, \mathcal{L}^\varepsilon {\bm \psi}_h \rangle_{L^2(\R; \C^2)} = 
- 4 \alpha_2 \vartheta\ \varepsilon \, {\rm Re} \Big[  \overline{{\bm \phi}(0)} \cdot \sigma_1 \left(P_Y {\bm \phi}\right)(0)\Big]\ + \mathcal{O}(\varepsilon^{3/2}) .
\label{eq:en-phi}\end{equation}
  Suppose we can choose   ${\bm \phi}$ so that $\alpha_2 \vartheta\ {\rm Re} \Big[  \overline{{\bm \phi}(0)} \cdot \sigma_1 \left(P_Y {\bm \phi}\right)(0)\Big]>0$.
Then, for $\varepsilon_0$ sufficiently small, the expression \eqref{eq:en-phi} is strictly negative for $0<\varepsilon<\varepsilon_0$
and the result would follow from Proposition \ref{prop:sql-eff-exist}. 
 Let  ${\phi(X) = -i[i u(X), \, u(-X)]}^\mathsf{T}$, where $u(X)$ is real-valued, and $C^1$ and decaying sufficiently rapidly. Then, 
 $\alpha_2 \vartheta\ {\rm Re} \Big[  \overline{{\bm \phi}(0)} \cdot \sigma_1 \left(P_Y {\bm \phi}\right)(0)\Big] = 2\alpha_2 \vartheta\ u(0)u^\prime(0)$. By \eqref{itm:quad-dgn-5}, $\alpha_2\vartheta\ne0$. Hence, $\langle {\bm \psi}_h, \, \mathcal{L}^\varepsilon {\bm \psi}_h \rangle_{L^2(\R; \C^2)}$ can be made negative for $\varepsilon$ sufficiently small by choosing $u$ with $\sgn\big(u(0)u'(0)\big)=-\sgn\big(\alpha_2\vartheta\big)$.  \qed

\begin{remark}
We note that the structure ${\phi(X) = [i u(X), \, u(-X)]}^\mathsf{T}$, with ${u \smallin L^2(\R)}$ real-valued, is satisfied by the bound states in the exactly-solvable case; see Supplementary Material \ref{supp:sql-eff-exact}.
\end{remark}

\smallskip

\section{Essential spectra of $\cancel{\mathfrak{D}}^\pm(\kappa)$; Proof of Proposition \ref{prop:spec-dfm-gap}}
\label{apx:dfm-edge-eff-gap}

\setcounter{equation}{0}
\setcounter{figure}{0}

In this appendix, we prove Proposition \ref{prop:spec-dfm-gap}, which characterizes the essential spectra of $\cancel{\mathfrak{D}}^\pm(\kappa)$ for ${\kappa \smallin \R}$. \\

\noindent {\it Proof of Proposition \ref{prop:spec-dfm-gap}.} Recall the pair of Dirac operators $\cancel{\mathfrak{D}}^\pm(\kappa)$ \eqref{eq:dfm-edge-eff_2}:
\begin{equation}
\label{eq:dfm-edge-eff_2_supp}
\cancel{\mathfrak{D}}^\pm(\kappa) = \pm ((a \cdot \sigma) P_X + (b \cdot \sigma) \kappa) + c \chi(X) \sigma_3 .
\end{equation}
For simplicity, we focus on $\cancel{\mathfrak{D}}^+(\kappa)$; the analysis for $\cancel{\mathfrak{D}}^-(\kappa)$ is analogous. Following the methods in Appendix \ref{apx:spec-sql-ess}, the essential spectrum of $\cancel{\mathfrak{D}}^+(\kappa)$ is determined by the spectra of
\begin{equation}
\label{eq:dfm-edge-eff_2_mod_supp}
\cancel{\mathfrak{D}}^+_\pm(\kappa) = ((a \cdot \sigma) P_X + (b \cdot \sigma) \kappa) \pm c \chi(X) \sigma_3 ,
\end{equation}
whose Fourier symbols are
\begin{equation}
\label{eq:dfm-edge-eff_ft_2_supp}
\cancel{\widehat{\mathfrak{D}}}^+_\pm(\xi; \kappa) = ((a \cdot \sigma) \xi + (b \cdot \sigma) \kappa) \pm c \chi(X) \sigma_3 .
\end{equation}
For each fixed ${\kappa \smallin \R}$, both matrices $\cancel{\widehat{\mathfrak{D}}}^+_\pm(\xi; \kappa)$ have the same eigenvalues ${\xi \mapsto \Omega_\pm(\xi; \kappa)}$, given by
\begin{equation}
\label{eq:Dirac_eigs_supp}
\Omega_{\pm}(\xi; \kappa) = a_0 \xi + b_0 \kappa \pm \sqrt{(a_1 \xi + b_1 \kappa)^2 + (a_2 \xi + b_2 \kappa)^2 + c^2} .
\end{equation}
As in Appendix \ref{apx:spec-sql-ess}, we have
\begin{equation}
{\rm spec}_{\rm ess}(\cancel{\mathfrak{D}}^\pm(\kappa)) = \Omega_-(\R; \kappa) \cup \Omega_+(\R; \kappa) .
\end{equation}
We derive a necessary condition for disjoint intervals.

\begin{remark}
\label{rmk:spec-dfm-gap-nec}
{\rm (Necessary condition for a band gap.)}
Fix ${\kappa \smallin \R}$, and observe that, as ${|\xi| \to +\infty}$,
\begin{equation}
\Omega_\pm(\xi; \kappa) \to (a_0 \pm \sqrt{a_1^2 + a_2^2}) |\xi| .
\end{equation}
For $\Omega_\pm(\R; \kappa)$ to be disjoint, the coefficients of the absolute value functions on the right-hand side must be of opposite sign, and hence
\begin{equation}
{-a}_0^2 + a_1^2 + a_2^2 > 0 .
\end{equation}
The images of these absolute value functions account for the unbounded components of ${\rm spec}_{\rm ess}(\cancel{\mathfrak{D}}^+(\kappa))$ \eqref{eqn:specwindowsDirac}.
\end{remark}

\noindent Under the necessary condition of Remark \ref{rmk:spec-dfm-gap-nec}, $\Omega_-(\xi; \kappa)$ is bounded above and $\Omega_+(\xi; \kappa)$ is bounded below. We define
\begin{equation}
\eta_-(\kappa) \equiv \sup_{\xi \smallin \R} \, \Omega_-(\xi; \kappa) \quad \text{and} \quad \eta_+(\kappa) \equiv \inf_{\xi \smallin \R} \, \Omega_+(\xi; \kappa) .
\end{equation}
To show that the necessary condition is sufficient, we study the critical points of $\Omega_\pm(\xi; \kappa)$.

All critical points, of both $\Omega_+(\xi; \kappa)$ and $\Omega_-(\xi; \kappa)$, are among the solutions of ${\partial_\xi \Omega_+(\xi; \kappa) \cdot \partial_\xi \Omega_-(\xi; \kappa) = 0}$. This gives the quadratic equation
\begin{equation}
\xi^2 + \frac{2 (a_1 b_1 + a_2 b_2) \kappa}{a_1^2 + a_2^2} \xi + \frac{(a_1 b_1 + a_2 b_2)^2 \kappa^2 - a_0^2 ((b_1^2 + b_2^2) \kappa^2 + c^2)}{(a_1^2 + a_2^2)({-a}_0^2 + a_1^2 + a_2^2)} = 0 .
\end{equation}
The discriminant is non-negative, yielding real solutions
\begin{equation}
\xi = \xi_\pm(\kappa) \equiv - \frac{(a_1 b_1 + a_2 b_2) \kappa}{a_1^2 + a_2^2} \mp \frac{a_0}{a_1^2 + a_2^2} \sqrt{\frac{(a_1 b_2 - a_2 b_1)^2 \kappa^2 + (a_1^2 + a_2^2) c^2}{{-a}_0^2 + a_1^2 + a_2^2}} .
\end{equation}
We observe that ${\xi = \xi_+(\kappa)}$ is a local minimum of $\Omega_+(\xi; \kappa)$ and ${\xi = \xi_-(\kappa)}$ is a local maximum of $\Omega_-(\xi; \kappa)$. It follows that
\begin{align}
\eta_-(\kappa) & = \Omega_-(\xi_-(\kappa); \kappa) \\
& = \Bigl( - \frac{a_0 (a_1 b_1 + a_2 b_2)}{a_1^2 + a_2^2} + b_0 \Bigr) \kappa - \frac{\sqrt{({-a}_0^2 + a_1^2 + a_2^2)((a_1 b_2 - a_2 b_1)^2 \kappa^2 + (a_1^2 + a_2^2) c^2)}}{a_1^2 + a_2^2} \nonumber , \\
\eta_+(\kappa) & = \Omega_+(\xi_+(\kappa); \kappa) \\
& = \Bigl( - \frac{a_0 (a_1 b_1 + a_2 b_2)}{a_1^2 + a_2^2} + b_0 \Bigr) \kappa + \frac{\sqrt{({-a}_0^2 + a_1^2 + a_2^2)((a_1 b_2 - a_2 b_1)^2 \kappa^2 + (a_1^2 + a_2^2) c^2)}}{a_1^2 + a_2^2} \nonumber ,
\end{align}
and therefore
\begin{equation}
\eta_+(\kappa) - \eta_-(\kappa) = \frac{2 \sqrt{({-a}_0^2 + a_1^2 + a_2^2)((a_1 b_2 - a_2 b_1)^2 \kappa^2 + (a_1^2 + a_2^2) c^2)}}{a_1^2 + a_2^2} .
\end{equation}
The minimum of this function occurs at ${\kappa = 0}$, implying the existence of a band gap:
\begin{equation}
\inf_{\kappa \smallin \R} \bigl( \eta_+(\kappa) - \eta_-(\kappa) \bigr) = \eta_+(0) - \eta_-(0) = 2 |c| \sqrt{\frac{-a_0^2 + a_1^2 + a_2^2}{a_1^2 + a_2^2}} > 0 .
\end{equation}
The proof is complete. \qed

\newpage
\renewcommand\thesection{S\arabic{section}}

\begin{center}
{\bf Supplementary Material}
\end{center}

\setcounter{section}{0}
\setcounter{subsection}{0}

\section{Bound states of $\mathfrak{S}(0)$ in an exactly solvable case}
\label{supp:sql-eff-exact}

\setcounter{equation}{0}
\setcounter{figure}{0}

\allowdisplaybreaks

We study the bound states of $\mathfrak{S}(0)$ in an exactly-solvable case.

\begin{proposition}
\label{prop:bound}
Consider $\mathfrak{S}(0)$ for a vertical edge, with $\alpha_0$, $\alpha_2$, ${\vartheta = 1}$, and  domain wall function ${\chi(X) = {\rm sgn}(X)}$:
\begin{equation*}
\mathfrak{S}(0) = -P_X^2 \sigma_2 + {\rm sgn}(X) \sigma_3 .
\end{equation*}
Then:
\begin{enumerate}
\item The essential spectrum of ${\mathfrak{S}(0)}$ is ${\R \setminus (-1, \, 1)}$.
\item ${\mathfrak{S}(0)}$ has two simple eigenpairs ${(\Omega_{+1},{\bm u}_{+1}(X))}$ and ${(\Omega_{-1} = -\Omega_{+1}, {\bm u}_{-1}(X) = \sigma_1 {\bm u}_{+1}(X))}$, where
\begin{align*}
\Omega_{+1} & = 1/\sqrt{2} , \\
{\bm u}_{+1}(X) & =
\! \begin{bmatrix}
u_{+1, 1}(X) \\
u_{+1, 2}(X)
\end{bmatrix} \!
=
\begin{cases}
A
\! \begin{bmatrix}
i (1 + \sqrt{2})^2 e^{-\omega X} \cos(\omega X - 3\pi/8) \\
-(1 + \sqrt{2}) e^{-\omega X} \sin(\omega X - 3\pi/8)
\end{bmatrix} \!
, & X > 0 , \\[3ex]
A
\! \begin{bmatrix}
i (1 + \sqrt{2}) e^{\omega X} \cos(\omega X - \pi/8) \\
-(1 + \sqrt{2})^2 e^{\omega X} \sin(\omega X - \pi/8)
\end{bmatrix} \!
, & X < 0 .
\end{cases}
\end{align*}
Here, ${\omega = \sqrt[4]{2}/2}$ and $A$ is a normalization constant.
\end{enumerate}
\end{proposition}

\noindent Figure \ref{fig:sql-eff-exact} displays the bound states. \\

\noindent {\it Proof of Proposition \ref{prop:bound}.} The eigenvalue equation ${\mathfrak{S}(0) {\bm u} = \Omega {\bm u}}$, where ${\bm u} = [u_1, u_2]^\mathsf{T}$, is
\begin{equation*}
\begin{bmatrix}
{\rm sgn}(X) & i P_X^2 \\
-i P_X^2 & - {\rm sgn}(X)
\end{bmatrix} \!
\! \begin{bmatrix}
u_1(X) \\
u_2(X)
\end{bmatrix} \!
= 
\! \begin{bmatrix}
{\rm sgn}(X) & - i \partial^2_X \\
i \partial^2_X & - {\rm sgn}(X)
\end{bmatrix} \!
\! \begin{bmatrix}
u_1(X) \\
u_2(X)
\end{bmatrix} \!
= \Omega
\! \begin{bmatrix}
u_1(X) \\
u_2(X)
\end{bmatrix} \! .
\end{equation*}
This is equivalent to the system of first-order equations:
\begin{equation*}
\partial_X {\bm U}(X) =
\! \begin{bmatrix}
0 & 0 & 1 & 0 \\
0 & 0 & 0 & 1 \\
0 & -i ({\rm sgn}(X) + \Omega) & 0 & 0 \\
-i ({\rm sgn}(X) - \Omega) & 0 & 0 & 0
\end{bmatrix} \!
\! {\bm U}(X) , \quad {\bm U}(X) = 
\! \begin{bmatrix}
u_1(X) \\
u_2(X) \\
\partial_X u_1(X) \\
\partial_X u_2(X)
\end{bmatrix} \! .
\end{equation*}

We solve this system separately for ${X > 0}$ and ${X < 0}$: \\

\noindent 1. For ${X > 0}$:
\begin{equation*}
\partial_X {\bm U}(X) =
\! \begin{bmatrix}
0 & 0 & 1 & 0 \\
0 & 0 & 0 & 1 \\
0 & -i(1 + \Omega) & 0 & 0 \\
-i(1 - \Omega) & 0 & 0 & 0
\end{bmatrix} \!
{\bm U}(X) .
\end{equation*}
The right-hand side matrix has eigenvalues
\begin{equation*}
\lambda_{+, +} = e^{i \pi/4} \sqrt[4]{1 - \Omega^2} , \ \ \lambda_{+, -} = -i e^{i \pi/4} \sqrt[4]{1 - \Omega^2} , \ \ \lambda_{-, +} = i e^{i \pi/4} \sqrt[4]{1 - \Omega^2} , \ \ \lambda_{-, -} = -e^{i \pi/4} \sqrt[4]{1 - \Omega^2} .
\end{equation*}
(Note that ${1 - \Omega^2 > 0}$ for $\Omega$ in the spectral gap.) Let ${\bm U}^{(+)}_\pm$ denote the eigenvectors with eigenvalues $\lambda_{-, \pm}$. Then the general solution is
\begin{equation*}
{\bm U}(X) = C^{(+)}_+ {\bm U}^{(+)}_+ e^{\lambda_{-, +} X} + C^{(+)}_- {\bm U}^{(+)}_- e^{\lambda_{-, -} X}, \quad X > 0 ,
\end{equation*}
where $C^{(+)}_\pm$ are constants, to be determined. \\

\noindent 2. For ${X < 0}$:
\begin{equation*}
\partial_X {\bm U}(X) =
\! \begin{bmatrix}
0 & 0 & 1 & 0 \\
0 & 0 & 0 & 1 \\
0 & -i (-1 + \Omega) & 0 & 0 \\
-i (-1 - \Omega) & 0 & 0 & 0
\end{bmatrix} \!
{\bm U}(X) .
\end{equation*}
The right-hand side matrix has the same eigenvalues as when ${X > 0}$. Let ${\bm U}^{(-)}_\pm$ denote the eigenvectors with eigenvalues $\lambda_{+, \pm}$. The general solution is
\begin{equation*}
{\bm U}(X) = C^{(-)}_+ {\bm U}^{(-)}_+ e^{\lambda_{+, +} X} + C^{(-)}_- {\bm U}^{(-)}_- e^{\lambda_{+, -} X}, \quad X < 0 ,
\end{equation*}
where $C^{(-)}_\pm$ are constants, to be determined.

For ${\bm U}(X)$ to be continuous at ${X = 0}$:
\begin{equation*}
C^{(-)}_+ {\bm U}^{(-)}_+ + C^{(-)}_- {\bm U}^{(-)}_- = C^{(+)}_+ {\bm U}^{(+)}_+ + C^{(+)}_- {\bm U}^{(+)}_- .
\end{equation*}
This is a system of four linear equations for the four unknown constants $\{ C^{(-)}_\pm, C^{(+)}_\pm \}$:
\begin{equation*}
[{\bm U}^{(-)}_+, \, {\bm U}^{(-)}_-, \, -{\bm U}^{(+)}_+, \, -{\bm U}^{(+)}_-] \, {\bm C} = {\bm 0} ,
\end{equation*}
where ${{\bm C} = [C^{(-)}_+, C^{(-)}_-, C^{(+)}_+, C^{(+)}_-]^\mathsf{T}}$. The determinant of the left-hand side matrix is
\begin{equation*}
D(\Omega) = \frac{8 (1 - 2 \Omega^2)}{(1 - \Omega^2)^{3/2}} ,
\end{equation*}
which has zeros ${\Omega = \Omega_{\pm 1} \equiv \pm 1/\sqrt{2}}$; these are the eigenvalues of $\mathfrak{S}(0)$ in the spectral gap.

\begin{proposition}
Since ${\sigma_1 \mathfrak{S}(0) = - \mathfrak{S}(0) \sigma_1}$, it is sufficient to consider ${\Omega = \Omega_{+1}}$. 
\end{proposition}

\noindent The kernel of the above matrix, with ${\Omega = \Omega_{+1}}$, is given by $\C {\bm C}_{+1}$, where
\begin{equation*}
{\bm C}_{+1} \equiv [-i (1 + \sqrt{2}), \, e^{i \pi/4} (1 + \sqrt{2}), \, -i e^{i \pi/4}, \, 1]^\mathsf{T} . \\
\end{equation*}
Hence, for ${\Omega = \Omega_{+1}}$,
\begin{align*}
{\bm U}(X) = 
\begin{cases}
C ( -i e^{i \pi/4} {\bm U}^{(+)}_+ e^{\lambda_{-, +} X} + {\bm U}^{(+)}_- e^{\lambda_{-, -} X}) , & X > 0 , \\[2ex]
C (1 + \sqrt{2}) (-i {\bm U}^{(-)}_+ e^{\lambda_{+, +} X} + e^{i \pi/4} {\bm U}^{(-)}_- e^{\lambda_{+, -} X}) , & X < 0 ,
\end{cases}
\end{align*}
where ${C \smallin \C \setminus \{ 0 \}}$ is a normalization constant, and ${\bm U}^{(+)}_\pm$, ${\bm U}^{(-)}_\pm$ are of the form
\begin{align*}
{\bm U}^{(+)}_+ & \equiv [e^{i \pi/4} a, \,   e^{i \pi/4} b, \, c, \, 1]^\mathsf{T} , \\
{\bm U}^{(+)}_- & \equiv [i e^{i \pi/4} a, \, -i e^{i \pi/4} b, \, -c, \, 1]^\mathsf{T} , \\
{\bm U}^{(-)}_+ & \equiv [-i e^{i \pi/4} \tilde{a}, \, i e^{i \pi/4} b, \, -\tilde{c}, \, 1]^\mathsf{T} , \\
{\bm U}^{(-)}_- & \equiv [-e^{i \pi/4} \tilde{a}, \, -e^{i \pi/4} b, \, \tilde{c}, \, 1]^\mathsf{T} ,
\end{align*}
with $a$, $b$, $c$, $\tilde{a}$, $\tilde{c}$, $ \smallin \R$ 
real, satisfying ${a = - (1 - \sqrt{2})^2 \tilde{a}}$,  $b = -\sqrt[4]{2} = -(1+\sqrt{2}) \tilde a = (\sqrt{2}-1) a$,  $c = 1+\sqrt{2}$, $\tilde c = 1 - \sqrt{2}$. 

The first two components of ${\bm U}(X)$ are
\begin{align*}
\begin{bmatrix}
u_1(X) \\
u_2(X)
\end{bmatrix} \!
& =
\begin{cases}
C \biggl( -i e^{i \pi/4}
\! \begin{bmatrix}
e^{i \pi/4} a \\
e^{i \pi/4} b
\end{bmatrix} \!
e^{\lambda_{-, +} X} +
\! \begin{bmatrix}
i e^{i \pi/4} a \\
-i e^{i \pi/4} b
\end{bmatrix} \!
e^{\lambda_{-, -} X} \biggr) , & X > 0 , \\[3ex]
C (1 + \sqrt{2}) \biggl( -i
\! \begin{bmatrix}
-i e^{i \pi/4} \tilde{a} \\
i e^{i \pi/4} b,
\end{bmatrix} \!
e^{\lambda_{+, +} X} + e^{i \pi/4} 
\! \begin{bmatrix}
-e^{i \pi/4} \tilde{a} \\
-e^{i \pi/4} b
\end{bmatrix} \!
e^{\lambda_{+, -} X} \biggr) , & X < 0 ,
\end{cases} \\
& = \begin{cases}
C \biggl(
\! \begin{bmatrix}
a \\
b
\end{bmatrix} \!
e^{\lambda_{-, +} X} + i e^{i \pi/4}
\! \begin{bmatrix}
a \\
-b
\end{bmatrix} \!
e^{\lambda_{-, -} X} \biggr) , & X > 0 , \\[3ex]
-C (1 + \sqrt{2}) \biggl( e^{i \pi/4}
\! \begin{bmatrix}
\tilde{a} \\
-b
\end{bmatrix} \!
e^{\lambda_{+, +} X} - i
\! \begin{bmatrix}
\tilde{a} \\
b
\end{bmatrix} \!
e^{\lambda_{+, -} X} \biggr) , & X < 0 .
\end{cases}
\end{align*}
Define ${\omega \equiv \sqrt[4]{2}/2}$ such that 
\begin{equation*}
\lambda_{+, +} = (1 + i) \omega, \ \ \lambda_{+, -} = (1 - i) \omega, \ \ \lambda_{-, +} = (-1 + i) \omega, \ \ \lambda_{-, -} = (-1 - i) \omega .
\end{equation*}
Then,
\begin{align*}
\begin{bmatrix}
u_1(X) \\
u_2(X)
\end{bmatrix} \!
& = \begin{cases}
C \biggl(
\! \begin{bmatrix}
a \\
b
\end{bmatrix} \!
e^{-\omega X} e^{i \omega X} + i e^{i \pi/4}
\! \begin{bmatrix}
a \\
-b
\end{bmatrix} \!
e^{-\omega X} e^{-i \omega X} \biggr) , & X > 0 , \\[3ex]
-C (1 + \sqrt{2}) \biggl( e^{i \pi/4}
\! \begin{bmatrix}
\tilde{a} \\
-b
\end{bmatrix} \!
e^{\omega X} e^{i \omega X} - i
\! \begin{bmatrix}
\tilde{a} \\
b
\end{bmatrix} \!
e^{\omega X} e^{-i \omega X} \biggr) , & X < 0 ,
\end{cases} \\
& = \begin{cases}
C e^{-\omega X}
\! \begin{bmatrix}
a (e^{i \omega X} + i e^{i \pi/4} e^{-i \omega X}) \\
b (e^{i \omega X} - i e^{i \pi/4} e^{-i \omega X})
\end{bmatrix} \!
, & X > 0 , \\[3ex]
-C (1 + \sqrt{2}) e^{\omega X} 
\! \begin{bmatrix}
\tilde{a} (e^{i \pi/4} e^{i \omega X} - i e^{-i \omega X}) \\
-b (e^{i \pi/4} e^{i \omega X} + i e^{-i \omega X})
\end{bmatrix} \!
, & X < 0 ,
\end{cases} \\
& = \begin{cases}
C e^{-\omega X}
\! \begin{bmatrix}
a (e^{i \omega X} + e^{i 3\pi/4} e^{-i \omega X}) \\
b (e^{i \omega X} - e^{i 3\pi/4} e^{-i \omega X})
\end{bmatrix} \!
, & X > 0 , \\[3ex]
-C (1 + \sqrt{2}) e^{\omega X} 
\! \begin{bmatrix}
\tilde{a} (e^{i \pi/4} e^{i \omega X} - i e^{-i \omega X}) \\
-b (e^{i \pi/4} e^{i \omega X} + i e^{-i \omega X})
\end{bmatrix} \!
, & X < 0 ,
\end{cases} \\
& = \begin{cases}
C e^{-\omega X}
\! \begin{bmatrix}
a e^{i 3\pi/8} (e^{-i 3\pi/8} e^{i \omega X} + e^{i 3\pi/8} e^{-i \omega X}) \\
b e^{i 3\pi/8} (e^{-i 3\pi/8} e^{i \omega X} - e^{i 3\pi/8} e^{-i \omega X})
\end{bmatrix} \!
, & X > 0 , \\[3ex]
-C (1 + \sqrt{2}) e^{\omega X} 
\! \begin{bmatrix}
\tilde{a} e^{i 3\pi/8} (e^{-i \pi/8} e^{i \omega X} - e^{i \pi/8} e^{-i \omega X}) \\
-b e^{i 3\pi/8} (e^{-i \pi/8} e^{i \omega X} - e^{i \pi/8} e^{-i \omega X})
\end{bmatrix} \!
, & X < 0 .
\end{cases} 
\end{align*}
Hence,
\begin{align*}
\begin{bmatrix}
u_1(X) \\
u_2(X)
\end{bmatrix} \!
& = \begin{cases}
- B e^{-\omega X}
\! \begin{bmatrix}
i a \cos(\omega X - 3\pi/8) \\
-b \sin(\omega X - 3\pi/8)
\end{bmatrix} \!
, & X > 0 , \\[3ex]
B (1 + \sqrt{2}) e^{\omega X}
\! \begin{bmatrix}
i \tilde{a} \cos(\omega X - \pi/8) \\
b \sin(\omega X - \pi/8)
\end{bmatrix} \!
, & X < 0 ,
\end{cases}
\end{align*}
where ${B \equiv 2i C e^{i 3\pi/8}}$ is a new normalization constant. Using ${a = - (1 - \sqrt{2})^2 \tilde{a}}$, ${b =  (1 + \sqrt{2}) \tilde{a}}$, and defining ${A \equiv B \tilde{a}}$ completes the proof. \qed

\begin{proposition}
\label{prop:bound-sym}
The bound state ${\bm u}_{+1}(X)$ is of the form ${[i v(X), \, v(-X)]^\mathsf{T}}$, where $v(X)$ is real-valued.
\end{proposition}

\begin{remark}
{\rm (The ${\alpha_0 \neq 1}$ case.)}
Consider the operator $\mathfrak{S}(0)$, defined in \eqref{eq:sql-edge-eff_2}, with the domain wall function ${\chi(X) = {\rm sgn}(X)}$. Assume
\begin{equation}
-(1 - \alpha_0)^2 + \alpha_2^2 > 0 .
\end{equation}
Then, the essential spectrum of $\mathfrak{S}(0)$ is 
\begin{equation}
{\rm spec}_{\rm ess}(\mathfrak{S}(0)) = 
\begin{cases}
\begin{aligned}
& \R \setminus \Bigl( -|\vartheta|, \, |\vartheta| \dfrac{(1 - \alpha_0)|1 - \alpha_0| + \alpha_2^2}{|\alpha_2| \sqrt{-(1 - \alpha_0)^2 + \alpha_2^2}} \Bigr) , & -|\alpha_2| \leq 1 - \alpha_0 \leq 0 , \\
& \R \setminus \Bigl( |\vartheta| \dfrac{(1 - \alpha_0)|1 - \alpha_0| - \alpha_2^2}{|\alpha_2| \sqrt{-(1 - \alpha_0)^2 + \alpha_2^2}}, \, |\vartheta| \Bigr) , & 0 \leq 1 - \alpha_0 \leq |\alpha_2| .
\end{aligned}
\end{cases}
\end{equation}
Further, $\mathfrak{S}(0)$ contains two simple eigenvalues within the spectral gap. These are $|\vartheta|$ multiplied by two of the distinct roots of:
\begin{align*}
0 & = \frac{(1 - \alpha_0) ( -(1 -\alpha_0)^2 + \alpha_2^2) \Omega}{(-(1 - \alpha_0) + \alpha_2)^2 (1 - \Omega^2)} \\
& \hspace{.25cm} - \frac{( 2 \alpha_2^2 \Omega^2 - (-(1 - \alpha_0)^2 + \alpha_2^2)) \sqrt{(-(1 - \alpha_0)^2 + \alpha_2^2) (1 - \Omega^2)}}{(-(1 - \alpha_0) + \alpha_2)^2 (-(1 - \alpha_0)^2 + \alpha_2^2)) (1 - \Omega^2)} 
\end{align*}
that reside in the spectral gap.
\end{remark}

\smallskip

\section{Slopes of eigenvalue curves; Proof of Proposition \ref{prop:sql-eig-local}}
\label{apx:sql-eig-local}

\setcounter{equation}{0}
\setcounter{figure}{0}


{\it Proof of Proposition \ref{prop:sql-eig-local}.} We first, using standard perturbation theory, derive expressions for the slopes, then separately prove the two assertions. 

First, consider the bound state eigenpair ${(\Omega_+, \, {\bm \psi}_{+1})}$, ${\Omega_+ > 0}$ of $\mathfrak{S}(0)$. For $|\kappa|$ small, we seek nearby eigenpairs ${(\Omega_+(\kappa), \psi_+(\kappa))}$ of $\mathfrak{S}(\kappa)$ via the formal ansatz:
\begin{align}
\Omega_+(\kappa) & = \Omega_+ + \kappa \Omega^{(1)}_{+1} + O(\kappa^2) , \\
\psi_+(X; \kappa) & = \psi_+^{(0)}(X) + \kappa \psi_+^{(1)}(X) + O(\kappa^2) , \ \ \text{as} \ \ \kappa \to 0 .
\end{align}
Substituting this into the eigenvalue equation for $\mathfrak{S}(\kappa)$ yields a hierarchy of equations. The first two are:
\begin{align}
\kappa^0: \quad & (\mathfrak{S}(0) - \Omega_+) \psi^{(0)}_{+1} = 0 , \\
\kappa^1: \quad & (\mathfrak{S}(0) - \Omega_+) \psi^{(1)}_{+1} = (\Omega^{(1)}_{+1} - 2 \alpha_1 P_X \sigma_1) \psi_+^{(0)} .
\end{align}
The solution to the first equation is ${\psi^{(0)}_{+1} = \psi_+}$. The second equation is solvable if and only if
\begin{equation}
\langle \psi^{(0)}_{+1}, \, (\Omega^{(1)}_{+1} - 2 \alpha_1 P_X \sigma_1) \psi^{(0)}_{+1} \rangle = 0 ,
\end{equation}
which is possible if and only if
\begin{equation}
\Omega'_+(0) = \Omega^{(1)}_{+1} = 2 \alpha_1 \langle \psi_+, \, P_X \sigma_1 \psi_+ \rangle .
\end{equation}
An analogous calculation holds for the bound state eigenpair ${(\Omega_-, \, {\bm \psi}_{-1})}$, ${\Omega_- = -\Omega_+}$, and we deduce
\begin{equation}
\Omega'_-(0) = 2 \alpha_1 \langle \psi_-, \, P_X \sigma_1 \psi_- \rangle .
\end{equation}
We will now prove the asserted relations.

First, consider the scenario of part (\ref{itm:sql-eig-local_1}). In particular, since ${\chi(-X) = -\chi(X)}$, it follows that
\begin{equation}
\mathfrak{S}(0) \circ \sigma_3 \mathcal{P} = - \sigma_3 \mathcal{P} \circ \mathfrak{S}(0) ,
\end{equation}
which implies ${\psi_-(X) = \sigma_3 \mathcal{P}[\psi_+](X)}$. Hence,
\begin{align}
\Omega'_-(0) & = 2 \alpha_1 \langle \psi_-, \, P_X \sigma_1 \psi_- \rangle \\
& = 2 \alpha_1 \langle \sigma_3 \mathcal{P}[\psi_+], \, P_X \sigma_1 (\sigma_3 \mathcal{P}[\psi_+]) \rangle \nonumber \\
& = 2 \alpha_1 \langle \sigma_3 \mathcal{P}[\psi_+], \, (-1)^2 \sigma_3 \mathcal{P}[P_X \sigma_1 \psi_+] \rangle \nonumber \\
& = 2 \alpha_1 \langle \psi_+, \, P_X \sigma_1 \psi_+ \rangle = \Omega'_+(0) ,
\end{align}
where we used ${\sigma_1 \sigma_3 = - \sigma_3 \sigma_1}$ and ${P_X \circ \mathcal{P} = - \mathcal{P} \circ P_X}$, and that $\sigma_3 \mathcal{P}$ is unitary. This proves part (\ref{itm:sql-eig-local_1}).

Now, consider the scenario of part (\ref{itm:sql-eig-local_2}). Since ${\chi(-X) = \chi(X)}$, it follows that 
\begin{equation}
\mathfrak{S}(0) \circ \mathcal{P} = \mathcal{P} \circ \mathfrak{S}(0) ,
\end{equation}
which implies ${\psi_{\pm 1}(X) = \mathcal{P}[\psi_{\pm 1}](X)}$. Thus
\begin{align}
\Omega'_{\pm 1}(0) & = 2 \alpha_1 \langle \psi_{\pm 1}, \, P_X \sigma_1 \psi_{\pm 1} \rangle \\
& = 2 \alpha_1 \langle \mathcal{P}[\psi_{\pm 1}], \, P_X \sigma_1 (\mathcal{P}[\psi_{\pm 1}]) \rangle \nonumber \\
& = 2 \alpha_1 \langle \sigma_3 \mathcal{P}[\psi_+], \, (-1) \mathcal{P}[P_X \sigma_1 \psi_{\pm 1}] \rangle \nonumber \\
& = - 2 \alpha_1 \langle \psi_{\pm 1}, \, P_X \sigma_1 \psi_{\pm 1} \rangle = - \Omega'_{\pm 1}(0) ,
\end{align}
where we used ${P_X \circ \mathcal{P} = - \mathcal{P} \circ P_X}$, and that $\mathcal{P}$ is unitary. Therefore ${\Omega'_{\pm 1}(0) = 0}$. This proves part (\ref{itm:sql-eig-local_2}). \qed

\smallskip

\section{Other numerics}

\setcounter{equation}{0}
\setcounter{figure}{0}

In this section, we show
some illustrative numerical experiments to demonstrate properties of
$\mathfrak{S}(\kappa)$ from \eqref{eq:sql-edge-eff_2}.   To assist the reader, we recall the operator as
\begin{equation*}
\mathfrak{S}(\kappa) = (1 - \alpha_0) (P_X^2 + \kappa^2) I + 2 \alpha_1 P_X \kappa \sigma_1 + \alpha_2 (-P_X^2 + \kappa^2) \sigma_2 + \vartheta \chi(X) \sigma_3 .
\end{equation*}

\begin{figure}[!t]
\begin{center}
\includegraphics[width=.45\textwidth]{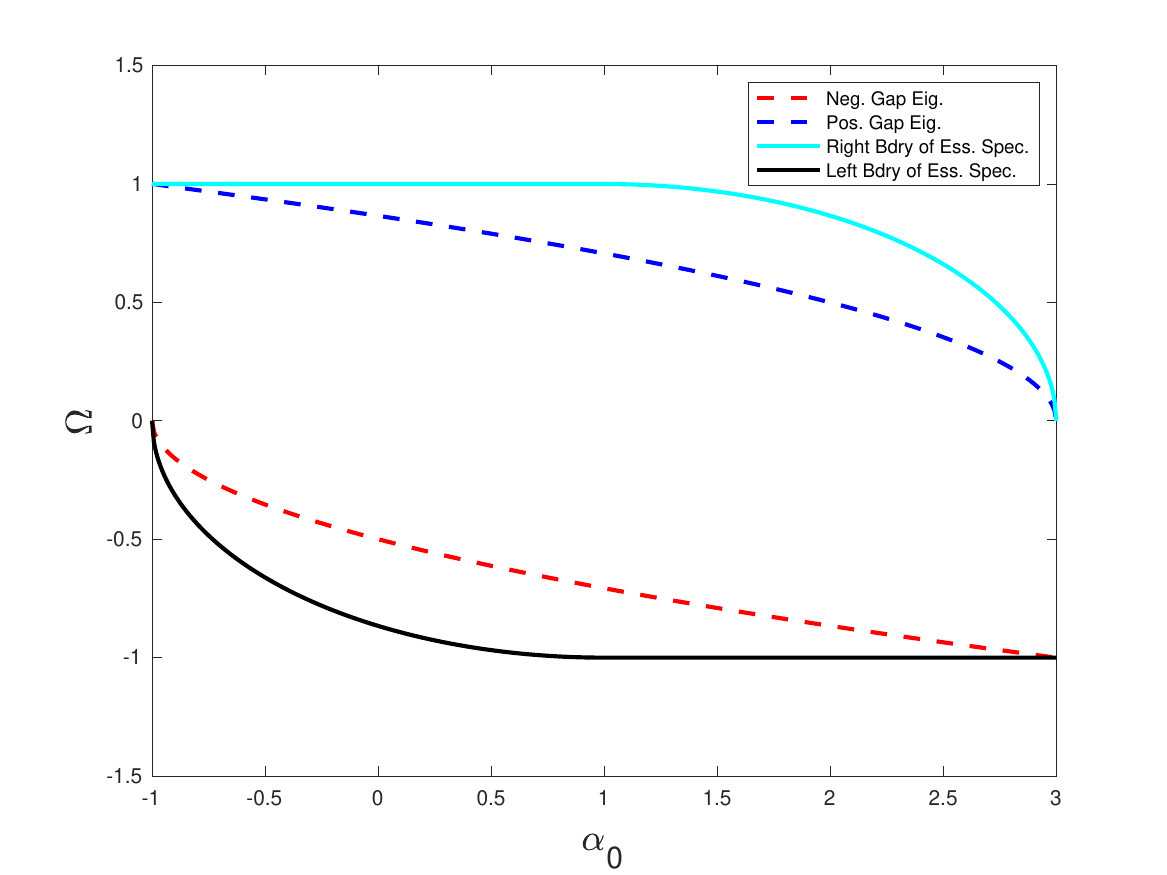}
\caption{Cartoon of the spectrum of $\mathfrak{S}(0)$ from \eqref{eq:sql-edge-eff_2} with $\chi(x) = \sgn(x)$ as it varies with $\alpha_0$ for a fixed $\alpha_2$.  We plot the boundary edges of the band gap and the energies of the corresponding eigenstates  with $\vartheta = 1$, $\alpha_2 =2$ and varying $-1<\alpha_0<3$. The gap eigenvalues branch off of the edges of the spectrum and are monotonically decreasing as a function of $\alpha_0$.}
\label{f:eigenvalues}
\end{center}
\end{figure}

First, in Figure \ref{f:eigenvalues}, we take $\kappa = 0$ and demonstrate how the edges of the essential spectrum and the gap eigenstates od the exactly solvable model when $\chi (x) = \sgn (x)$. We have otherwise take $\vartheta=1$ and $\alpha_2 = 2$, then we vary the values $-1<\alpha_0<3$.  The figure shows the gap eigenstates colliding with the essential spectrum as $\alpha_0$ approaches the values $-1$ and $3$, as well as the consistent existence of $2$ gap eigenvalues for all $\alpha_0$ in this range.

Next, again in the exactly solvable case with $\chi (x) = \sgn (x)$ and $\vartheta=1$, we consider the $\kappa$ variation of \eqref{eq:sql-edge-eff_2} for a number of illustrative choices of parameters $\alpha_0$, $\alpha_1$, and $\alpha_2$.  We have included four cases in Figure \ref{f:71kappa} to demonstrate the possible behaviors of the spectrum.  First in the top left, we take $\alpha_0 = 1.5, \alpha_2=2$, $\alpha_1 = 0$, which gives a particularly interesting scenario showing $2$ traversing gap eigenvalues but no global spectral gap in the bulk.
In the top right we take  $\alpha_0 = .5, \alpha_2=1$, $\alpha_1 = .5$ to demonstrate the symmetric property for $1-\alpha_0$ alternating in sign.  Lastly, the bottom figures demonstrate that a spectral gap and traversing states are stable under fluctuation in $\alpha_0 \sim 1$.  In the bottom left, we take $\alpha_0 = 0.9, \alpha_2=1$, $\alpha_1 = 1$; and $\alpha_0 = 1.1, \alpha_2=1$, $\alpha_1 = 1$ in the bottom right. 

\begin{figure}[!ht]
\begin{center}
\includegraphics[width=.3\textwidth]{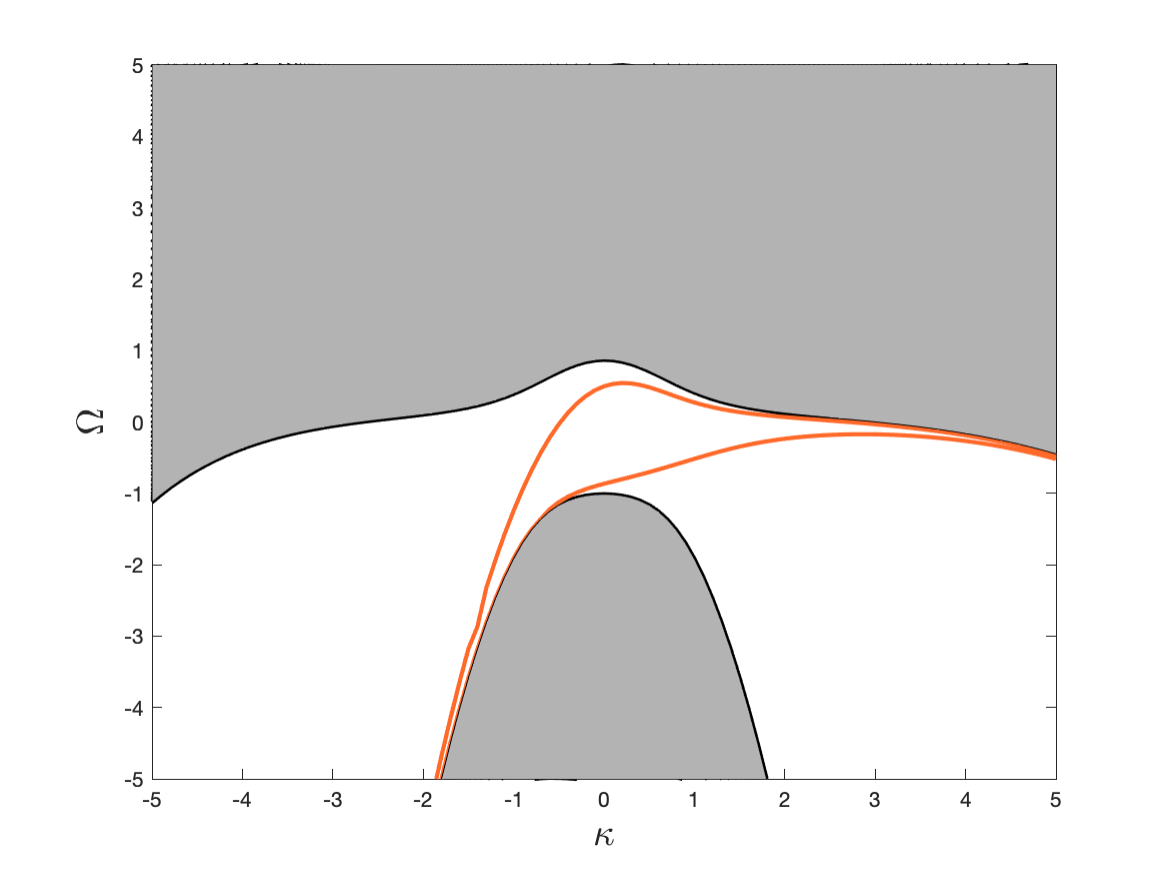} \ \ 
\includegraphics[width=.3\textwidth]{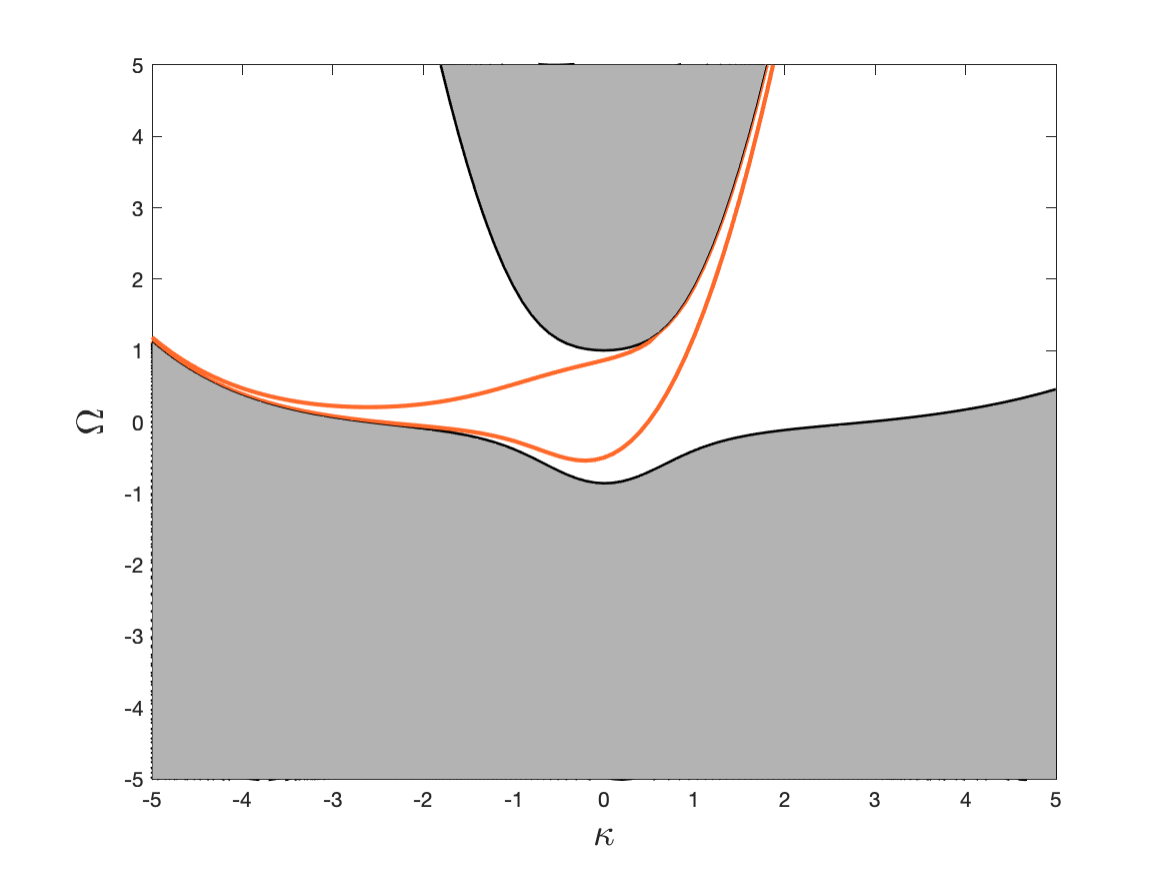} \\
\includegraphics[width=.3\textwidth]    {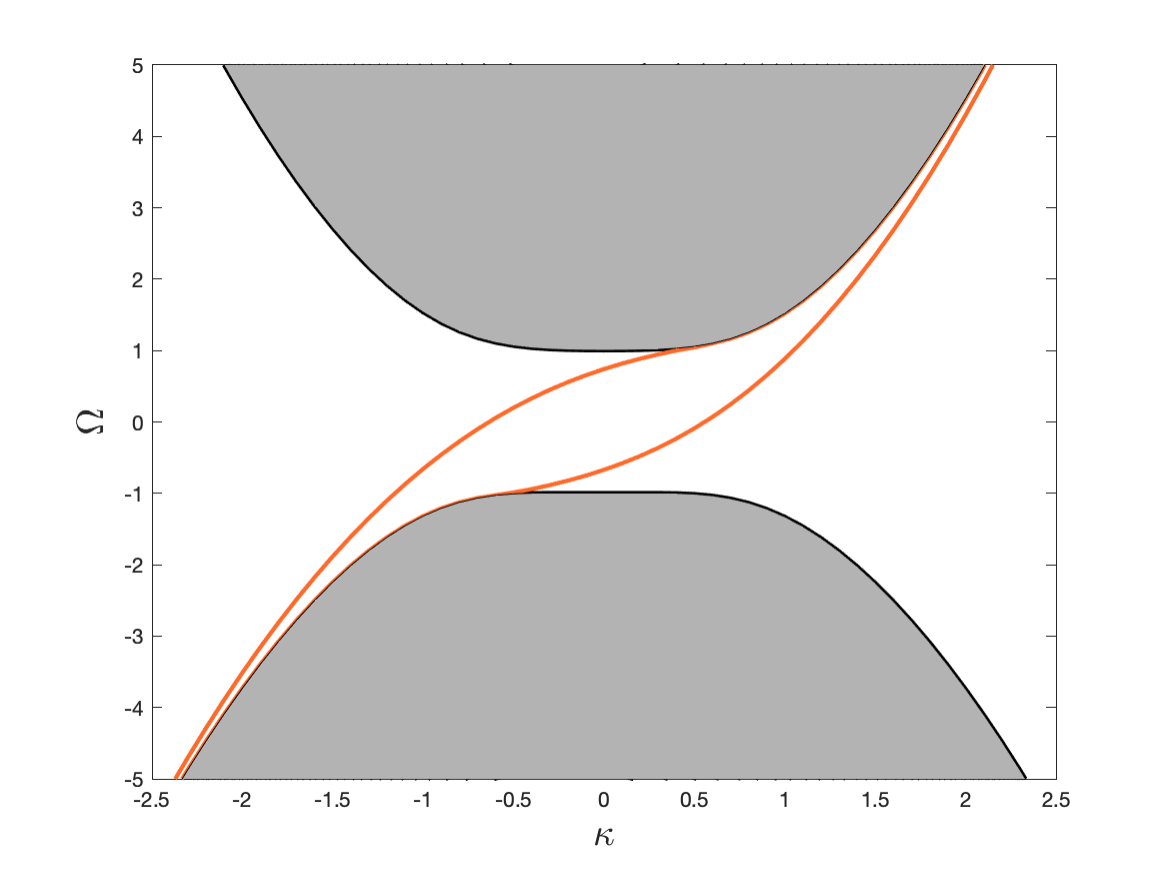} \ \
\includegraphics[width=.3\textwidth]{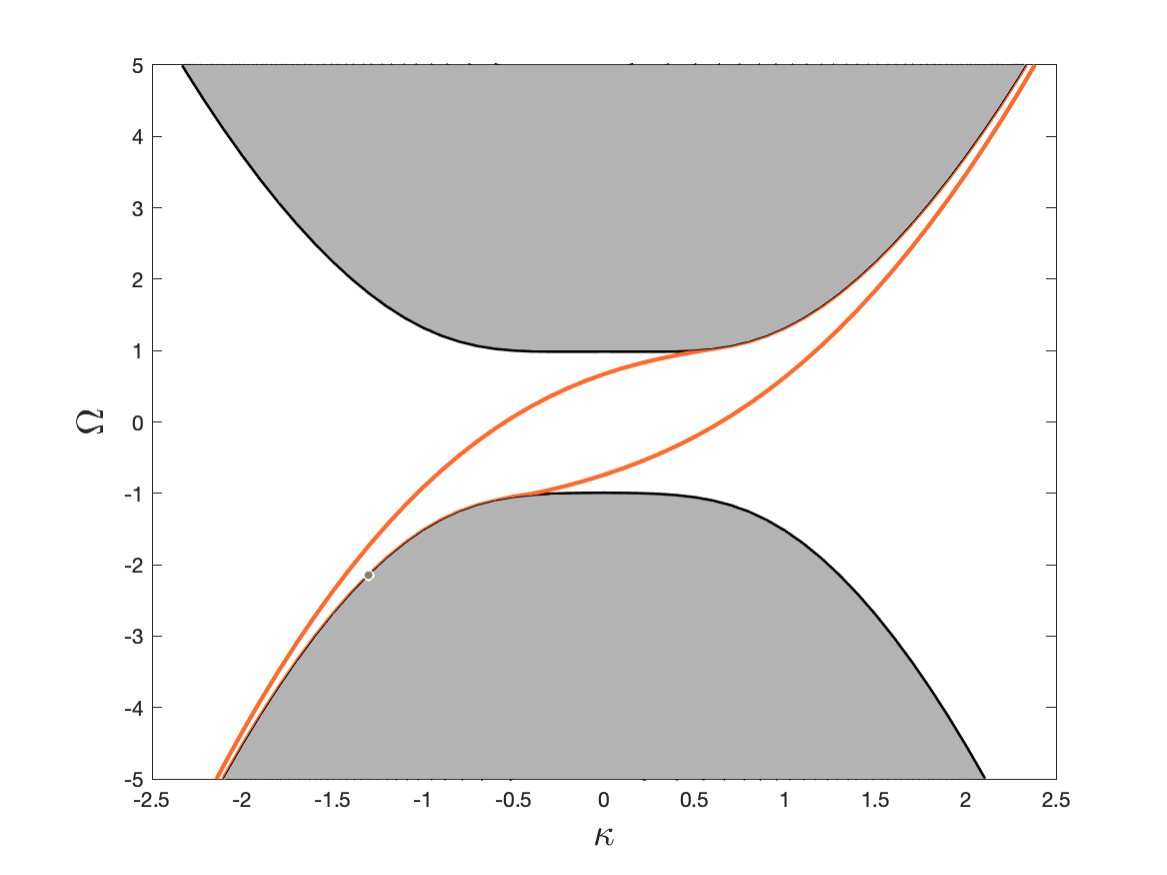}
\caption{Spectrum around the gap of $\mathfrak{S}(\kappa)$ from \eqref{eq:sql-edge-eff_2} $-5 \leq \kappa \leq 5$  with $\chi (x) = \sgn (x)$ and parameter values:  $\alpha_0 = 1.5, \alpha_2=1$, $\alpha_1 = .5$, $\vartheta = 1$ (Top left);  $\alpha_0 = .5, \alpha_2=1$, $\alpha_1 = .5$, $\vartheta = 1$ (Top Right); $\alpha_0 = 0.9, \alpha_2=1$, $\alpha_1 = 1$, $\vartheta = 1$ (Bottom Left) and $\alpha_0 = 1.1, \alpha_2=1$, $\alpha_1 = 1$, $\vartheta = 1$ (Bottom Right).}
\label{f:71kappa}
\end{center}
\end{figure}

\begin{figure}[!h]
\begin{center}
\includegraphics[width=.3\textwidth]{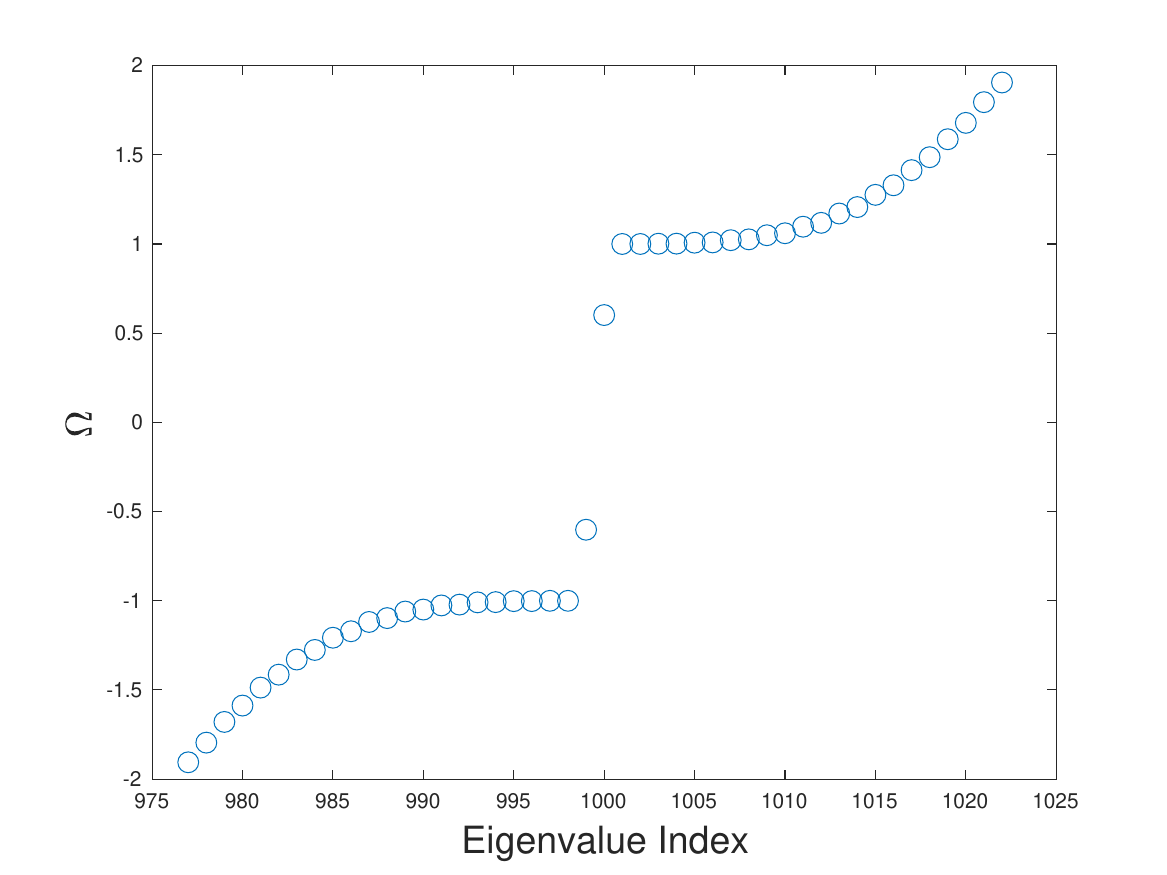} \ \ 
\includegraphics[width=.3\textwidth]{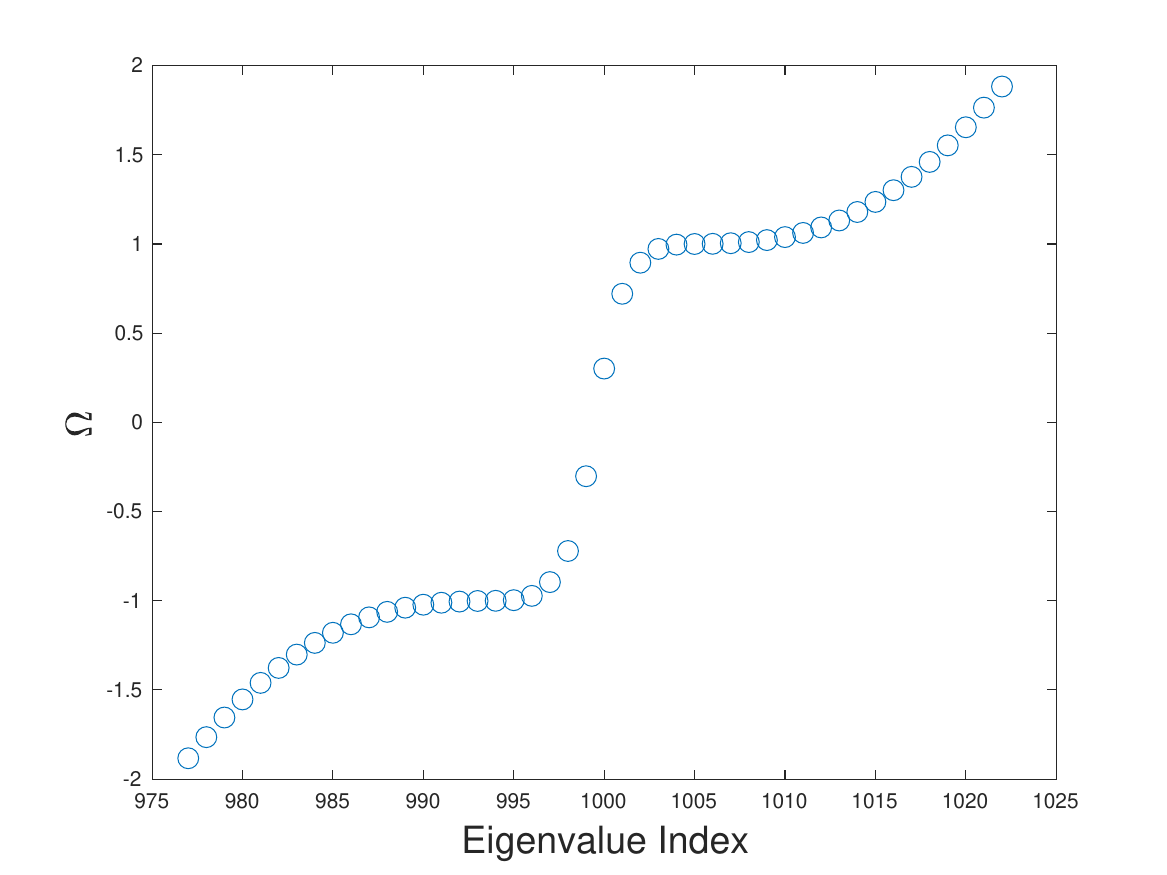} \\
\includegraphics[width=.3\textwidth]{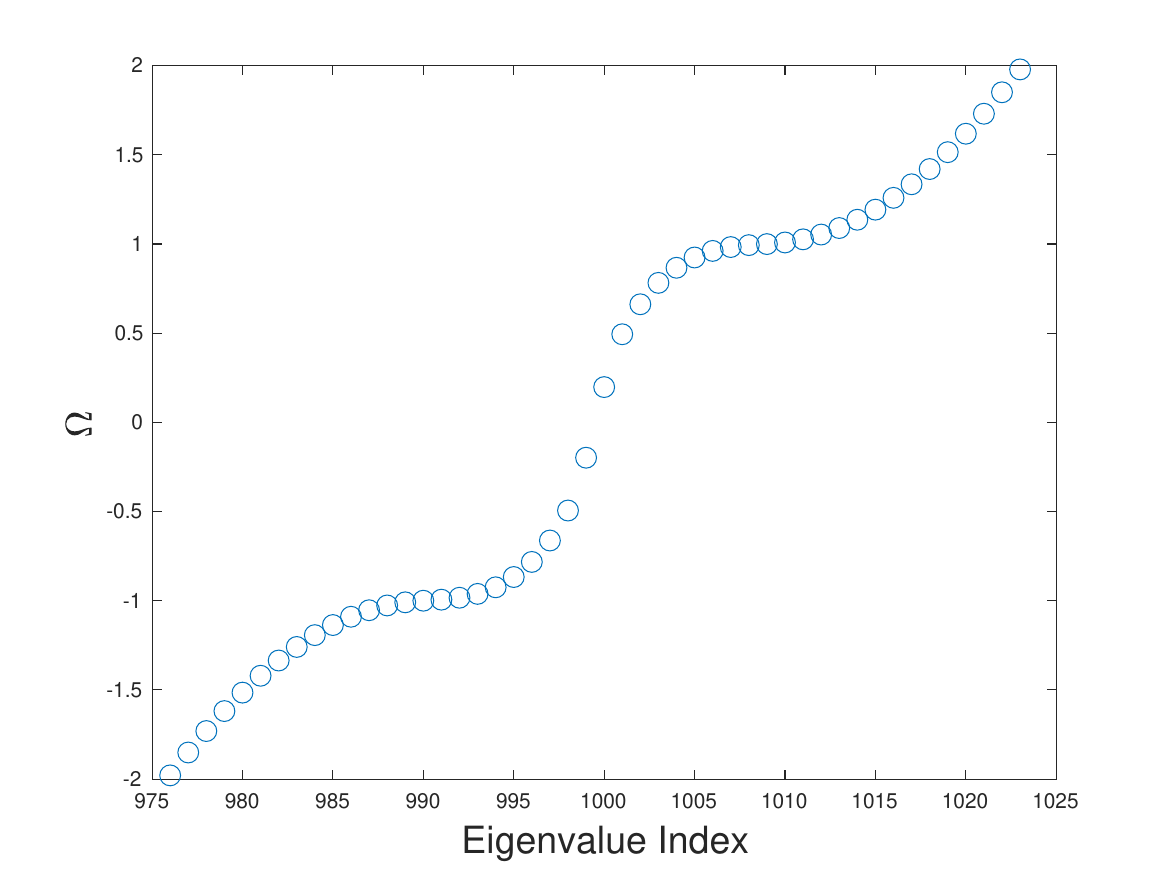} \ \
\includegraphics[width=.3\textwidth]{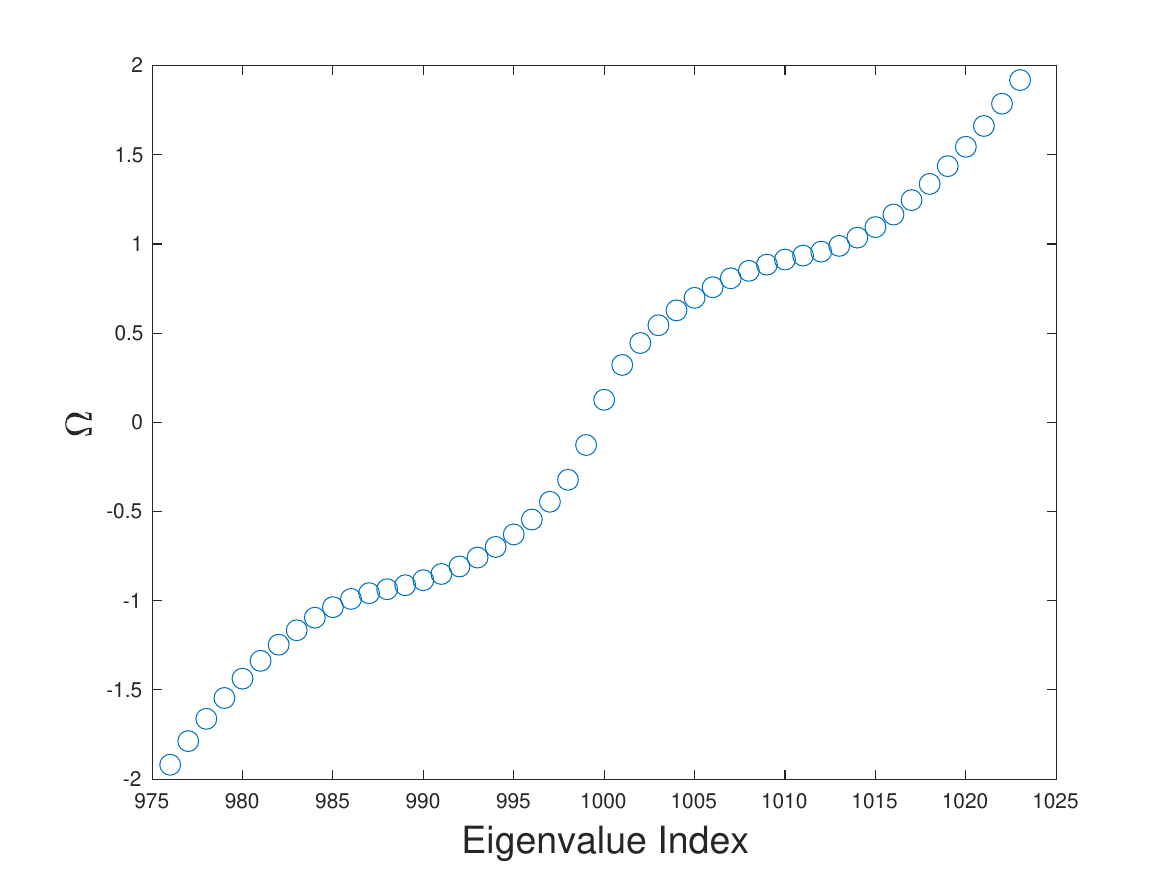}
\caption{Spectrum around the gap of $\mathfrak{S}(0)$ from \eqref{eq:sql-edge-eff_2}  with $\chi (x) = \tanh (x/L)$ and parameter values:  $\alpha_0 = 1, \alpha_2=2$, $\vartheta = 1$ and $L = 1$ (Top left), $L = 5$ (Top Right); $L=10$ (Bottom Left) and $L=20$ (Bottom Right).}
\label{f:tanh2}
\end{center}
\end{figure}

Moving away from the exactly solvable model, we consider $\chi(x) = \tanh (x/L)$, taking increasing values of $L$. The widening of the defect edge allows for multiple curves as they are able to carry more eigenvalues.  See Figure \ref{f:tanh2}, where we fix $\kappa = 0$ and $\alpha_0 = 1, \alpha_2=2$, $\vartheta = 1$, then take $L = 1$ (Top left), $L = 5$ (Top Right); $L=10$ (Bottom Left) and $L=20$ (Bottom Right).   The wider wells are able to carry more bound states, which explains the increasing number of gap states that we observe as $L$ increases.
    
Lastly, we consider effects of multiple defect edges again in the setting $\kappa =0$, which gives multiple curves depending upon the number of defect edges and their spacing, fix other parameters.  See Figures \ref{f:evenodd} and \ref{f:closedws}.  Specifically, we consider an even number of defect edges in $\mathfrak{S}(0)$ from \eqref{eq:sql-edge-eff_2} with 
    \begin{equation}
    \label{eqn:evenchi} 
    \chi(x) = \chi_{L,{\rm even}} (x) := \tanh (x-L) + (1-\tanh (x+L))
    \end{equation}
and an odd number of defect edges in $\mathfrak{S}(0)$ from \eqref{eq:sql-edge-eff_2} with
    \begin{equation}
    \label{eqn:oddchi}
      \chi(x) = \chi_{L,{\rm odd}} (x) := \tanh (x-L) - (1+\tanh(x)) + (1 +\tanh(x+L)).
    \end{equation}
    In these figures, we take
$\alpha_0 = 1, \alpha_2=1$, $\vartheta = 1$; and $L = 2$ for even and $L=4$ for odd.  Each defect edge introduces $2$ gap states, but their respective energies depends somewhat on the edge spacing.  

\begin{figure}[!t]
\begin{center}
\includegraphics[scale = 0.25]{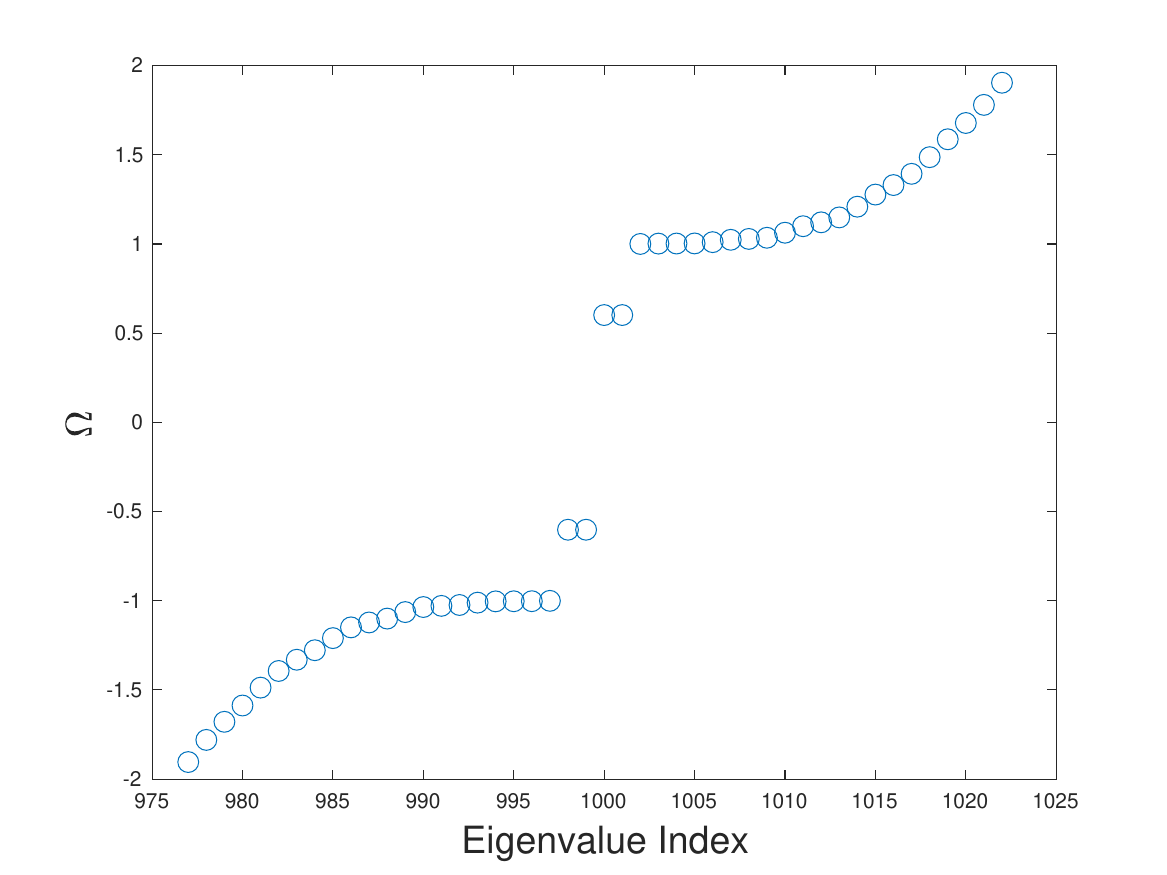} \ \ 
\includegraphics[scale = 0.25]{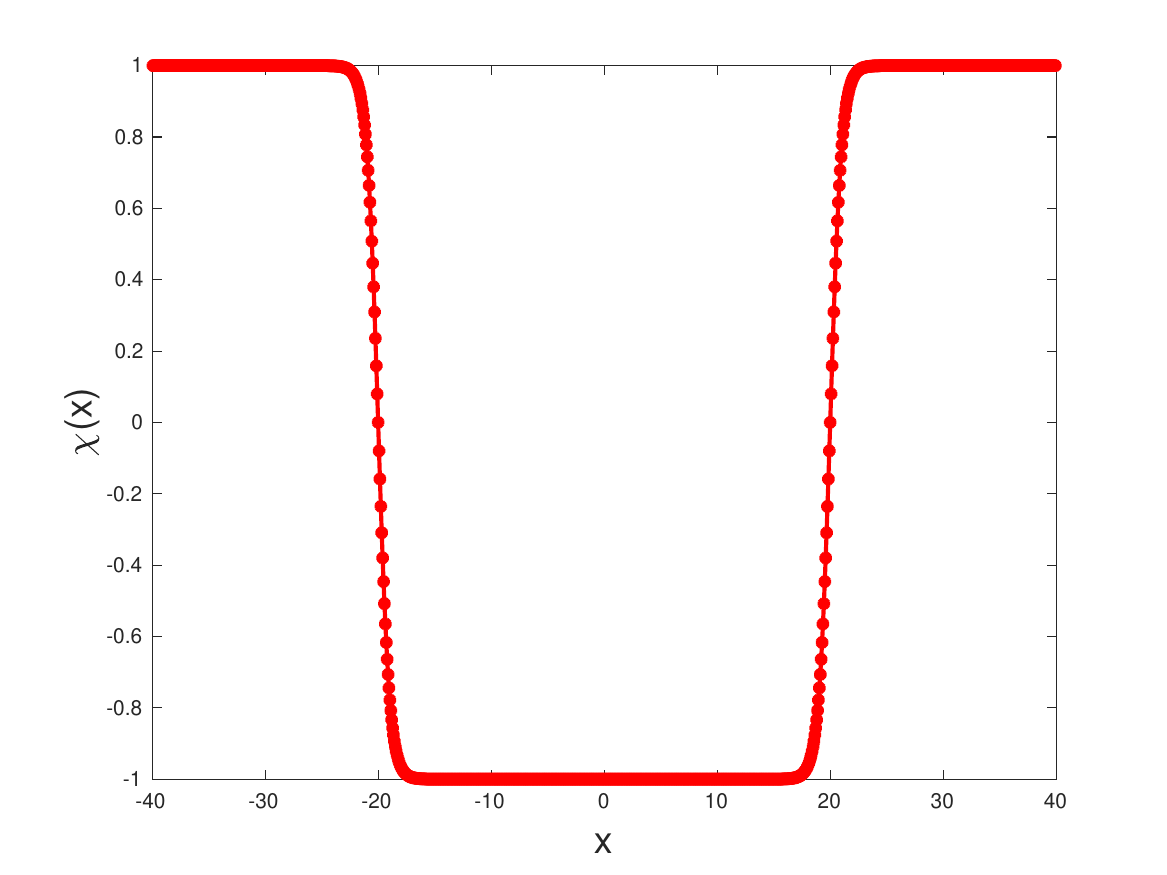} \\
\includegraphics[scale = 0.25]{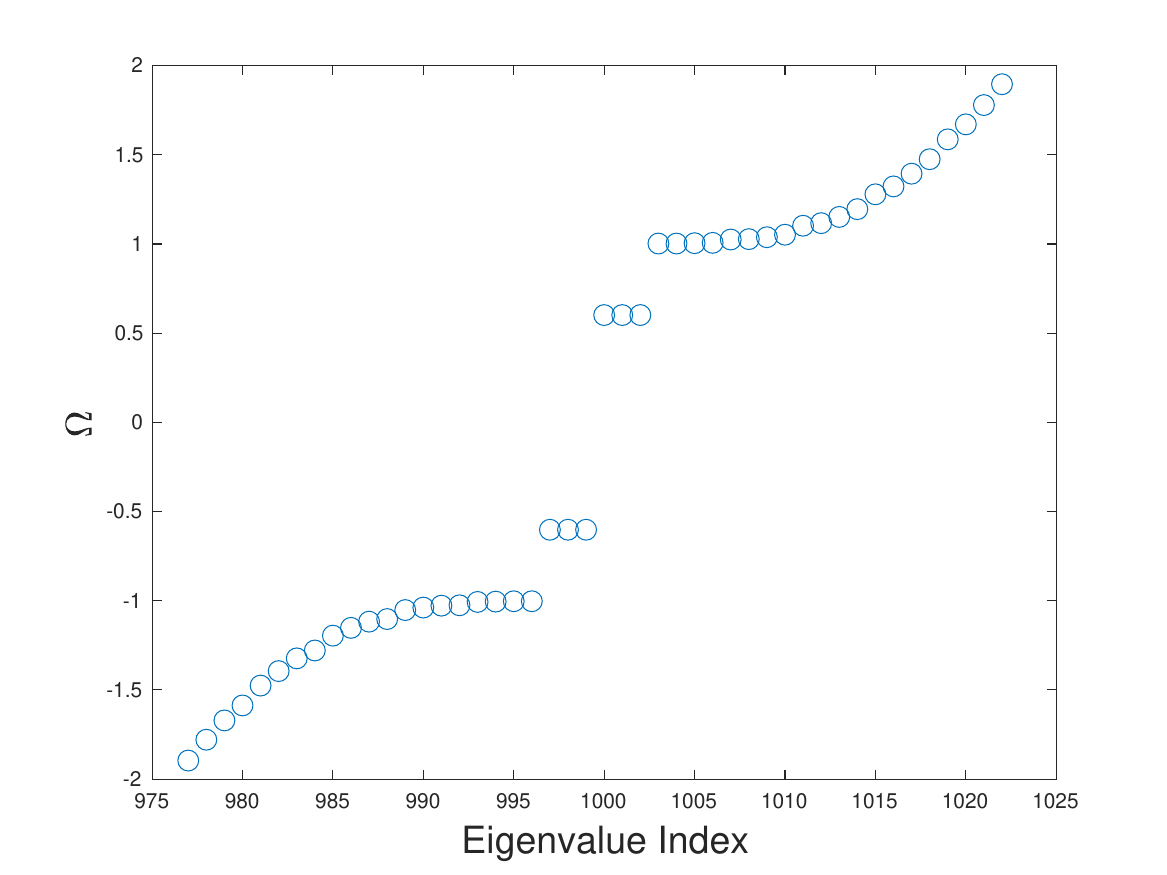} \ \ 
\includegraphics[scale = 0.25]{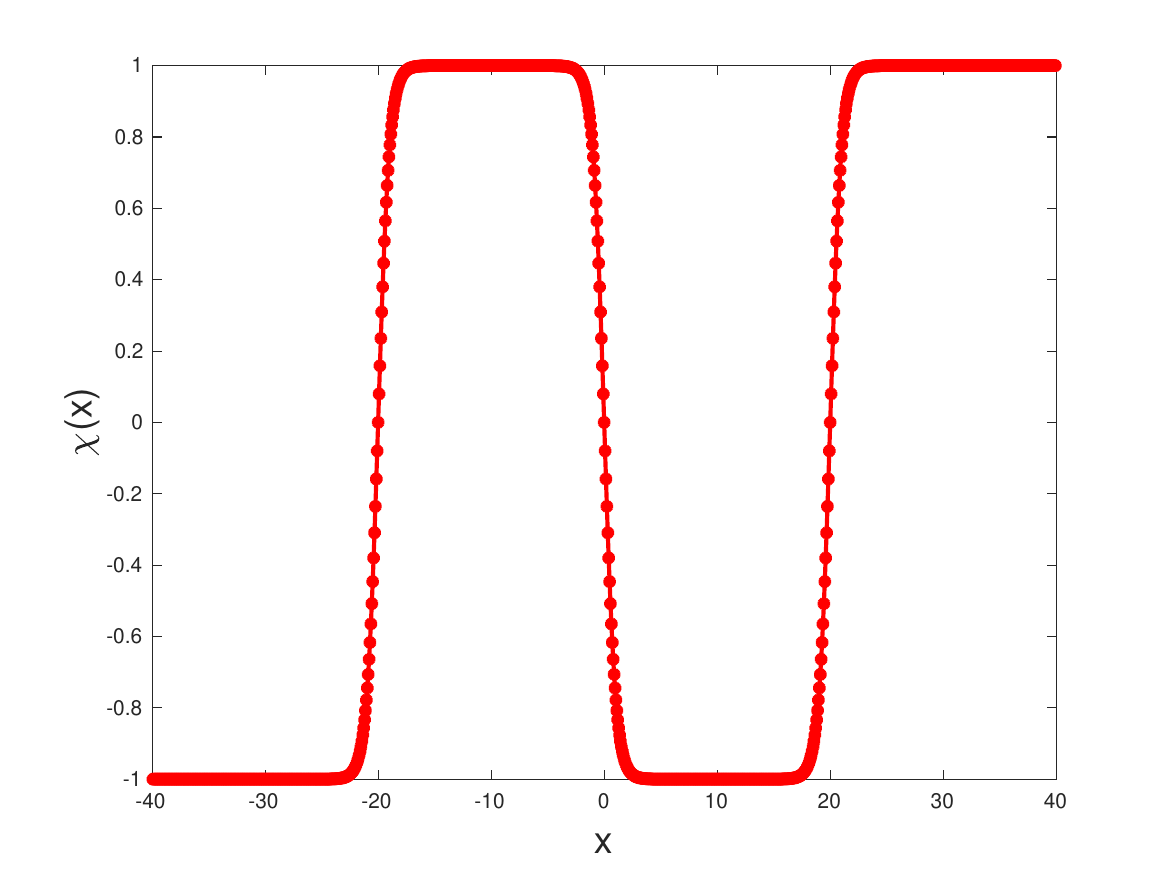}
\caption{Spectrum around the gap of $\mathfrak{S}(0)$ from \eqref{eq:sql-edge-eff_2}.  Here, we explore the impact of multiple defect edges that are well-separated. The grouping of gap eigenvalues demonstrated on the (left) for both the even (top) and the odd (bottom) number of domain walls given by $\chi = \chi_{L,{\rm even}/{\rm odd}}$ on the (right).  Here, we have parameter values:  $\alpha_0 = 1, \alpha_2=2$, $\vartheta = 1$, $L = 20$. }
\label{f:evenodd}
\end{center}
\end{figure}

\begin{figure}[!t]
\begin{center}
\includegraphics[scale = 0.25]{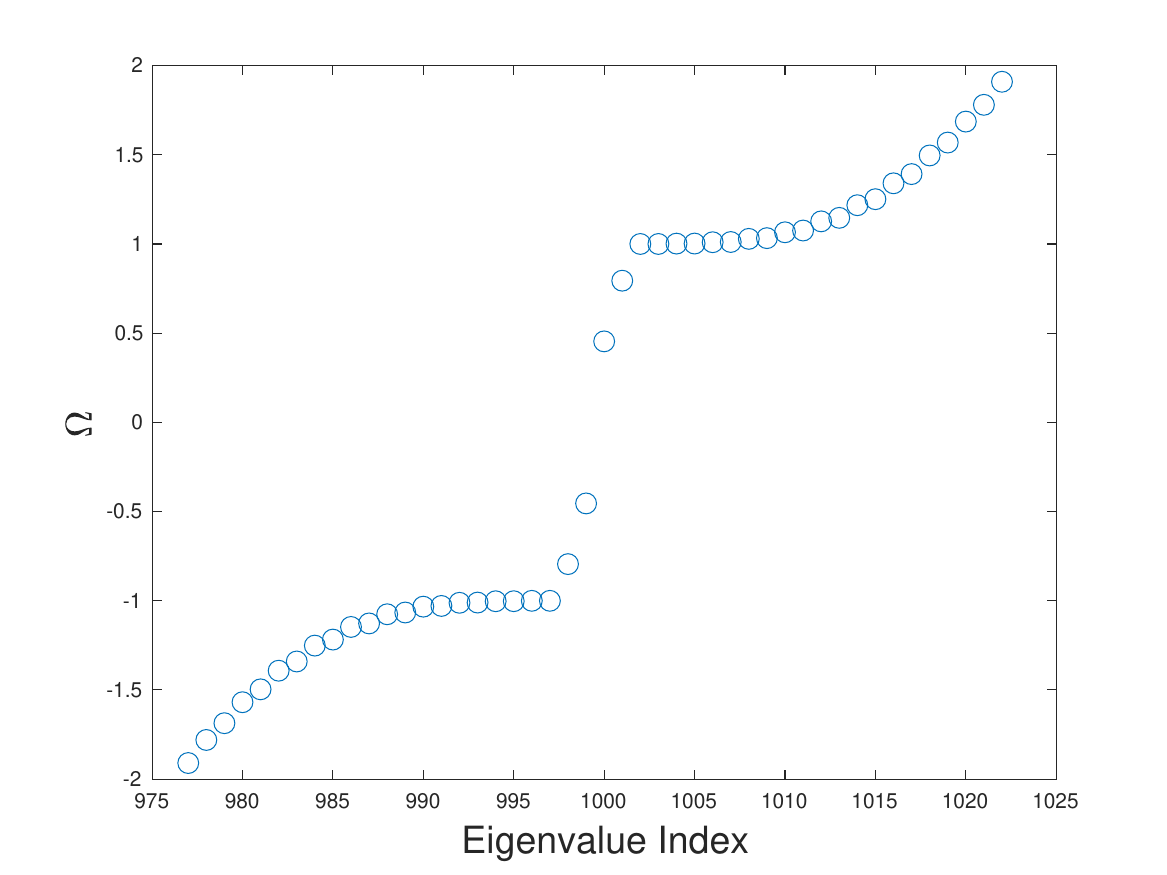} \ \ 
\includegraphics[scale = 0.25]{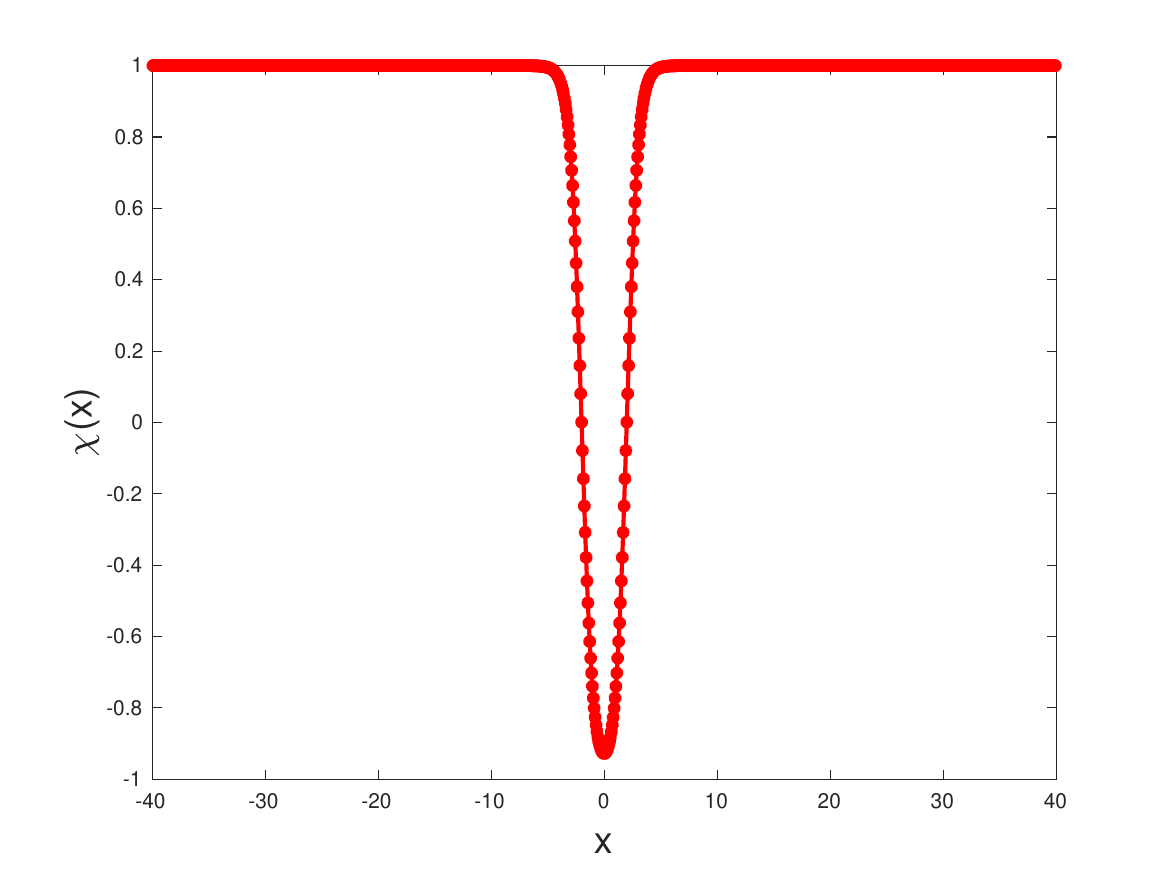} \\  
\includegraphics[scale = 0.25]{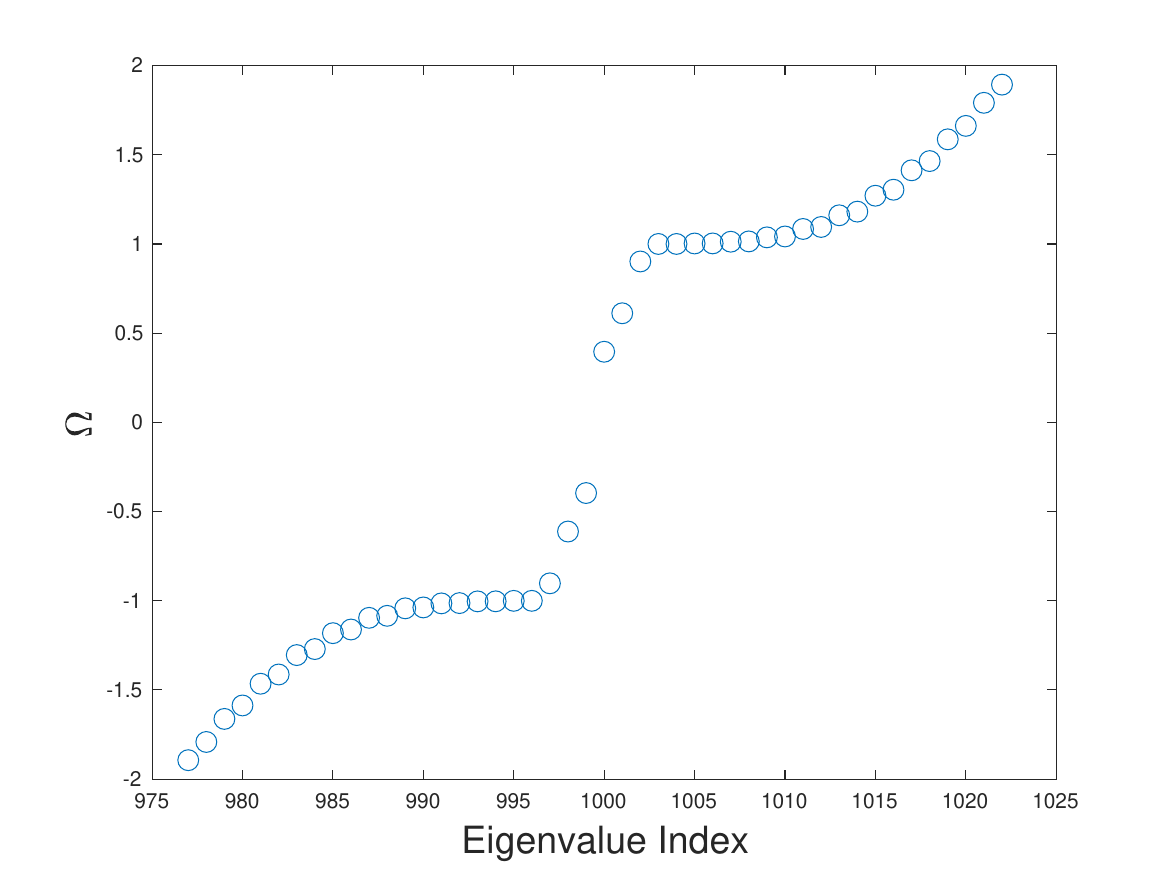} \ \ 
\includegraphics[scale = 0.25]{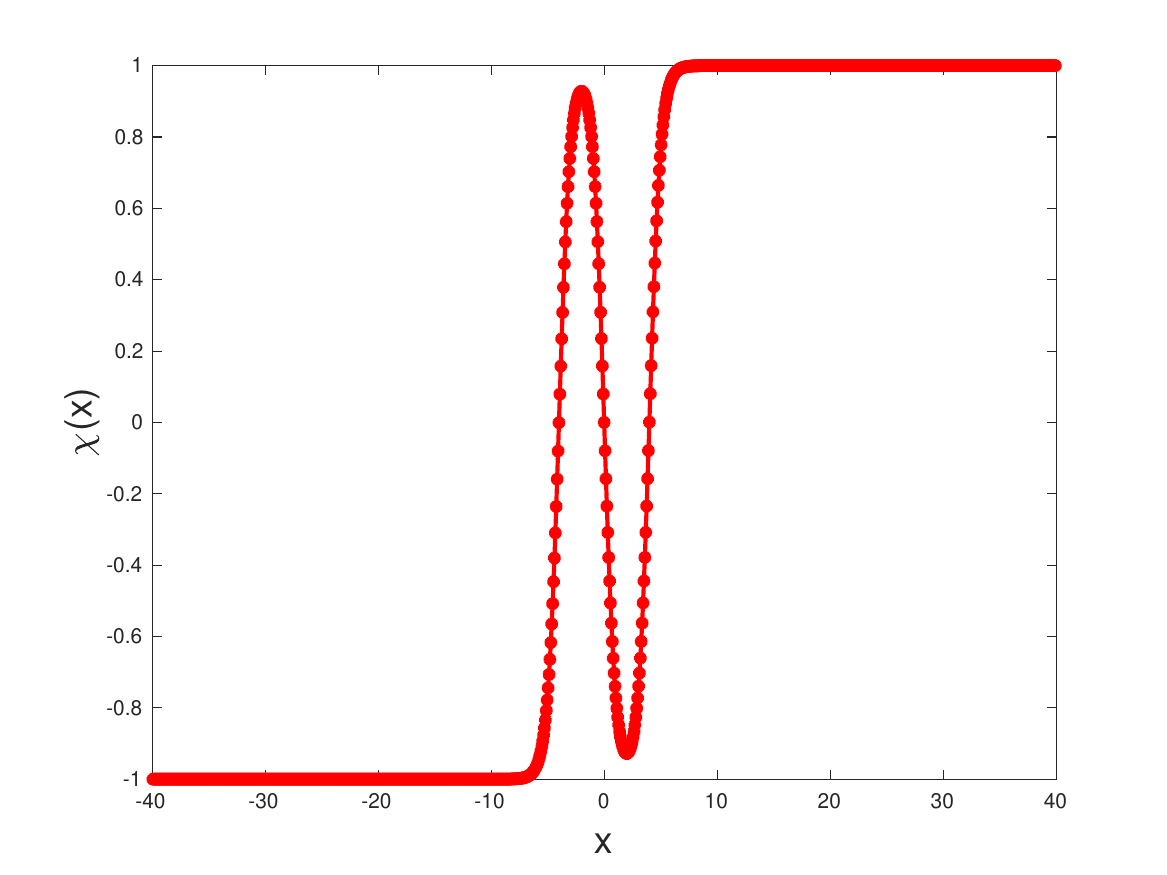}
\caption{ Spectrum around the gap of $\mathfrak{S}(0)$ from \eqref{eq:sql-edge-eff_2}.  Here, we explore the impact of multiple defect edges as the edges approach each other.  The grouping of gap eigenvalues demonstrated on the (left) for both the even (top) and the odd (bottom) number of domain walls given by $\chi = \chi_{L,{\rm even}/{\rm odd}}$ on the (right).  Here, we have parameter values:  $\alpha_0 = 1, \alpha_2=1$, $\vartheta = 1$; $L = 2$ for even and $L=4$ for odd.}
\label{f:closedws}
\end{center}
\end{figure}

\smallskip

\clearpage
\section{Proofs of Propositions in Section \ref{sec:spec-dfm-disc}}
\label{supp:dfm-edge-eff-disc}

\setcounter{equation}{0}
\setcounter{figure}{0}


\subsection{On the eigenvalue curve through zero energy; Proof of Proposition \ref{prop:dfm-eig-local}}
\label{supp:dfm-eig-local}

We prove Proposition \ref{prop:dfm-eig-local}, which gives the local behavior of the eigenvalue curves ${\kappa \mapsto \Omega_\pm(\kappa)}$ of the Dirac operators $\cancel{\mathfrak{D}}^\pm(\kappa)$ \eqref{eq:dfm-edge-eff_2} passing through ${\Omega_\pm(0) = 0}$.

\subsubsection{Proof of Proposition \ref{prop:dfm-eig-local}, part \ref{itm:dfm-eig-local-1}}
\label{supp:dfm-eig-local-1}

\noindent {\it Proof of Proposition \ref{prop:dfm-eig-local}, part \ref{itm:dfm-eig-local-1}.} For simplicity, consider the $\cancel{\mathfrak{D}}^+(\kappa)$ case; the $\cancel{\mathfrak{D}}^-(\kappa)$ case proceeds similarly. First, from \eqref{eq:dfm-edge-eff_2}, observe that
\begin{equation}
\label{eq:dfm-edge-eff_3}
\cancel{\mathfrak{D}}^+(\kappa) = \cancel{\mathfrak{D}}^+(0) + (b \cdot \sigma) \kappa ,
\end{equation}
where $b$ and $\sigma$ are defined in \eqref{eqn:abdefs} and \eqref{eq:def-pauli} respectively. Now consider the bound state eigenpair ${(0, \psi^\star_+)}$ of $\cancel{\mathfrak{D}}^+(0)$. For $|\kappa|$ small, we seek nearby eigenpairs of $\cancel{\mathfrak{D}}^+(\kappa)$ via the formal ansatz:
\begin{align}
\Omega(\kappa) & = \kappa \Omega^{(1)} + \kappa^2 \Omega^{(2)}  + O(\kappa^3) , \\
\psi(X; \kappa) & = \psi^{(0)}(X) + \kappa \psi^{(1)}(X) + \kappa^2 \psi^{(2)}(X) + O(\kappa^3) \ \ \text{as} \ \ |\kappa| \to 0 .
\end{align}
Substituting this into the eigenvalue equation for $\cancel{\mathfrak{D}}^+(\kappa)$ and grouping terms by orders of $\kappa$ yields a hierarchy of equations. The first three are:
\begin{align}
\label{eq:zero-pert-0}
\kappa^0 : \quad & \cancel{\mathfrak{D}}^+(0) \psi^{(0)} = 0 , \\
\label{eq:zero-pert-1}
\kappa^1 : \quad & \cancel{\mathfrak{D}}^+(0) \psi^{(1)} = (\Omega^{(1)} - b \cdot \sigma) \psi^{(0)} , \\
\label{eq:zero-pert-2}
\kappa^2 : \quad & \cancel{\mathfrak{D}}^+(0) \psi^{(2)} = (\Omega^{(1)} - b \cdot \sigma) \psi^{(1)} + \Omega^{(2)} \psi^{(0)} .
\end{align}
Below, we solve the equations order by order.

The solution to the first equation is ${\psi^{(0)} = c^{(0)} \psi^\star_+}$, where ${c^{(0)} \smallin \C}$. 

The second equation is solvable if and only if
\begin{equation}
\langle \psi^\star_+, \, (\Omega^{(1)} - b \cdot \sigma) \psi^{(0)} \rangle = 0 ,
\end{equation}
which is possible if and only if
\begin{equation}
\Omega^{(1)} = \Omega^{(1)}_+ \equiv \langle \psi^\star_+, \, (b \cdot \sigma) \psi^\star_+ \rangle .
\end{equation}
The general solution is then
\begin{equation}
\psi^{(1)} = c^{(0)} \cancel{\mathfrak{D}}^+(0)^{-1} (\Omega^{(1)} - b \cdot \sigma) \psi^\star_+ + c^{(1)} \psi^\star_+ ,
\end{equation}
where ${c^{(1)} \smallin \C}$.

The third equation is solvable if and only if
\begin{equation}
\langle \psi^\star_+, \, (\Omega^{(1)} - b \cdot \sigma) \psi^{(1)} + \Omega^{(2)} \psi^{(0)} \rangle = 0 ,
\end{equation}
which is possible if and only if
\begin{equation}
\Omega^{(2)} = \Omega^{(2)}_+ \equiv - \langle (\Omega^{(1)}_+ - b \cdot \sigma) \psi^\star_+, \, \cancel{\mathfrak{D}}^+(0)^{-1} (\Omega^{(1)}_+ - b \cdot \sigma) \psi^\star_+ \rangle .
\end{equation}
The solution $\psi^{(2)}$ can then be determined by inverting $\cancel{\mathfrak{D}}^+(0)$. \qed

\subsubsection{Odd domain walls; Proof of Proposition \ref{prop:dfm-eig-local}, part \ref{itm:dfm-eig-local-2}}
\label{supp:dfm-eig-local-2}

We first state and prove a global symmetry:

\begin{proposition}
\label{prop:dfm-odd-sym}
Suppose ${(\Omega, \psi(X))}$ is an eigenpair of $\cancel{\mathfrak{D}}^\pm(\kappa)$ \eqref{eq:dfm-edge-eff_2}, and assume ${\chi(-X) = -\chi(X)}$. Then ${(-\Omega, \mathcal{P}[\psi](X) = \psi(-X))}$ is an eigenpair of $\cancel{\mathfrak{D}}^\pm(-\kappa)$.
\end{proposition}

\begin{proof}
We prove the result for $\cancel{\mathfrak{D}}^+(\kappa)$; the argument for $\cancel{\mathfrak{D}}^-(\kappa)$ is identical. Since ${P_X \circ \mathcal{P} =}$ ${- \mathcal{P} \circ P_X}$,
\begin{align}
\mathcal{P} \circ \cancel{\mathfrak{D}}^+(\kappa) & = \mathcal{P} \circ ((a \cdot \sigma) P_X + (b \cdot \sigma) \kappa + c \chi(X) \sigma_3) \\
& = (-(a \cdot \sigma) P_X + (b \cdot \sigma) \kappa + c \chi(-X) \sigma_3) \circ \mathcal{P} \nonumber \\
& = -((a \cdot \sigma) P_X - (b \cdot \sigma) \kappa + c \chi(X) \sigma_3)) \circ \mathcal{P} = - \cancel{\mathfrak{D}}^+(-\kappa) \circ \mathcal{P} . \nonumber
\end{align}
Hence, if ${(\Omega, \psi)}$ is an eigenpair of $\cancel{\mathfrak{D}}^+(\kappa)$, we have
\begin{equation}
\cancel{\mathfrak{D}}^+(-\kappa) \mathcal{P}[\psi] = - \mathcal{P}[\cancel{\mathfrak{D}}^+(\kappa) \psi] = - \mathcal{P}[\Omega \psi] = - \Omega \mathcal{P}[\psi] .
\end{equation}
That is, ${(-\Omega, \mathcal{P}[\psi])}$ is an eigenpair of $\cancel{\mathfrak{D}}^+(\kappa)$. The proof is complete.
\end{proof}

We now proceed to prove Proposition \ref{prop:dfm-eig-local}, part \ref{itm:dfm-eig-local-2}. \\

\noindent {\it Proof of Proposition \ref{prop:dfm-eig-local}, part \ref{itm:dfm-eig-local-2}.} First, recall that ${\Omega_\pm(0) = 0}$ is a simple eigenvalue of $\cancel{\mathfrak{D}}^\pm(\kappa)$; see Proposition \ref{prop:zero-energy}. By regular perturbation theory, for $|\kappa|$ small, $\Omega_\pm(\kappa)$ is a simple eigenvalue of $\cancel{\mathfrak{D}}^\pm(\kappa)$ and ${\kappa \mapsto \Omega_\pm(\kappa)}$ is analytic; see Proposition \ref{prop:dfm-eig-local}, part \ref{itm:dfm-eig-local-1}. Next, the symmetry of Proposition \ref{prop:dfm-odd-sym} implies, for each $\kappa$, that if $\Omega_\pm(\kappa)$ is an eigenvalue of $\cancel{\mathfrak{D}}(\kappa)$, then $-\Omega_\pm(\kappa)$ is an eigenvalue of $\cancel{\mathfrak{D}}(-\kappa)$. Finally, since $\Omega_\pm(\kappa)$ is simple for $|\kappa|$ small, it follows that ${\Omega_\pm(-\kappa) = -\Omega_\pm(\kappa)}$. \qed

\subsection{Dirac operators with linear eigenvalue curves; Proof of Proposition \ref{prop:dfm-eig-lin}}
\label{supp:dfm-eig-lin}

\noindent {\it } We seek eigenpairs of the form ${(\Omega(\kappa), \psi(\kappa)) = (\kappa \Omega^{(1)}, \psi^\star_\pm)}$. Substitution into  \eqref{eq:dfm-edge-eff_3} yields:
\begin{equation}
(\cancel{\mathfrak{D}}^\pm(0) + (b \cdot \sigma) \kappa) \psi^\star_\pm = (b \cdot \sigma) \kappa \psi^\star_\pm = \kappa \Omega^{(1)} \psi^\star_\pm .
\end{equation}
That is, ${(b \cdot \sigma - \Omega^{(1)}) \psi_\pm^\star = 0}$. Hence, in addition to being an eigenvector of 
 \begin{equation}
A = 
\! \begin{bmatrix}
i a_0 & i a_1 + a_2 \\
{-i a_1} + a_2 & {-i a_0}
\end{bmatrix} \!
,
\end{equation}
(see Proposition \ref{prop:zero-energy}), $\psi^\star_\pm$ must also be an eigenvector of ${b \cdot \sigma}$. Hence, if $\psi^\star_+$ and $\psi^\star_-$ are also eigenvectors of ${b \cdot \sigma}$, $A$ and ${b \cdot \sigma}$ must be simultaneously diagonalizable; this occurs exactly when they commute:
\begin{equation}
\! \begin{bmatrix}
i a_0 & i a_1 + a_2 \\
{-i a_1} + a_2 & {-i a_0}
\end{bmatrix} \!
\! \begin{bmatrix}
b_0 & b_1 - i b_2 \\
b_1 + i b_2 & b_0
\end{bmatrix} \!
=
\! \begin{bmatrix}
b_0 & b_1 - i b_2 \\
b_1 + i b_2 & b_0
\end{bmatrix} \!
\! \begin{bmatrix}
i a_0 & i a_1 + a_2 \\
{-i a_1} + a_2 & {-i a_0}
\end{bmatrix} \!
.
\end{equation}
This yields the system of three equations:
\begin{equation}
a_1 b_1 - a_2 b_2 = 0 \quad \text{and} \quad a_0 (b_1 \pm i b_2) = 0 .
\end{equation}
By Remark \ref{rmk:gamma}, $b_1$ and $b_2$ cannot be simultaneously zero. Hence, separating real and imaginary parts, the two latter equations are satisfied if and only if ${a_0 = 0}$. The first equation, together with ${a_0 = 0}$, are exactly the constraints \eqref{eq:dfm-eig-lin}. \qed

\raggedright
\bibliographystyle{plain}
\bibliography{bibliography.bib}

\end{document}